\documentclass[11pt, a4paper]{article}
\usepackage{url}
\usepackage{a4wide}
\usepackage{graphicx}
\usepackage{natbib}
\usepackage[shortlabels]{enumitem}
\usepackage{hyperref}
\usepackage{setspace}
\usepackage{multirow, booktabs}
\usepackage{amsmath, amsthm, amssymb, amsfonts}
\usepackage{bm, bbm}
\usepackage{longtable}
\usepackage{color}
\usepackage{mathtools}
\usepackage{float}
\usepackage{appendix}

\usepackage[ruled,vlined]{algorithm2e}

\allowdisplaybreaks

\DeclareMathAlphabet\mathbfcal{OMS}{cmsy}{b}{n}

\newcommand{\cblue}{\color{black}}
\newcommand{\mbf}{\mathbf}
\newcommand{\mc}{\mathcal}
\DeclareMathOperator*{\argmin}{arg\,min}
\DeclareMathOperator*{\argmax}{arg\,max}

\newcommand{\bbm}{\begin{bmatrix}}
\newcommand{\ebm}{\end{bmatrix}}

\newcommand{\vep}{\varepsilon}
\newcommand{\tcr}[1]{\textcolor{black}{#1}}

\renewcommand{\l}{\left}
\renewcommand{\r}{\right}
\newcommand{\tcb}[1]{\textcolor{black}{#1}}

\def\wh{\widehat}
\def\wt{\widetilde}

\makeatletter
\newcommand*{\customenum}[1]{%
  \expandafter\@customenum\csname c@#1\endcsname%
}

\newcommand*{\@customenum}[1]{%
  $\ifcase#1\or\text{(B5)}\or\text{(C1)}\or\text{(C3)}\or\text{(D3)}%
    \else\@ctrerr\fi$%
}
\AddEnumerateCounter{\customenum}{\@customenum}{\text{(C3)}}
\makeatother

\newcommand{\Var}[0]{\mathsf{Var}}
\newcommand{\p}{\mathsf{P}}

\newcommand{\cosux}{\cos( \langle u, X^{(j)}_1 - \wt X^{(j)}_{1} \rangle )}
\newcommand{\cosvxj}{\cos( \langle v, X^{(j)}_{1+\ell} - \wt X^{(j)}_{1+\ell} \rangle)}
\newcommand{\cosuy}{\cos( \langle u, X^{(j-1)}_1 - \wt X^{(j-1)}_{1} \rangle)}
\newcommand{\cosvyj}{\cos( \langle v, X^{(j-1)}_{1+\ell} - \wt X^{(j-1)}_{1+\ell} \rangle)}
\newcommand{\cosuxy}{\cos( \langle u, \wt X^{(j)}_1 - X^{(j-1)}_{1} \rangle )}
\newcommand{\cosvxyj}{\cos( \langle v, \wt X^{(j)}_{1+\ell} -  X^{(j-1)}_{1+\ell} \rangle )}

\newcommand{\hst}{\wt h_{s,t}}
\newcommand{\thetajl}{\theta_{j}}

\newcommand{\tstatkl}{T_\ell(G,k)}

\newcommand{\phixj}{\phi^{(j)}_\ell (u,v) }
\newcommand{\phiyj}{\phi^{(j-1)}_\ell (u,v) }

\theoremstyle{definition}
\newtheorem{theorem}{Theorem}
\newtheorem{lemma}{Lemma}

\newtheorem{assumption}{Assumption}
\newtheorem{example}{Example}
\theoremstyle{remark}
\newtheorem{remark}{Remark}

\title{Nonparametric data segmentation in multivariate time series via joint characteristic functions}
\author{Euan T. McGonigle$^1$ and Haeran Cho$^2$}

\begin{document}

\maketitle

\footnotetext[1]{School of Mathematical Sciences, University of Southampton.
Email: \url{e.t.mcgonigle@soton.ac.uk}.}

\footnotetext[2]{Institute for Statistical Science, School of Mathematics, University of Bristol.
Email: \url{haeran.cho@bristol.ac.uk}. Supported by Leverhulme Trust Research Project Grant RPG-2019-390.} 

\begin{abstract}

Modern time series data often exhibit complex dependence and structural changes which are not easily characterised by shifts in the mean or model parameters.
We propose a nonparametric data segmentation methodology for multivariate time series termed NP-MOJO. 
By considering joint characteristic functions between the time series and its lagged values, NP-MOJO is able to detect change points in the marginal distribution, but also those in possibly non-linear serial dependence, all without the need to pre-specify the type of changes.
We show the theoretical consistency of NP-MOJO in estimating the total number and the locations of the change points, and demonstrate the good performance of NP-MOJO against a variety of change point scenarios. 
We further demonstrate its usefulness in applications to seismology and economic time series. 
\end{abstract}


\section{Introduction}
Change point analysis has been an active area of research for decades, dating back to \cite{page1954continuous}. 
Literature on change point detection continues to expand rapidly due to its prominence in numerous applications, including biology \citep{jewell2020fast}, financial analysis \citep{lavielle2007adaptive} and environmental sciences \citep{carr2017exceptional}. 
Considerable efforts have been made for developing computationally and statistically efficient methods for data segmentation, a.k.a.\ multiple change point detection, in the mean of univariate data under independence \citep{killick2012, frick2014, fryzlewicz2014wild} and permitting serial dependence \citep{tecuapetla2017, dette2020multiscale, cho2020two, cho2021multiple}.
There also exist methods for detecting changes in the covariance \citep{aue2009, wang2021optimal}, parameters under linear regression \citep{bai1998estimating, xu2024change} or other models \citep{psr2014, safikhani2022joint} in fixed and high dimensions.
For an overview, see \cite{truong2020selective} and \cite{cho2021data}.

Any departure from distributional assumptions such as independence and Gaussianity tends to result in poor performance of change point algorithms. 
Furthermore, it may not be realistic to assume any knowledge of the type of change point that occurs, or to make parametric assumptions on the data generating process, for time series that possess complex structures and are observed over a long period.
Searching for change points in one property of the data (e.g.\ mean), when the time series instead undergoes changes in another (e.g.\ variance), may lead to misleading conclusions and inference on such data. 
Therefore, it is desirable to develop flexible, nonparametric change point detection techniques that are applicable to detect general changes in the underlying distribution of serially dependent data. 

There are several strategies for the nonparametric change point detection problem, such as those based on the empirical cumulative distribution and density functions \citep{carlstein1988, zou2014nonparametric, haynes2017computationally, padilla2021optimal, jula2022multiscale, padilla2022optimal2, misael2023change},
kernel transforms of the data \citep{harchaoui2009, celisse2018new, arlot2019kernel, li2019scan} or $U$-statistics measuring the `energy'-based distance between different distributions \citep{matteson2014, chakraborty2021high, boniece2022change}. 
There also exist graph-based methods applicable to non-Euclidean data \citep{chen2015graph, chu2019}.
All these methods can only detect changes in the marginal distribution of the data and apart from \cite{misael2023change}, assume serial independence.
We also mention \cite{cho2012a}, \cite{preuss2015detection} and \cite{korkas2017multiple} where the problem of detecting changes in the second-order structure is addressed, but their methods do not have power against changes in non-linear dependence.

We propose {NP-MOJO}, a NonParametric MOving sum procedure for detecting changes in the JOint characteristic function, which detects multiple changes in serial, possibly non-linear dependence as well as marginal distributions of a multivariate time series $\{ X_t \}_{t=1}^n$.
We adopt a moving sum procedure to scan the data for multiple change points.
The moving sum methodology has successfully been applied to a variety of change point testing \citep{chu1995mosum, huvskova2001permutation} and data segmentation problems \citep{eichinger2018}. 
Here, we combine it with a detector statistic carefully designed to detect changes in complex dependence structure beyond those detectable from considering the marginal distribution only.
Specifically, we utilise an energy-based distributional discrepancy that measures any change in the joint characteristic function of the time series at some lag $\ell \ge 0$, which allows for detecting changes in the joint distribution of $(X_t, X_{t + \ell})$ beyond the changes in their linear dependence.
To the best of our knowledge, NP-MOJO is the first nonparametric methodology which is able to detect changes in non-linear serial dependence in multivariate time series.

We establish that NP-MOJO achieves consistency in estimating the number and locations of the change points for a given lag, \tcb{providing convergence rates for the change point location estimators,} and propose a methodology that extends this desirable property of single-lag NP-MOJO to multiple lags. 
Combined with a dependent multiplier bootstrapping procedure, NP-MOJO and its multi-lag extension perform well across a wide range of change point scenarios in simulations and real data applications. \tcb{Accompanying R software implementing NP-MOJO is available as the R package} \verb!CptNonPar! \tcb{\citep{CptNonPar2023} on CRAN.}

\section{Model and measure of discrepancy}
\label{sec:model}

We observe a multivariate time series $\{ X_t \}_{t=1}^n$ of (finite) dimension $p$, where
\begin{align}\label{model-eq}
X_t = \sum_{j=0}^{q} X^{(j)}_t \cdot \mathbb{I} \{ \theta_{j}+1 \leq t \leq \theta_{j+1} \}
\end{align}
with $X_t = ( X_{t1}, \ldots , X_{tp})^\top$ and \tcb{$0 = \theta_0 < \theta_1 < \cdots < \theta_q < \theta_{q + 1} = n$}.
For each sequence $\{ X_t^{(j)} : \ t \geq 1\}, \, j = 0, \ldots, q$, there exists an $\mathbb{R}^p$-valued measurable function $g^{(j)}(\cdot) = ( g^{(j)}_1 (\cdot) , \ldots , g_p^{(j)} (\cdot) )^\top$ such that $X_t^{(j)} = g^{(j)}(\mc F_t)$ with $\mc F_t = \sigma( \vep_s : s \leq t )$, and i.i.d.\ random elements $\vep_t$. 
We assume that $g^{(j-1)} \neq g^{(j)}$ for all $j = 1, \ldots, q$, such that under the model~\eqref{model-eq}, the time series undergoes $q$ change points at locations $\Theta = \{ \theta_1, \ldots , \theta_{q} \}$, with the notational convention that $\theta_0=0$ and $\theta_{q+1} = n$. 
That is, $\{ X_t \}_{t=1}^n$ consists of $q+1$ stationary segments 
where the $j$-th segment is represented in terms of a segment-dependent `output' $g^{(j)} (\mc F_t)$, with the common `input' $\mc F_t$ shared across segments such that dependence across the segments is not ruled out.
Each segment has a non-linear Wold representation as defined by \cite{wu2005nonlinear}; this representation includes commonly adopted time series models including ARMA and GARCH processes.

Denote the inner product of two vectors $x$ and $y$ by $\langle x, y \rangle = x^\top y$ and $\imath$ the imaginary unit with $\imath^2 = -1$.
At some integer $\ell$, define the joint characteristic function of $\{X^{(j)}_t\}_{t \in \mathbb{Z}}$ at lag $\ell$, as
\begin{align*}
\phi^{(j)}_\ell (u, v) = \mathbb{E} \l\{ \exp \l( \imath \langle u, X^{(j)}_1 \rangle + \imath \langle v, X^{(j)}_{1+\ell} \rangle \r) \r\}, \quad 0 \leq j \leq q .
\end{align*}
We propose to measure the size of changes between adjacent segments under~\eqref{model-eq}, using an `energy-based' distributional discrepancy given by
\begin{align}\label{eq:test-stat-char}
d_\ell^{(j)} = \int_{\mathbb{R}^{p}} \int_{\mathbb{R}^{p}} \l\vert \phi^{(j)}_\ell (u,v) - \phi_{\ell}^{(j - 1)} (u,v) \r\vert^2 w(u,v) du dv, \quad 1 \leq j \leq q ,
\end{align}
where $w(u,v)$ is a positive weight function for which the above integral exists. 
For given lag $\ell \ge 0$, the quantity $d_\ell^{(j)}$ measures the weighted $L_2$-norm of the distance between the lag $\ell$ joint characteristic functions of $\{X^{(j - 1)}_t\}_{t \in \mathbb{Z}}$ and $\{X^{(j)}_t\}_{t \in \mathbb{Z}}$. A discrepancy measure of this form is a natural choice for nonparametric data segmentation, since:
\begin{lemma}
\label{lemma:discrepancy} 
\tcb{For any $\ell \geq 0$, $d^{(j)}_{\ell} = 0$ if and only if $(X_1^{(j)}, X^{(j)}_{1+\ell}) \stackrel{d}{=} (X_1^{(j-1)} , X_{1+\ell}^{(j-1)})$. }
\end{lemma}
Lemma~\ref{lemma:discrepancy} extends the observation made in \cite{matteson2014} about the correspondence between the characteristic function and marginal distribution. It shows that by considering the joint characteristic functions $\phi^{(j)}_\ell (u, v)$ at multiple lags $\ell \ge 0$, the discrepancy $d^{(j)}_{\ell}$ is able to capture changes in the serial dependence as well as those in the marginal distribution of $\{ X_t \}_{t = 1}^n$.

\tcb{The following lemma lists some choices of the weight function $w(u, v)$ and the associated representation of $d^{(j)}_\ell$ as the kernel-based discrepancy between $Y_t^{(j)} = (X_t^{(j)}, X^{(j)}_{t+\ell})$ and $Y_t^{(j - 1)}$, extending the observation made in \cite{matteson2014} for the setting where a sequence of independent observations are undergoing changes in the marginal distribution.}
Let $\Vert x \Vert$ denote the Euclidean norm of a vector $x$, and define $\tilde{Y}^{(j)}_t = ( \tilde X^{(j)}_t, \tilde X^{(j)}_{t+\ell})$ where $\tilde X^{(j)}_t = g^{(j)} (\tilde{\mc F}_t)$ with $\tilde{\mc F}_t = \sigma( \tilde{\vep}_s : s \leq t )$ and $\tilde{\vep}_t$ is an independent copy of $\vep_t$.

\begin{lemma}\label{lemma:weight-int-identities}
\begin{enumerate}[label = (\roman*)]
\item \label{lemma:weight-int1} For any $\beta>0$, suppose that $d^{(j)}_\ell$ in~\eqref{eq:test-stat-char} is obtained with respect to the following weight function:
\begin{align*} 
w_1(u,v) = C_1(\beta,p)^{-2} \exp \left \{ -\frac{1}{2 \beta^2} \left(  \Vert u \Vert^2 + \Vert v \Vert^2 \right)   \right \} \text{ \ with \ } C_1(\beta, p) =  (2 \pi)^{p/2} \beta^{p}.
\end{align*} 
Then, the function $h_1 : \mathbb{R}^{2p} \times \mathbb{R}^{2p} \to [0,1]$ defined as $h_1(x, y) = \exp(-\beta^2 \Vert x - y \Vert^2 /2) $ for $x, y \in \mathbb{R}^{2p}$, satisfies
\begin{align*} 
d^{(j)}_{\ell}  &=  \mathbb{E} \left \{ h_1 \left( Y_1^{(j)},\tilde Y_1^{(j)} \right) \right \} +  \mathbb{E} \left\{ h_1 \left(Y_1^{(j-1)}, \tilde Y_1^{(j-1)} \right) \right \} - 2 \mathbb{E} \left\{ h_1 \left(\tilde Y_1^{(j)}, Y_1^{(j-1)} \right) \right\}.
\end{align*}

\item \label{lemma:weight-int2} For any $\delta >0$, suppose that $d^{(j)}_\ell$ is obtained with
\begin{align*}
w_2 (u,v) = C_2(\delta,p)^{-2} \prod_{s=1}^p u_s^2 v_s^2 \exp \left\{ -\delta \left( u_s^2 + v_s^2 \right)   \right \} \text{ \ with \ } C_2 (\delta,p) = \frac{\pi^{p/2}}{2^p \delta^{3p/2}}.
\end{align*} 
Then, the function $h_2 : \mathbb{R}^{2p} \times \mathbb{R}^{2p} \to [-2e^{-2/3},1]$ defined as 
\begin{align*}
h_2 (x, y) = \prod_{r=1}^{2p} \frac{ \left \{ 2 \delta - (x_r - y_r )^2  \right \} \exp \left \{ -\frac{1}{4\delta} (x_r - y_r)^2 \right\} }{2 \delta}
\end{align*}
for $x = (x_1, \ldots, x_{2p})^\top$ and $y = (y_1, \ldots, y_{2p})^\top$,  satisfies
\begin{align*}
d^{(j)}_{\ell}  &=  \mathbb{E} \left\{ h_2 \left( Y_1^{(j)},\tilde Y_1^{(j)} \right) \right\} +  \mathbb{E} \left\{ h_2 \left(Y_1^{(j-1)}, \tilde Y_1^{(j-1)} \right) \right\} - 2 \mathbb{E} \left\{ h_2 \left(\tilde Y_1^{(j)}, Y_1^{(j-1)} \right) \right\}.
\end{align*}

\end{enumerate}
\end{lemma}
\tcb{Lemma \ref{lemma:weight-int-identities} is a special case of Bochner's Theorem applied to the chosen weight functions, see for example Section 5.3 of \cite{sejdinovic2013equivalence}}. The weight function $w_1$ is commonly referred to as the Gaussian weight function. 
Both $w_1$ and $w_2$ are unit integrable 
and separable in their arguments, such that $d^{(j)}_\ell$ is well-defined due to the boundedeness of the characteristic function. 
We provide an alternative weight function in Appendix~\ref{sec:weight} and also refer to \cite{fan2017multivariate} for other suitable choices.
\begin{remark}
From Lemma~\ref{lemma:weight-int-identities}, $d^{(j)}_\ell$ can be viewed as the squared maximum mean discrepancy (MMD) on a suitably defined reproducing kernel Hilbert space with the associated kernel function; see Lemma~6 of \cite{gretton2012kernel} and Section 2.6 of \cite{celisse2018new}.  
We also note the literature on the (auto)distance correlation for measuring and testing dependence in multivariate \citep{szekely2007} and time series \citep{zhou2012measuring, fokianos2017consistent, davis2018} settings. 
\end{remark}

\begin{remark}
\tcr{In Model~\eqref{model-eq} (and in our theoretical results), the dimension $p$ of the time series is assumed fixed. We would expect practical performance to deteriorate with increasing dimension since we use an energy-based method. For example, when the time series undergoes a change in both mean and variance, the pre- and post-change segments of the time series can be separated into an ``inner layer" and ``outer layer" based on their pairwise Euclidean distances. However, as \cite{chen2017new} note, ``data points in the outer layer find themselves to be closer to points in the inner layer than other points in the outer layer", due to the curse of dimensionality. See for example \cite{ramdas2015decreasing} or Section 2.2 of \cite{chu2019} for further discussion.}
\end{remark}

\section{Methodology}\label{sec:method}

\subsection{The NP-MOJO procedure} \label{sec:single}

In this section we describe our proposed nonparametric moving sum procedure for detecting changes in the joint characteristic function, henceforth referred to as NP-MOJO. The identities given in Lemma~\ref{lemma:weight-int-identities} allow for the efficient computation of the statistics approximating $d^{(j)}_\ell$ and their weighted sums, which forms the basis for 
the NP-MOJO procedure for detecting multiple change points from a multivariate time series $\{X_t\}_{t = 1}^n$ under the model~\eqref{model-eq}.
Throughout, we present the procedure with a generic kernel $h$ associated with some weight function $w$.
We first introduce NP-MOJO for the problem of detecting changes in the joint distribution of $Y_t = (X_t, X_{t + \ell})$ at a given lag $\ell \ge 0$, and extend it to the multi-lag problem in Section~\ref{sec:multi}.

For fixed bandwidth $G \in \mathbb{N}$, NP-MOJO scans the data using a detector statistic computed on neighbouring moving windows of length $G$, which approximates the discrepancy between the local joint characteristic functions of the corresponding windows measured analogously as in~\eqref{eq:test-stat-char}.
Specifically, the detector statistic at location $k$ is given by the following two-sample $V$-statistic:
\begin{align}
\tstatkl  = \frac{1}{(G-\ell)^2} \left \{ \sum_{s,t = k - G+1}^{k - \ell} h(Y_s, Y_t) + \sum_{s,t = k+1}^{k +G - \ell} h(Y_s, Y_t) -2 \sum_{s = k - G+1}^{k - \ell} \sum_{t = k +1}^{k +G - \ell} h(Y_s, Y_t)  \right \} \nonumber
\end{align}
for $k = G, \ldots , n-G$, as an estimator of the local discrepancy measure
\begin{align}
\label{eq:pop-test-stat}
\mc D_\ell (G,k) = \sum_{j=0}^q \left( \frac{G - \ell -  |k - \theta_{j}\vert}{G-\ell} \right)^2 d^{(j)}_\ell \cdot \mathbb{I}{ \{ \vert k - \theta_{j} \vert \le G - \ell \}}. 
\end{align}
\tcb{At given $k$, the statistic $\tstatkl$ measures the difference in the distribution of $Y_t$ over the disjoint intervals of length $G - \ell$ around $k$, and satisfies}
\begin{equation} \label{eq:t-stat-exp}
\tcb{\mathbb{E} \{ T_\ell (G,k) \} = \mathcal{D}_\ell (G,k) + \mathcal{O}(G^{-1/2} ).   }
\end{equation}
We have $\mc D_\ell(G, k) = 0$ when the section of the data $\{X_t, \, \vert t - k \vert \le G - \ell \}$ does not undergo a change and accordingly, $T_\ell(G,k)$ is expected to be close to zero.
On the other hand, if $\vert k - \theta_j \vert < G - \ell$, then $\mc D_\ell(G, k)$ increases and then decreases around $\theta_j$ with a local maximum at $k = \theta_j$. \tcb{The statistic $T_\ell(G, k)$ is expected to behave similarly: in particular, at any change point location $\theta_j$, we have that $\mathbb{E} \{ T_\ell (G, \theta_j) \}  = d^{(j)} _\ell + \mathcal{O}(G^{-1/2} )$ (see Lemma~D.4 in the supplementary material for further information).}
We illustrate this using the following example.

\begin{figure}[]
\centering
\includegraphics[width = 0.8\textwidth]{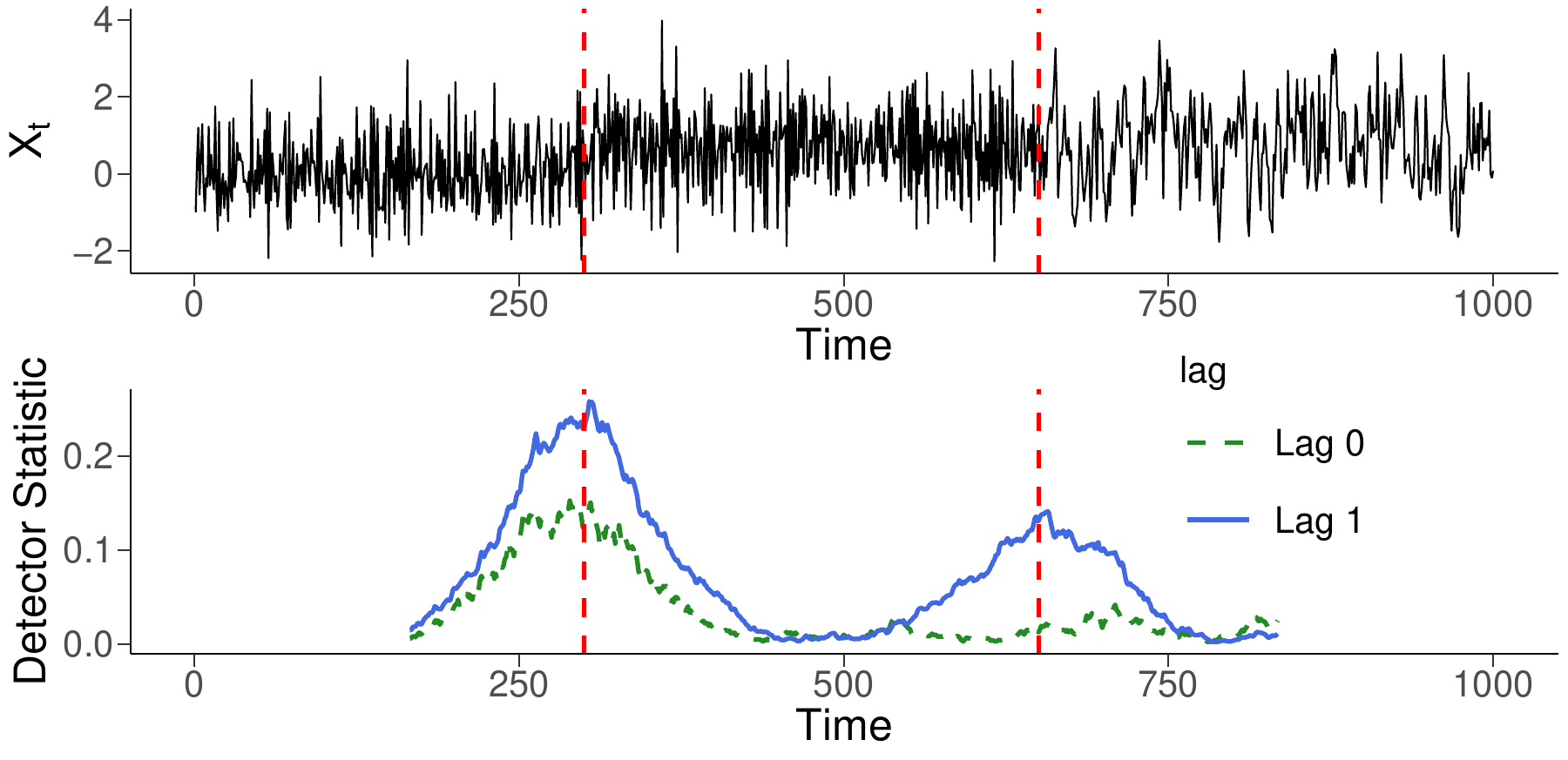}
\caption{Top: time series of length $n = 1000$ with change points $\theta_1 = 300$ and $\theta_2 = 650$ (vertical dashed lines), see Example~\ref{ex:time-series}. Bottom: corresponding detector statistics $\tstatkl$ computed at lags $\ell = 0$ (dashed) and $\ell = 1$ (solid).}
\label{mosum-example}
\end{figure}

\begin{example}\label{ex:time-series}
A univariate time series $\{ X_t\}_{t=1}^{n}$ of length $n = 1000$ is generated as $X_t = \mu_t + \vep_t$, where $\mu_t = 0.7 \cdot \mathbb{I}{\{ t > \theta_1\} }$ and $\vep_t = \vep^{(1)}_t \cdot \mathbb{I}{\{ t < \theta_2 \} }  + \vep^{(2)}_t \cdot \mathbb{I}{\{ t \geq \theta_2 \} }$, with $\theta_1 = 300$ and $\theta_2 = 650$.
Each $\vep^{(j)}_t$ is an autoregressive (AR) process of order 1: $\vep_t^{(1)} = 0.5 \vep_{t-1}^{(1)} + W_{t}$ and $\vep_t^{(2)} = -0.5 \vep_{t-1}^{(2)} + W_{t}$, where $\{W_t\}_{t \in \mathbb{Z}}$ is a white noise process with $\Var(W_t) = \sqrt{1-0.5^2}$. This choice leads to $\Var(X_t) = 1$ for all $t$, see the top panel of Figure~\ref{mosum-example} for a realisation.
Then, the mean shift at $\theta_1$ is detectable at all lags while the autocorrelation change at $\theta_2$ is detectable at odd lags only, i.e.\ $d^{(2)}_\ell = 0$ for even $\ell \ge 0$.
The bottom panel of Figure~\ref{mosum-example} plots $T_\ell(G, k), \, G \le k \le n - G$, computed using kernel $h_{2}$ in Lemma~\ref{lemma:weight-int-identities}~\ref{lemma:weight-int2} with $G = 166$.
At lag $\ell = 0$, the detector statistic forms a prominent peak around $\theta_1$ but it is flat around $\theta_2$; at lag $\ell = 1$, the statistic $T_1(G, k)$ forms local maxima around both $\theta_j, \, j = 1, 2$. 
\end{example}

Based on these observations, it is reasonable to detect and locate the change points in the joint distribution of $(X_t, X_{t + \ell})$ as significant local maximisers of $T_\ell(G, k)$.
We adopt the selection criterion, first considered by \cite{eichinger2018} in the context of detecting mean shifts from univariate time series, for simultaneous estimation of multiple change points.
For some fixed constant $\eta \in (0, 1)$ and a threshold $\zeta_\ell (n, G) > 0$, we identify any local maximiser of $T_\ell(G, k)$, say $\wh\theta$, which satisfies
\begin{align}\label{eq:mosum:est} 
T_\ell( G,\wh \theta ) > \zeta_\ell (n, G) \quad \text{and} \quad
\wh \theta = {\argmax}_{k: \, \vert k - \wh \theta \vert \le \eta G} \tstatkl.
\end{align}
\tcb{That is, $\wh{\theta}$ is declared a change point if it is a local maximiser of $T_\ell(G, k)$ over a sufficiently large interval of size $\eta G$, at which the threshold $\zeta_\ell (n, G)$ is exceeded.} We denote the set of such estimators fulfilling~\eqref{eq:mosum:est} by $\wh{\Theta}_\ell$ with $\wh q_\ell = \vert \wh{\Theta}_\ell \vert$.
The choice of $\zeta_\ell (n, G)$ is discussed in Section~\ref{sec:bootstrap}. 

\subsection{Theoretical properties}
\label{sec:theory}

For some finite integer $\ell \ge 0$, we define the index set of the change points {\it detectable} at lag $\ell$ as $\mc I_\ell = \{ 1 \leq j \leq q : \, d^{(j)}_\ell \neq 0 \}$, and denote its cardinality by $q_\ell = \vert \mc I_\ell \vert \le q$.
Not all change points are detectable at all lags, see Example~\ref{ex:time-series} where we have $\mc I_0 = \{ 1 \}$ and $\mc I_1 = \{ 1, 2 \}$.
In this section, we show that the single-lag NP-MOJO described in Section~\ref{sec:single} consistently estimates the total number $q_\ell$ and the locations $\{\theta_j, \, j \in \mc I_\ell\}$ of the change points detectable at lag $\ell$, by $\wh{\Theta}_\ell$. 

Writing {$g_{ti}(\cdot) = \sum_{j=0}^q g^{(j)}_i(\cdot ) \cdot \mathbb{I}{\{ \theta_{j}+1 \leq t \leq \theta_{j+1} \}}$}, define $X_{ti, \{t-s \}} = g_{ti} ( \mc F_{t,\{t-s\}} )$, where $F_{t,\{t-s\}} = \sigma ( \ldots , \vep_{t-s-1} , \tilde{\vep}_{t-s} , \vep_{t-s+1} , \ldots , \vep_t )$ is a coupled version of $\mc F_t$ with $\vep_{t-s}$ replaced by its independent copy $\tilde{\vep}_{t-s}$. 
For a random variable $Z$ and $\nu > 0$, let $\Vert Z \Vert_\nu = \{ \mathbb{E} ( |Z|^\nu) \}^{1/\nu}$.
Analogously as in \cite{xu2024change}, we define the element-wise functional dependence measure and its cumulative version as 
\begin{align}
\label{eq:func:dep}
\delta_{s, \nu, i}  = \sup_{t \in \mathbb{Z}} \Vert X_{ti} - X_{ti,\{t-s\}}  \Vert_\nu  \text{ \ and \ } \Delta_{m, \nu} = \max_{1 \leq i \leq p} \sum_{s=m}^{\infty} \delta_{s, \nu, i}  , \  m \in \mathbb{Z}.
\end{align}
Then, we make the following assumptions on the degree of serial dependence in $\{X_t \}_{t=1}^n$.
\begin{assumption}\label{assum:functional-dep}
There exist some constants $C_F, C_X \in (0, \infty)$ and $\gamma_1 \in (0, 2)$ such that \begin{align*}
\sup_{m \geq 0} \exp (C_F m^{\gamma_1}) \Delta_{m,2}  \leq C_X .
\end{align*}
\end{assumption}

\begin{assumption}\label{assum:beta-mixing}
The time series $\{X_t \}_{t=1}^n$ is continuous and $\beta$-mixing with $\beta (m) \leq C_\beta m^{-\gamma_2}$ for some constants $C_\beta \in (0, \infty)$ and $\gamma_2 \ge 1$, where 
\begin{align*}
\beta (m) = \sup_{t \in \mathbb{Z}} \left( \sup \frac{1}{2} \sum_{r=1}^R \sum_{s=1}^S \left\vert \p (A_r \cap B_s ) - \p (A_r ) \p (B_s) \right\vert \right).
\end{align*}
Here, the inner supremum is taken over all pairs of finite partitions $\{A_1, \ldots, A_R\}$ of $\mc F_t = \sigma(\vep_u, \, u \le t)$ and $\{B_1, \ldots, B_S\}$ of $\sigma(\vep_u, \, u \ge t + m)$.
\end{assumption}

Assumptions~\ref{assum:functional-dep} and~\ref{assum:beta-mixing} require the serial dependence in $\{X_t\}_{t = 1}^n$, measured by $\Delta_{m, 2}$ and $\beta(m)$, to decay exponentially, and both are met by a range of linear and non-linear processes \citep{wu2005nonlinear, mokkadem1988mixing}.
Under Assumption~\ref{assum:functional-dep}, we have $\Vert X_{it} \Vert_2 < \infty$ for all $i$ and $t$. 
Assumption~\ref{assum:functional-dep} is required for bounding $T_{\ell}(G, k) - \mathbb{E} \{T_{\ell}(G, k) \}$ uniformly over $k$, while Assumption~\ref{assum:beta-mixing} is used for controlling the bias $\mathbb{E} \{T_{\ell}(G, k) \} - \mc D_\ell(G, k)$ which is attributed to serial dependence. 
A condition similar to Assumption~\ref{assum:beta-mixing} is often found in the time series literature making use of distance correlations, see e.g.\ \cite{davis2018} and \cite{yousuf2022targeting}.  
\tcb{Under the stronger assumption that $\{X^{(j)}_t\}$ and $\{X^{(j + 1)}_t\}$ are independent, we can derive the analogous results as those presented in Theorems~\ref{thm:consistency} and~\ref{thm:multilag-consistent}, under Assumption~\ref{assum:beta-mixing} only.}

\begin{assumption}\label{assum:kernel}
The kernel function $h$ is symmetric and bounded, and can be written as $h(x, y) = h_0 (x-y)$ for some  function $h_0: \mathbb{R}^{2p} \to \mathbb{R}$ that is Lipschitz continuous with respect to $\Vert \cdot \Vert$ with Lipschitz constant $C_h \in (0, \infty)$.
\end{assumption}
Assumption~\ref{assum:kernel} on the kernel function $h$ is met by $h_1$ and $h_2$ introduced in Lemma~\ref{lemma:weight-int-identities}, with constants $C_h$ bounded by $\beta e^{-1/2}$ and $2 \sqrt{2} p^{3/2} \delta^{-1/2}$, respectively.

\begin{assumption}\label{assum:change-point}
\begin{enumerate}[label = (\roman*)]
\item \label{assum:min-space} 
$G^{-1} \log(n) \to 0$ as $n \to \infty$ while $\min_{0 \leq j \leq q} ( \theta_{j+1} - \theta_{j} ) \geq 2G$.
\item \label{assum:change-size} ${ \sqrt{G/ \log (n)} } \min_{j \in \mc I_\ell} d^{(j)}_\ell \to \infty$.
\end{enumerate}
\end{assumption}

Recall that $\mc I_\ell$ denotes the index set of detectable change points at lag $\ell$, i.e.\ $d^{(j)}_\ell > 0$ iff $j \in \mc I_\ell$.
However, this definition of detectability is too weak to ensure that all $\theta_j, \, j \in \mc I_\ell$, are detected by NP-MOJO with high probability at lag $\ell$, since we do not rule out the case of local changes where $d^{(j)}_\ell~\to~0$.
Consider Example~\ref{ex:time-series}: the change in the autocorrelations results in $d^{(2)}_\ell > 0$ for all odd $\ell$ but the size of change is expected to decay exponentially fast as $\ell$ increases.
Assumption~\ref{assum:change-point} allows for local changes provided that $\sqrt{G/ \log (n)} d^{(j)}_\ell$ diverges sufficiently fast. 
Assumption~\ref{assum:change-point}~\ref{assum:min-space} on the minimum spacing of change points, is commonly imposed in the literature on change point detection using moving window-based procedures.
Assumption~\ref{assum:change-point} does not rule out $G / n \to 0$ and permits the number of change points $q$ to increase in $n$.
We discuss the selection of bandwidth in Section~\ref{sec:tuning}.

\begin{theorem}
\label{thm:consistency}
Let Assumptions~\ref{assum:functional-dep}, \ref{assum:beta-mixing}, \ref{assum:kernel} and~\ref{assum:change-point} hold and $\ell \ge 0$ be a finite integer, and set the threshold as $\zeta_\ell (n,G) = c_\zeta \sqrt{\log(n)/G}$ for some constant $c_\zeta > 0$. Then, there exists $c_0 > 0$, depending only on $C_F$, $C_X$, $\gamma_1$, $C_\beta$, $\gamma_2$ \tcb{and $p$}, such that as $n \to \infty$,
\begin{align*}
\p \left(\wh{q}_\ell = q_\ell , \, \max_{j \in \mc I_\ell} \min_{\wh\theta \in \wh{\Theta}_\ell} \, d^{(j)}_\ell \vert \wh{\theta} -\theta_{j} \vert \leq c_0 \sqrt{G \log (n)}  \right) \to 1.    
\end{align*}
\end{theorem}

Theorem~\ref{thm:consistency} establishes that, for given $\ell$, NP-MOJO correctly estimates the total number and the locations of the change points detectable at lag $\ell$ \tcb{(including the no change case where $q_\ell = 0$)}. 
In particular, by Assumption~\ref{assum:change-point}, the change point estimators satisfy
\begin{align*}
\min_{\wh\theta \in \wh{\Theta}_\ell} \, \vert \wh{\theta} -\theta_{j} \vert
= O_P \left\{ (d^{(j)}_\ell)^{-1} \sqrt{G \log (n)} \r\} = o_P \{ \min( \theta_{j} - \theta_{j-1} , \theta_{j+1} - \theta_j) \} \text{ \ for all \ } j \in \mc I_\ell,
\end{align*}
i.e.\ the change point estimators converge to the true change point locations in the rescaled time.
Further, the rate of estimation is inversely proportional to the size of change $d^{(j)}_\ell$, such that the change points associated with larger $d^{(j)}_\ell$ are estimated with better accuracy.
Also making use of the energy-based distributional discrepancy, \cite{matteson2014} establish the consistency of their proposed E-Divisive method for detecting changes in (marginal) distribution under independence.
In addition to detection consistency, we further derive the rate of estimation for NP-MOJO which is applicable to detect changes in complex time series dependence besides those in marginal distribution, in broader situations permitting serial dependence.

\tcr{Compared to the optimal rate of estimation known for some parametric change point problems, the rate reported in Theorem~\ref{thm:consistency} is sub-optimal due to the bias of order $O(G^{-1/2})$ (see~\eqref{eq:t-stat-exp}) in $U$- and $V$-statistics in the presence of serial dependence. 
In the next theorem, we relax Assumptions~\ref{assum:functional-dep} and~\ref{assum:beta-mixing} to serial independence, and derive a faster rate of estimation for detecting change points in the marginal distribution (namely $\theta_j, \, j \in \mc I_0 = \{1, \ldots, q_0\}$) using NP-MOJO with lag $\ell = 0$.}
\tcr{\begin{theorem}
\label{thm:consistency-marginal}
Let Assumptions~\ref{assum:kernel} and~\ref{assum:change-point} hold, the latter with $\ell = 0$, and assume that $\{ X_t \}_{t=1}^n$ are independent over time, so that $q_0 = q$.
Set the threshold as $\zeta(n, G) = c_\zeta \sqrt{\log(n)/G}$ for some constant $c_\zeta > 0$. 
Then, there exists $c_0 > 0$ depending on $p$, such that as $n \to \infty$,
\begin{align*}
\p \left( \hat{q} = q,  { \max_{1 \le j \le q} \min_{\wh\theta \in \wh{\Theta}_0} } (d^{(j)}_0)^2 \vert \wh{\theta} -\theta_{j} \vert \leq c_0 \log (n)  \right) \to 1.    
\end{align*}
\end{theorem}}

\subsection{Multi-lag extension of NP-MOJO}
\label{sec:multi}

In this section, we address the problem of combining the results of the NP-MOJO procedure when it is applied with multiple lags.
Let $\mc L \subset \mathbb{N}_0 = \{ 0, 1, \ldots \}$ denote a (finite) set of non-negative integers.
Recall that given $\ell \in \mc L$, NP-MOJO returns a set of change points estimators $\wh{\Theta}_\ell$. 
Denote the union of change point estimators over all lags $\mc L$ by $\widetilde{\Theta} = \bigcup_{\ell \in \mc L} \wh{\Theta}_\ell = \{\wt\theta_j, \, 1 \le j \le Q: \, \wt\theta_1 < \ldots, < \wt\theta_Q \}$, and denote by $\mathbb{T}(\wt\theta) = \max_{\ell \in \mc L} T_\ell(G, \wt\theta)$ the maximum detector statistic at $\wt\theta$ across all $\ell \in \mc L$. 
We propose to find a set of the final change point estimators $\wh{\Theta} \subset \wt{\Theta}$ by taking the following steps; we refer to this procedure as multi-lag NP-MOJO.

\begin{enumerate}[label = {Step~\arabic*.}]
\setcounter{enumi}{-1}
\item Set $\wh{\Theta} = \emptyset$ and select a constant $c \in (0, 2]$.

\item Set $\wt{\Theta}_1 = \wt{\Theta}$ and $m = 1$.
Iterate Steps~2--4 for $m = 1, 2, \ldots$, while $\wt{\Theta}_{m } \ne \emptyset$.

\item Let $\wt{\theta}_m = \min \, \widetilde{\Theta}_{m}$ and identify $\mc C_m= \{ \wt{\theta} \in \wt{\Theta}_{m} : \, \wt{\theta} - \wt{\theta}_m  < cG \}$.

\item Identify
$\wh\theta_m = \argmax_{\wt\theta \in \mc C_m} \mathbb{T}(\wt\theta)$; if there is a tie, we arbitrarily break it. 

\item Add $\wh{\theta}_m$ to $\wh{\Theta}$ and update $m \leftarrow m + 1$ and $\widetilde{\Theta}_{m} = \widetilde{\Theta}_{m-1} \setminus \mc C_{m-1}$.
\end{enumerate}

At iteration $m$ of the multi-lag NP-MOJO, Step~2 identifies the minimal element from the current set of candidate change point estimators $\wt{\Theta}_m$, and a cluster of estimators $\mc C_m$ whose elements are expected to detect the identical change points from multiple lags.
Then, Step~3 finds an estimator $\wh\theta \in \mc C_m$, which is associated with the largest detector statistic at some lag, and it is added to the set of final estimators. 
This choice is motivated by Theorem~\ref{thm:consistency}, which shows each $\theta_j$ is estimated with better accuracy at the lag associated with the largest change in the lagged dependence (measured by $d^{(j)}_\ell$).
Iterating these steps until all the elements of $\wt\Theta$ are either added to $\wh{\Theta}$ or discarded, we obtain the set of final change point estimators.

We define a subset of $\mc L$ containing the lags at which the $j$-th change point is detectable, as $\mc L^{(j)}  = \{ \ell \in \mc L : \, d^{(j)}_\ell \neq 0 \}$.
Re-visiting Example~\ref{ex:time-series}, when we set $\mc L = \{ 0, 1\}$, it follows that $\mc L^{(1)} = \{ 0, 1 \}$ and $\mc L^{(2)} = \{ 1 \}$.
To establish the consistency of the multi-lag NP-MOJO, we formally assume that all changes points are detectable at some lag $\ell \in \mc L$.
\begin{assumption}\label{assum:lag-assumption}
For $\mc L \subset \mathbb{N}_0$ with $L = \vert \mc L \vert < \infty$, we have
$\cup_{\ell \in \mc L} \mc I_\ell = \{1, \ldots, q\}$.
Equivalently, $\mc L^{(j)} \neq \emptyset$ for all $j = 1, \ldots, q$.
\end{assumption}

Under Assumptions~\ref{assum:functional-dep}--\ref{assum:lag-assumption}, the consistency of the multi-lag NP-MOJO procedure is largely a consequence of Theorem~\ref{thm:consistency}.
Assumption~\ref{assum:change-point}~\ref{assum:change-size} requires that at any lag $\ell \in \mc L$ and a given change point $\theta_j$, we have either $j \in \mc I_\ell$ with $d^{(j)}_\ell$ large enough (in the sense that $\sqrt{G/ \log (n)} d^{(j)}_\ell \to \infty$), or $j \notin \mc I_\ell$ such that $d^{(j)}_\ell = 0$.
Such a dyadic classification of the change points rules out the possibility that for some $j$, we have $d^{(j)}_\ell > 0$ but $d^{(j)}_\ell = O\{\sqrt{\log (n) / G} \}$, in which case $\theta_j$ may escape detection by NP-MOJO at lag $\ell$.
We therefore consider the following alternative:
\begin{assumption}
\label{assum:change-point2}
\begin{enumerate}[label = (\roman*)]
\item \label{assum:min-space2} 
$G^{-1}\log(n) \to 0$ as $n \to \infty$ while $\min_{0 \leq j \leq q} ( \theta_{j+1} - \theta_{j} ) \geq 4G$.
\item \label{assum:change-size2} $\sqrt{G/ \log (n)} \min_{1 \le j \le q} \max_{\ell \in \mc L^{(j)}} d^{(j)}_\ell  \to \infty$.
\end{enumerate}
\end{assumption}
Compared to Assumption~\ref{assum:change-point}, Assumption~\ref{assum:change-point2} requires that the change points are further apart from one another relative to $G$ by the multiplicative factor of two.
At the same time, the latter only requires that for each $j = 1, \ldots, q$, there exists {\it at least one} lag $\ell \in \mc L$ at which $d^{(j)}_\ell$ is large enough to guarantee the detection of $\theta_j$ by NP-MOJO with large probability. 
Theorem~\ref{thm:multilag-consistent} establishes the consistency of multi-lag NP-MOJO under either Assumption~\ref{assum:change-point} or~\ref{assum:change-point2}.

\begin{theorem}\label{thm:multilag-consistent}
Suppose that Assumptions~\ref{assum:functional-dep}, \ref{assum:beta-mixing}, \ref{assum:kernel} and~\ref{assum:lag-assumption} hold and at each $\ell \in \mc L$, we set $\zeta_\ell(n, G) = c_{\zeta, \ell} \sqrt{\log(n)/G}$ with some constants $c_{\zeta, \ell} > 0$.
Let $\wh{\Theta} = \{\wh\theta_j, \, 1 \le j \le \wh q: \, \wh\theta_1 < \ldots < \wh\theta_{\wh q}\}$ denote the set of estimators returned by multi-lag NP-MOJO with tuning parameter $c$.
\begin{enumerate}[label = (\roman*)]
\item \label{thm:multilag1} If Assumption~\ref{assum:change-point} holds for all $\ell \in \mc L$ and $c = 2\eta$ with $\eta \in (0, 1/2]$, then with $c_0$ as in Theorem~\ref{thm:consistency}, \tcr{depending only on $C_F$, $C_X$, $\gamma_1$, $C_\beta$, $\gamma_2$ and $p$},
\begin{align*}
\p \left(\wh{q} = q , \, \max_{1 \leq j \leq q} \max_{\ell \in \mc L^{(j)}} d^{(j)}_\ell \left\vert \wh{\theta}_{j} -\theta_{j} \right\vert \leq c_0 \sqrt{G \log (n)}   \right) \to 1 \text{ \ as \ } n \to \infty.   
\end{align*}

\item \label{thm:multilag2} If Assumption~\ref{assum:change-point2} holds and $c = 2$, then the conclusion of~\ref{thm:multilag1} holds.
\end{enumerate}
\end{theorem}

Under Assumption~\ref{assum:change-point2}~\ref{assum:change-size2}, which is weaker than Assumption~\ref{assum:change-point}~\ref{assum:change-size}, we may encounter a situation where $\sqrt{G/\log(n)} d^{(j)}_\ell = O(1)$ while $d^{(j)}_\ell > 0$ at some lag $\ell \in \mc L$.
Then, we cannot guarantee that such $\theta_j$ is detected by NP-MOJO at lag $\ell$ and, even so, we can only show that its estimator $\wt\theta \in \wt{\Theta}_\ell$ satisfies $\vert \wt\theta - \theta_j \vert = O(G)$. 
This requires setting the tuning parameter $c$ maximally for the clustering in Step~2 of multi-lag NP-MOJO, see Theorem~\ref{thm:multilag-consistent}~\ref{thm:multilag2}.
At the same time, there exists a lag well-suited for the localisation of each change point and Step~3 identifies an estimator detected at such lag, and the final estimator inherits the rate of estimation attained at the favourable lag.

\subsection{Threshold selection via dependent wild bootstrap}
\label{sec:bootstrap}

Theorem~\ref{thm:consistency} gives the choice of the threshold $\zeta_\ell (n,G) = c_\zeta \sqrt{\log(n)/G}$ which guarantees the consistency of NP-MOJO in multiple change point estimation.
The choice of $c_\zeta$ influences the finite sample performance of NP-MOJO but it depends on many unknown quantities involved in specifying the degree of serial dependence in $\{X_t\}_{t = 1}^n$ (see Assumptions~\ref{assum:functional-dep} and~\ref{assum:beta-mixing}), which makes the theoretical choice of little practical use.
Resampling is popularly adopted for the calibration of change point detection methods including threshold selection.
However, due to the presence of serial dependence, permutation-based approaches such as that adopted in \cite{matteson2014} or sample splitting adopted in \cite{padilla2021optimal} are inappropriate.

We propose to adopt the dependent wild bootstrap procedure proposed in \cite{leucht2013dependent}, in order to approximate the quantiles of $\max_{G \leq k \leq  n-G} T_\ell(G,k) $ in the absence of any change point, from which we select $\zeta_\ell(n, G)$. 
Let $\{ W_t^{[r]} \}_{t=1}^{n-G}$ denote a bootstrap sequence generated as a Gaussian AR($1$) process with $\Var (W_t^{[r]}) = 1$ and the AR coefficient $\exp(-1/b_n)$, where the sequence $\{ b_n\}$ is chosen such that $b_n = o (n)$ and $\lim_{n \to \infty} b_{n} = \infty$.
We construct bootstrap replicates using $\{ W_t^{[r]} \}_{t=1}^{n-G}$ as $T^{[r]}_{\ell} = \max_{G \leq k \leq n-G}  T^{[r]}_\ell(G, k)$, where

\begin{align*}
\begin{split}
T^{[r]}_\ell(G,k) &= \frac{1}{(G-\ell)^2} \left \{ \sum_{s,t = k - G+1}^{k - \ell} \bar{W}^{[r]}_{s,k} \bar{W}^{[r]}_{t,k} h(Y_s, Y_t) + \sum_{s,t = k+1}^{k +G - \ell} \bar{W}^{[r]}_{s-G,k}  \bar{W}^{[r]}_{t-G,k} h(Y_s, Y_t) \right. \\
&  \qquad \qquad \qquad \qquad  \left. -2 \sum_{s = k - G+1}^{k - \ell} \sum_{t = k +1}^{k +G - \ell} \bar{W}^{[r]}_{s,k} \bar{W}^{[r]}_{t-G,k} h(Y_s, Y_t)  \right \},
\end{split}
\end{align*}
with $\bar{W}^{[r]}_{t,k} = W^{[r]}_t - (G-\ell)^{-1} \sum_{u = {k} -G+1}^{{ k}-\ell} W^{[r]}_u$.
Independently generating $\{W^{[r]}_t\}_{t = 1}^{n - G}$ for $r = 1, \ldots, R$ ($R$ denoting the number of bootstrap replications), we store $T^{[r]}_\ell$ and select the threshold as $\zeta_\ell (n,G) = q_{1-\alpha} ( \{T^{[r]}_{\ell} \}_{r=1}^{R} )$, the $(1-\alpha)$-quantile of $\{T^{[r]}_{\ell} \}_{r=1}^{R}$ for the chosen level $\alpha \in (0, 1]$. 
Additionally, we can compute the importance score for each 
$\wh{\theta} \in \wh{\Theta}_\ell$ as 
\begin{align}\label{eq:cpt-score}
s(\wh{\theta}) = \frac{\l\vert\left\{ 1 \le r \le R : \, T_\ell (G, \wh{\theta}) \ge T^{[r]}_{\ell,r} \right\} \r\vert }{R}. 
\end{align}
 Taking a value between $0$ and $1$, the larger $s(\wh\theta)$ is, the more likely that there exists a change point close to $\wh\theta$ empirically.
The bootstrap procedure generalises to the multi-lag NP-MOJO straightforwardly.
In practice, we observe that setting $\wh{\theta}_j = \argmax_{\tilde{\theta} \in \mc C_j} s(\tilde{\theta})$ (with some misuse of the notation, $s(\cdot)$ is computed at the relevant lag for each $\wt\theta$) works well in Step~3 of multi-lag NP-MOJO. This is attributed to the fact that this score inherently takes into account the varying scale of the detector statistics at multiple lags and `standardises' the importance of each estimator.
In all numerical experiments, our implementation of multi-lag NP-MOJO is based on this choice of $\wh\theta_j$. 
We provide the algorithmic descriptions of NP-MOJO and its multi-lag extension in Algorithms~\ref{algo-single-lag} and~\ref{algo-multi-lag} in Appendix~\ref{sec:algo}.

\section{Implementation of NP-MOJO}
\label{sec:tuning}

In this section, we discuss the computational aspects of NP-MOJO and provide recommendations for the choice of tuning parameters based on extensive numerical results. \tcb{Numerical studies analysing NP-MOJO's sensitivity to these tuning parameters can be found in Appendix~\ref{sec:sim-study}.}
\smallskip

\textbf{Computational complexity:} owing to the moving sum-based approach, the cost of sequentially computing $T_{\ell}(G, k)$ from $T_{\ell}(G, k - 1)$ is $O(G)$, giving the overall cost of computing $T_{\ell}(G, k), \, G \le k \le n - G$, as $O (nG)$. 
Exact details of the sequential update are given in Appendix~\ref{sec:cost}. 
The bootstrap procedure described in Section~\ref{sec:bootstrap} is performed once per lag for simultaneously detecting multiple change points, in contrast with E-Divisive \citep{matteson2014} that requires the permutation-based testing to be performed for detecting each change point. 
With $R$ bootstrap replications, the total computational cost is $O (\vert \mc L \vert RnG)$ for multi-lag NP-MOJO using the set of lags $\mc L$ and bootstrapping, as opposed to $O (Rq n^2)$ for E-Divisive. \tcr{Furthermore, the bootstrap procedure can be parallelised in a straightforward manner, which we include as an option in the implementation of the method.}

\tcr{We ran simulations to compare the computational speed of the competing nonparametric methods. We simulate realisations under the change in mean model~\eqref{eq:mean:model}, with increasing values of sample size $n$ and the number of equispaced change points $q$ ($(n, q) \in \{ (100,1), (500,2), (1000,3), (2000,5), (5000,10), (10000, 20) \}$). We use the same settings for each method as in the main simulation study, using the parallelised version of multi-lag NP-MOJO when $n \geq 2000$, and compute the average run time over 100 realisations. The results are displayed in Figure~\ref{fig:speed}. The fastest method by far is cpt.np, followed by KCPA and NP-MOJO. E-Divisive and NWBS are noticeably slower than the other methods. In particular, when $n=10000$, the average running time of cpt.np is $0.17$ seconds, KCPA is $46.26$ seconds, NP-MOJO is $2.31$ minutes, NWBS is $30.06$ minutes, and E-Divisive is $70.37$ minutes. Also, we observe that KCPA's running time is increasing at a faster rate than NP-MOJO's, and may exceed the running time of NP-MOJO for larger values of $n$.}

\begin{figure}[h!t!]
\centering
\includegraphics[width = 0.5\textwidth]{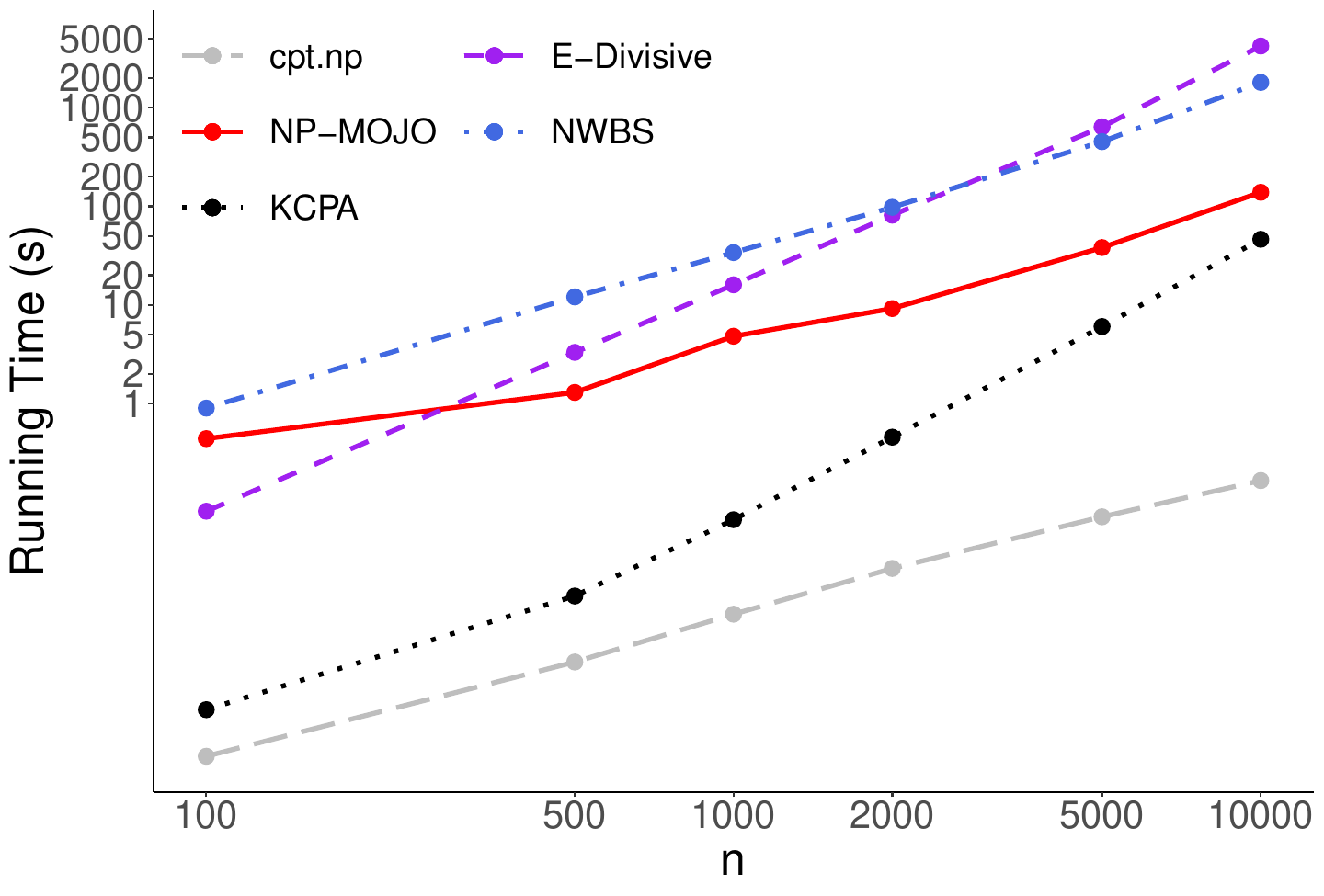}
\caption{Running time comparisons between the competing nonparametric methods. Both
axes are in the log scale.}
\label{fig:speed}
\end{figure}
{Kernel function:} \tcb{as with any kernel-based approach, NP-MOJO's performance will vary with the choice of kernel, and a kernel that works well for one type of change point may not be the best for another type of change point. }
Based on empirical performance and versatility to a wide range of change point scenarios (see Appendix~\ref{sec:sim-study:kernel}), we recommend the use of the kernel function $h_2$ in Lemma~\ref{lemma:weight-int-identities}~\ref{lemma:weight-int2}. The parameter $\delta$ is set using the `median trick', a common heuristic used in kernel-based methods \citep{li2019scan}. Specifically, we set $\delta$ to be a half the median of all $\Vert Y_s - Y_t \Vert^2$ involved in the calculation of $\tstatkl$. For $p$-variate i.i.d.\ Gaussian data with common variance~$\sigma^2$, this corresponds to $\delta \approx \sigma p$ as the dimension $p$ increases \citep{ramdas2015decreasing}. \tcb{As with the kernel $h_2$, the median trick can also be used when setting $\beta$ if the kernel $h_1$ is used.}

\textbf{Bandwidth:} due to the nonparametric nature of NP-MOJO, it is advised to use a larger bandwidth than that shown to work well for the moving sum procedure for univariate mean change detection \citep{eichinger2018}. 
In our simulation studies and data applications, we set $G=\lfloor n/6 \rfloor$. 

It is often found that using multiple bandwidths and merging the results improves the adaptivity of moving window-based procedures, such as the `bottom-up' merging proposed by \cite{messer2014} or the localised pruning of \citep{cho2020two}. We empirically explore the multiscale extension of the multi-lag NP-MOJO with bottom-up merging, see Appendix~\ref{sec:multiscale} for details of its implementation and Appendix~\ref{sec:sim-multiscale} for a proof of concept numerical study \tcr{involving multiscale change point scenarios.}
\tcb{We leave a theoretical investigation into the multiscale extension of NP-MOJO for future research.}

\textbf{Parameters for change point estimation:} we set $\eta = 0.4$ in~\eqref{eq:mosum:est} following the recommendation in \cite{meier2021mosum}. 
For multi-lag NP-MOJO, we set $c = 1$ for clustering the estimators from multiple lags, a choice that lies between those recommended in Theorem~\ref{thm:multilag-consistent}~\ref{thm:multilag1} and~\ref{thm:multilag2}, since we do not know whether Assumptions~\ref{assum:change-point} or~\ref{assum:change-point2} hold in practice.
\tcb{Appendices~\ref{sec:sim-study:eta} and~\ref{sec:sim-study:c} demonstrate that within a reasonable range, NP-MOJO is insensitive to the choices of $\eta$ and $c$.}
To further guard against spurious estimators, we only accept those $\wh\theta$ that lie in intervals of length greater than $\lfloor 0.02 G \rfloor$ where the corresponding $T_\ell(G, k)$ exceeds $\zeta_\ell (n, G)$.

\textbf{Parameters for the bootstrap procedure:} the choice of $b_{n}$ sets the level of dependence in the multiplier bootstrap sequences. \cite{leucht2013dependent} show that a necessary condition is that $\lim_{n \to \infty} (b_{n}^{-1} + b{_n} n^{-1}) =0$, giving a large freedom for choice of $b_{n}$. We recommend $b_{n} =  1.5 n^{1/3}$, which works well in practice. 
\tcb{Appendix~\ref{sec:sim-study:bn} demonstrates that within a reasonable range, NP-MOJO is insensitive to the choice of $b_n$. As for $\alpha$, its choice amounts to setting the level of significance in statistical testing.
This provides a more systematic alternative to the problem of model selection in multiple change point detection compared to others, such as those requiring the selection of a threshold that is known up to a rate (or a range of rates, see e.g.\ \citeauthor{misael2023change}, \citeyear{misael2023change}), or constants involved in the penalty of a penalised cost function \citep{arlot2019kernel}.}
In all numerical experiments, we use $\alpha = 0.1$ with $R = 499$ bootstrap replications.

\textbf{Set of lags $\mc L$:} \tcb{the flexibility of NP-MOJO in its ability to detect changes in dependence, comes at the price of having to select the set of lags $\mc L$.} The choice of $\mc L$ depends on the practitioner's interest and domain knowledge, a problem commonly faced by general-purpose change point detection methods, such as the choice of the quantile level in \cite{jula2022multiscale}, the parameter of interest in \cite{zhao2022segmenting} and the estimating equation in \cite{kirch2024data}.
For example, for monthly data, using $\mc L = \{0, 3, 12\}$ allows for detecting changes in the quarterly and yearly seasonality.
Even when the interest lies in detecting changes in the marginal distribution only, it helps to jointly consider multiple lags, since any marginal distributional change is likely to result in changes in the joint distribution of $(X_{t}, X_{t+\ell})$. \tcb{As we consider time series that exhibit short range dependence, we would expect that NP-MOJO will not have detection power at large lags.}
In simulations, we use $\mc L = \{0, 1, 2\}$ which works well not only for detecting changes in the mean and the second-order structure, but also for detecting changes in (non-linear) serial dependence and higher-order characteristics. \tcr{For a practical approach to lag selection, see Appendix~\ref{sec:lag-select} in the supplementary material, where we propose a semi-automatic method for choosing the set of lags $\mc L$ given some initial set $\tilde{\mc L}$. }

\section{Simulation study}
\label{sec:sim}

\tcr{We conduct extensive simulation studies with varying change point scenarios ($30$ scenarios where $q \ge 1$, $7$ with $q = 0$), sample sizes ($n \in \{ 500, 1000, 2000, 10000 \}$) and dimensions $p \in \{ 1,2, 5,10 \}$, and consider both evenly-spaced and multiscale change point settings}. We provide complete descriptions of the simulation studies in Appendix~\ref{sec:sim-study} where, for comparison, we consider not only nonparametric but also parametric data segmentation procedures well-suited to detect the types of changes in consideration, which include changes in the mean, second-order and higher-order moments and non-linear serial dependence.
Due to space constraint, here we focus on a selection of the results \tcr{in the evenly-spaced setting with $n = 1000$} comparing both single-lag and multi-lag NP-MOJO 
(denoted by NP-MOJO-$\ell$ and NP-MOJO-$\mc L$ respectively), with the nonparametric competitors: E-Divisive \citep{matteson2014}, NWBS \citep{padilla2021optimal}, KCPA \citep{celisse2018new, arlot2019kernel} and cpt.np \citep{haynes2017computationally}.
E-Divisive and KCPA are applicable to multivariate data segmentation whilst NWBS and cpt.np are not. The scenarios are: 
\begin{enumerate}[label=\customenum*]
\item \label{model-b5} $X_t = \sum_{j = 0}^3 \Sigma_j^{1/2} \mathbb{I}{\{\theta_j + 1 \le t \le \theta_{j + 1}\}} \cdot \vep_t$, where $\vep_t = (\vep_{1t}, \vep_{2t})^\top$ with $\vep_{it} \sim_{\text{i.i.d.}} t_5$, $(\theta_{1}, \theta_{2}, \theta_{3}) = (250, 500, 750)$, $\Sigma_0 = \Sigma_2 = \begin{psmallmatrix}1 & 0\\ 0 & 1\end{psmallmatrix}$ and $\Sigma_1 = \Sigma_3 = \begin{psmallmatrix}1 & 0.9\\ 0.9 & 1\end{psmallmatrix}$. 
\item \label{model-c1} $X_t = X^{(j)}_t = a_j  X^{(j)}_{t-1} + \vep_t$ for $\theta_j + 1 \le t \le \theta_{j + 1}$, where $q = 2$, $(\theta_{1}, \theta_{2}) = (333, 667)$ and $(a_0, a_1, a_2) = (-0.8, 0.8, -0.8)$.
\item \label{model-c3} $X_t = X^{(j)}_t = \sigma^{(j)}_t \vep_t$ with $(\sigma_t^{(j)})^2 = \omega_j + \alpha_j (X^{(j)}_{t - 1})^2 + \beta_j (\sigma^{(j)}_{t - 1})^2$ for $\theta_j + 1 \le t \le \theta_{j + 1}$, where $q = 1$, $\theta_1 = 500$, $(\omega_0, \alpha_0, \beta_0) = (0.01, 0.7, 0.2)$ and $(\omega_1, \alpha_1, \beta_1) = (0.01, 0.2, 0.7)$.
\item\label{model-d3} $X_t = 0.4X_{t-1} + \vep_t$ where $\vep_t \sim_{\text{i.i.d.}} \mc N(0, 0.5^2)$ for $t \leq \theta_1$ and $t \geq \theta_2 + 1$, and $ \vep_t \sim_{\text{i.i.d.}} \text{Exponential}(0.5) - 0.5$ for $\theta_1 + 1 \leq t \leq \theta_2$, with $q = 2$ and $(\theta_{1}, \theta_{2}) = (333, 667)$.
\end{enumerate}

\tcb{Additional simulations for differing sample sizes can be found in Appendix~\ref{sec:sim-study:n},} \tcr{and simulations with uneven spacing between neighbouring segments examining the performance of the multiscale version of multi-lag NP-MOJO are given in Appendix~\ref{sec:sim-multiscale}.}
The above scenarios consider changes in the covariance of bivariate, non-Gaussian random vectors in~\ref{model-b5}, changes in the autocorrelation (while the variance stays unchanged) in~\ref{model-c1}, a change in the parameters of an ARCH($1$, $1$) process in~\ref{model-c3}, and changes in higher moments of serially dependent observations in~\ref{model-d3}. 
Table~\ref{main-text-table} reports the distribution of the estimated number of change points and the average covering metric (CM) and V-measure (VM) over 1000 realisations. 
Taking values between $[0, 1]$, CM and VM close to $1$ indicates better accuracy in change point location estimation, see Appendix~\ref{sec:power} for their definitions and complete discussions of change point scenarios.

In the case of~\ref{model-c1}, we have $q_\ell = 0, \, \ell \ne 1$, while $q_1 = 2$, and thus we report $\wh q_\ell - q_\ell$ for the respective single-lag NP-MOJO-$\ell$.
Across all scenarios, NP-MOJO-$\mc L$ shows good detection and estimation accuracy and demonstrates the efficacy of considering multiple lags, see~\ref{model-c3} and~\ref{model-d3} in particular.
As the competitors are calibrated for the independent setting, they tend to either over- or under-detect the number of change points in the presence of serial dependence in~\ref{model-c1}, \ref{model-c3} and~\ref{model-d3}.
In Appendix~\ref{sec:power}, we compare NP-MOJO against change point methods proposed for time series data where it performs comparably to methods specifically calibrated for the change point scenarios considered. 

 \begingroup
 \setlength{\tabcolsep}{3pt}
 \setlength{\LTcapwidth}{\textwidth}
   {\small
 \begin{longtable}{c c ccccc cc}
 \caption{Distribution of the estimated number of change points and the average CM and VM over 1000 realisations. The modal value of $\wh q - q$ in each row is given in bold.
 Also, the best performance for each metric is underlined for each scenario. }
 \label{main-text-table}
 \endfirsthead
 \endhead
 \toprule	
 && \multicolumn{5}{c}{$\wh{q} - q$ / $\wh{q}_\ell - q_\ell$} &   &   \\ 
  Model       & Method      &  $\leq-2$     & $-1$        & $\mathbf{0}$  & $1$    & $\geq 2$    & CM & VM       \\ 
  \cmidrule(lr){1-2} \cmidrule(lr){3-7} \cmidrule(lr){8-9}
        \ref{model-b5}    &  NP-MOJO-$0$      & 0.000       & 0.001  &   \bf{0.997}   &   0.002   & 0.000            & \underline{0.974} & \underline{0.959}  \\
  &  NP-MOJO-$1$        & 0.005    & 0.121    & \bf{0.867}     & 0.007   &  0.000 & 0.931 & 0.927 \\
        &  NP-MOJO-$2$      &   0.006  & 0.103   & \bf{0.884}   &   0.007  &  0.000  & 0.935 & 0.929 \\  &  NP-MOJO-$\mc L$   &  0.000      &    0.001   & \underline{\textbf{0.999}}  & 0.000     &   0.000          &  0.973 & 0.958    \\ \cmidrule(lr){3-7} \cmidrule(lr){8-9}
        &  E-Divisive      &  \bf{0.670}      &  0.189   &  0.101  & 0.032   &  0.008         & 0.431 & 0.335 \\
        &  KCPA      &  0.322      &  0.000  &  {\bf{0.662}}     &   0.015        &  0.001 & 0.775 & 0.725 \\      
      \cmidrule(lr){1-2}    \cmidrule(lr){3-7} \cmidrule(lr){8-9} 
        \ref{model-c1}    &  NP-MOJO-$0$     &  --     &  -- & \bf{0.851}   & 0.140     &     0.009          &  -- & --  \\
 &  NP-MOJO-$1$        &  0.000   &   0.002  & \bf{0.956}  &        0.042 & 0.000      & {0.978} & {0.961}   \\
       &  NP-MOJO-$2$      & -- &   -- & \bf{0.836}  & 0.149        &  0.015              & --   & -- \\ 
        &  NP-MOJO-$\mc L$   &  0.000      &    0.002   & \underline{\textbf{0.986}}  & 0.012     &   0.000          &  \underline{0.980} & \underline{0.963}    \\ \cmidrule(lr){3-7} \cmidrule(lr){8-9}
       &  E-Divisive       & 0.001 &  0.001   & 0.012   & 0.035   &  \textbf{0.951}         & 0.685 & 0.686  \\
       &  KCPA      & \textbf{0.792} & 0.002   & 0.065   & 0.025   &  {0.116}         & 0.399 & 0.132 \\      
       &  NWBS      & 0.013     & {0.001}   & 0.007 &   0.015  &   \textbf{0.964}    & 0.398   & 0.558  \\
       &  cpt.np      &  {0.000}   & 0.000 &   0.002  &   0.003    & \textbf{0.995}   & 0.593 & 0.647 \\  
  \cmidrule(lr){1-2}  \cmidrule(lr){3-7} \cmidrule(lr){8-9} 
              \ref{model-c3}    &  NP-MOJO-$0$      &  --      &  0.409  & \textbf{0.533}    & 0.056     &     0.002          & 0.744  & 0.484  \\
 &  NP-MOJO-$1$       & --    &   0.236  &  {\textbf{0.682}}  &  0.081  &    0.001 & 0.819 &  0.633    \\
       &  NP-MOJO-$2$      & -- &  0.299  & \textbf{{0.626}}  & 0.073       &  0.002    &  {0.787} &{0.571} \\
       &  NP-MOJO-$\mc L$   &  --      &    0.210   & \underline{\textbf{0.727}}  & 0.062     &   0.001          &  \underline{0.823} & \underline{0.645}     \\ \cmidrule(lr){3-7} \cmidrule(lr){8-9}
       &  E-Divisive      & --       &  0.032   & 0.327   & 0.211   &  \textbf{0.430}         & 0.742 & 0.602   \\
       &  KCPA       &  --      & \textbf{0.418}   & 0.262 & 0.171     &   0.149       & 0.667  & 0.370\\      
       &  NWBS      & --     & \textbf{0.895}   & 0.048 &   0.020  &   0.037    & 0.525   & 0.069  \\
       &  cpt.np      & --     & 0.000   & 0.013 &     0.047   &  \textbf{0.940}   & 0.634 & 0.554  \\  
  \cmidrule(lr){1-2}  \cmidrule(lr){3-7} \cmidrule(lr){8-9}     
              \ref{model-d3}   &  NP-MOJO-$0$      &  0.003      & 0.139  & {\bf{0.809}}     & 0.049      &   0.000         & {0.899}  &  {0.872} \\
 &  NP-MOJO-$1$        &0.006     &  0.155   &  \bf{0.792}   & 0.047    & 0.000 & {0.892} &  {0.864}\\
       &  NP-MOJO-$2$      & 0.021     & 0.248   &  \bf{0.685}  &   0.045  & 0.001   &  0.848 & 0.819 \\ 
       &  NP-MOJO-$\mc L$   &  0.002      &    0.082   & \underline{\textbf{0.914}}  & 0.002     &   0.000          &  \underline{0.917} & \underline{0.884}    \\ \cmidrule(lr){3-7} \cmidrule(lr){8-9}
       &  E-Divisive       & 0.005       &  0.002   & 0.072   & 0.118   &  \textbf{0.803}         & 0.681 & 0.707  \\
       &  KCPA      &  0.441      & 0.012   & \bf{0.481} & 0.052     &   0.014       & 0.667  & 0.500 \\      
       &  NWBS      & 0.047     & 0.015   & 0.139 &   0.124   &   \bf{0.675}    & 0.680   & 0.676 \\
       &  cpt.np      & 0.000     & 0.000   & 0.045 &     0.055   &  \textbf{0.900}   & 0.726 & 0.756 \\  
       \bottomrule
 \end{longtable}}
 \endgroup

\section{Data applications}
\label{sec:data}

\subsection{California seismology measurements data set}
\label{sec:seis}

We analyse a data set from the High Resolution Seismic Network, operated by the Berkeley Seismological Laboratory. 
Ground motion sensor measurements were recorded in three mutually perpendicular directions at $13$ stations near Parkfield, California, USA for $740$ seconds from 2am on December 23rd 2004. The data has previously been analysed in \cite{xie2019asynchronous} and \cite{chen2020high}. \cite{chen2020high} pre-process the data by removing a linear trend and down-sampling, and the processed data is available in the \verb!ocd! \verb!R! package \citep{ocd2020}. 
According to the Northern California Earthquake Catalog, an earthquake of magnitude 1:47 Md hit near Atascadero, California (50 km away from Parkfield) at 02:09:54.01.

\begin{figure}[H]
\centering
\includegraphics[width =\textwidth]{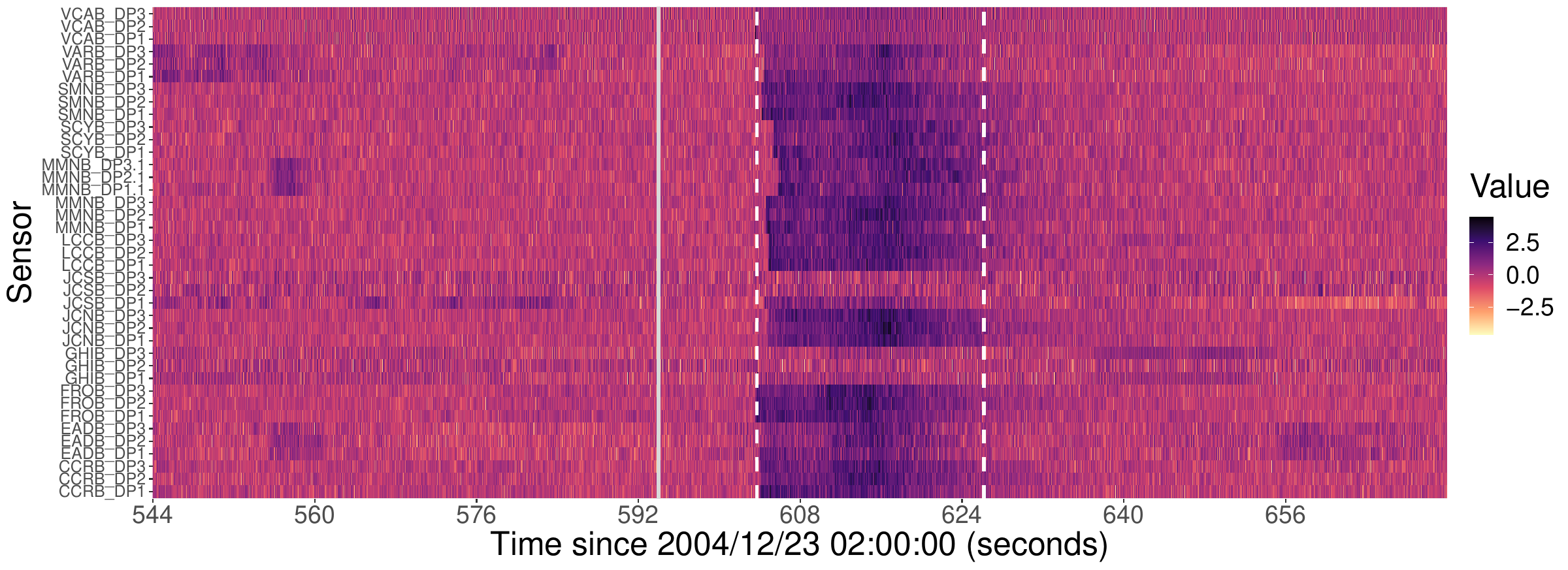}
\caption{Heat map of standardised sensor data. Change points detected by multi-lag NP-MOJO are shown in vertical dashed lines, and the time of the earthquake is given by solid vertical line.}
\label{parkfield-plot-2}
\end{figure}

We analyse time series of dimension $p = 39$ and length $n = 2000$ by taking a portion of the data set between $544$ and $672$ seconds after $2$am, which covers the time at which the earthquake occurred ($594$ seconds after).
We apply the multi-lag NP-MOJO with tuning parameters selected as in Section~\ref{sec:tuning}, using $G = 333$ and set of lags $\mc L = \{0, \ldots , 4\}$. 
We detect two changes at all lags; the first occurs at between 603.712 and 603.968 seconds after 2am and may be attributed to the earthquake.  
As noted in \cite{chen2020high}, P waves, which are the primary preliminary wave and arrive first after an earthquake, travel at up to $6$km/s in the Earth's crust.
This is consistent with the delay of approximately $9$ seconds between the occurrence of the earthquake and the first change point detected by multi-lag NP-MOJO.
We also note that performing online change point analysis, \cite{xie2019asynchronous} and \cite{chen2020high} report a change at 603.584 and 603.84 seconds after the earthquake, respectively.
The second change is detected at between 626.176 and 626.496 seconds after 2am.
It may correspond to the ending of the effect of the earthquake, as sensors return to `baseline' behaviour. 
Figure~\ref{parkfield-plot-2} plots the heat map of the data with each series standardised for ease of visualisation, along with the onset of the earthquake and the two change points detected by the multi-lag NP-MOJO.
It suggests, amongst other possible distributional changes, the time series undergoes mean shifts as found in \cite{chen2020high}.
We also examine the sample correlations computed on each of the three segments, see Figure~\ref{parkfield-cors} where the data exhibit a greater degree of correlation in segment $2$ compared to the other two segments.
Recalling that each station is equipped with three sensors, we notice that pairwise correlations from the sensors located at the same stations undergo greater changes in correlations. 
A similar observation is made about the sensors located at nearby stations.

\begin{figure}[]
\centering
\includegraphics[width =\textwidth]{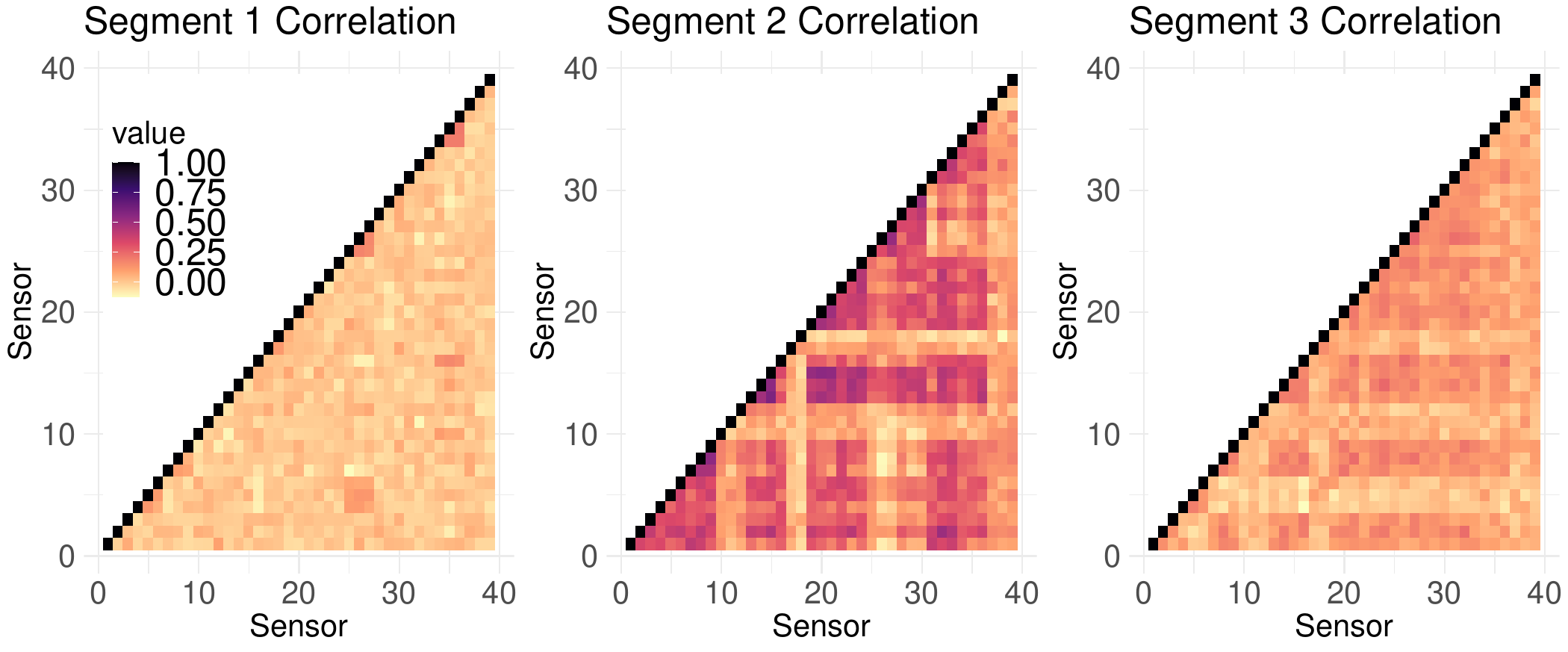}
\caption{Sample correlations from the three segments defined by the change point estimators.}
\label{parkfield-cors}
\end{figure}

\subsection{US recession data}
\label{sec:recession}

We analyse the US recession indicator data set.
Recorded quarterly between $1855$ and $2021$ ($n = 667$), $X_t$ is recorded as a $1$ if any month in the quarter is in a recession (as identified by the Business Cycle Dating Committee of the National Bureau of Economic Research), and $0$ otherwise.
The data has previously been examined for change points under piecewise stationary autoregressive models for integer-valued time series in \cite{hudecova2013structural} and \cite{diop2021piecewise}. 
We apply the multi-lag NP-MOJO with $G = 111$ and $\mc L = \{0, \ldots, 4\}$. 
All tuning parameters are set as recommended in Section~\ref{sec:tuning} with one exception, $\delta$ for the kernel $h_2$.
We select $\delta = 1$ for lag $0$ and $2$ otherwise, since pairwise distances for binary data are either $0$ or $1$ when $\ell  = 0$ such that the median heuristic would not work as desired. 

At all lags, we detect a single change point located between 1933:Q1 and 1938:Q2. Multi-lag NP-MOJO estimates the change point at 1933:Q1, which is comparable to the previous analyses: \cite{hudecova2013structural} report a change at 1933:Q1 and \cite{diop2021piecewise} at 1932:Q4. 
The change coincides with the ending of the Great Depression and beginning of World War II. 
The left panel of Figure~\ref{recession-plot-2} plots the detected change along with the sample average of $X_t$ over the two segments (superimposed on $\{X_t\}_{t = 1}^n$), showing that the frequency of recession is substantially lower after the change. 
The right panel plots the detector statistics $T_\ell(G, k)$ at lags $\ell \in \mc L$, divided by the respective threshold $\zeta_\ell(n, G)$ obtained from the bootstrap procedure. 
The thus-standardised $T_4(G, k)$, shown in solid line, displays the change point with the most clarity, attaining the largest value over the widest interval above the threshold (standardised to be one). 
At lag $4$, the detector statistic has the interpretation of measuring any discrepancy in the joint distribution of the recession indicator series and its yearly lagged values.

\begin{figure}[]
\centering
\includegraphics[width =\textwidth]{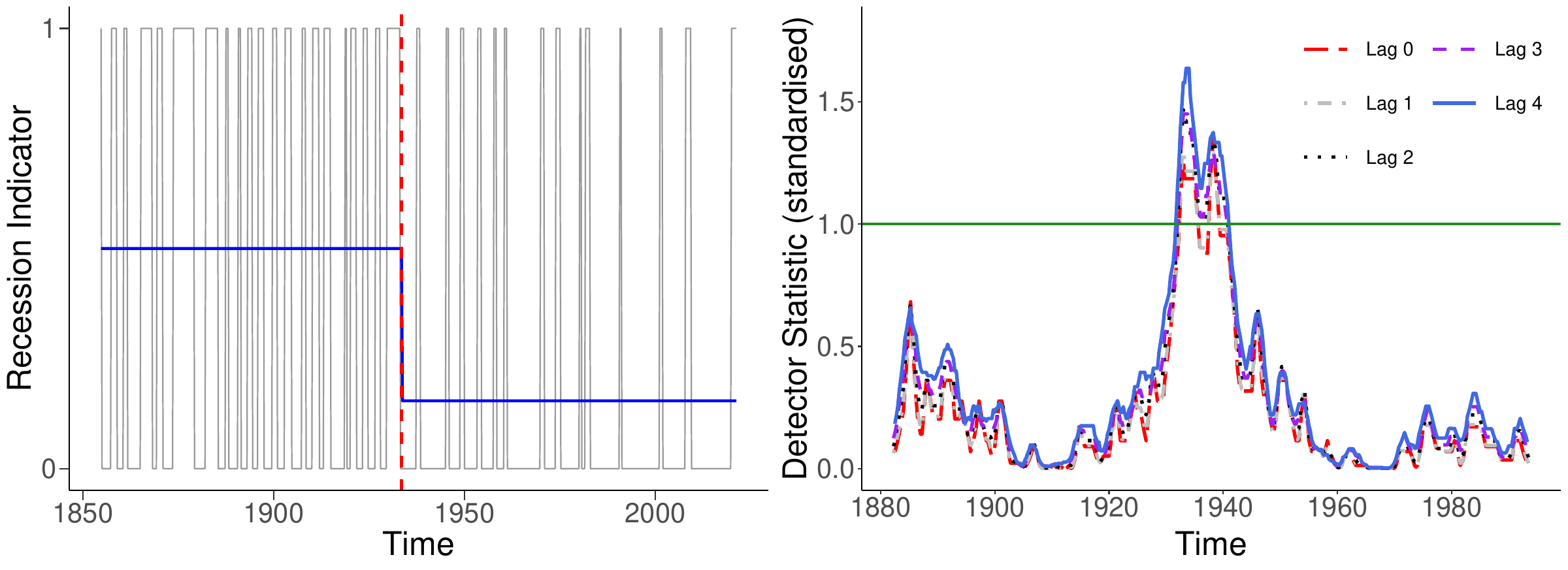}
\caption{Left: quarterly US recession indicator series. A change point detected by multi-lag NP-MOJO is shown in vertical dashed lines and the sample means over the two segments in solid line. Right: $T_\ell(G, k), \, G \le k \le n - G$ for lags $\ell \in \mc L$, after standardisation by respective thresholds.}
\label{recession-plot-2}
\end{figure}

\bibliographystyle{apalike}
\bibliography{paper-ref}

\begin{thebibliography}{}

\bibitem[Anastasiou et~al., 2022]{breakfast}
Anastasiou, A., Chen, Y., Cho, H., and Fryzlewicz, P. (2022).
\newblock {\em breakfast: Methods for Fast Multiple Change-Point Detection and
  Estimation}.
\newblock R package version 2.3.

\bibitem[Arbelaez et~al., 2010]{arbelaez2010contour}
Arbelaez, P., Maire, M., Fowlkes, C., and Malik, J. (2010).
\newblock Contour detection and hierarchical image segmentation.
\newblock {\em IEEE Transactions on Pattern Analysis and Machine Intelligence},
  33(5):898--916.

\bibitem[Arlot et~al., 2019]{arlot2019kernel}
Arlot, S., Celisse, A., and Harchaoui, Z. (2019).
\newblock A kernel multiple change-point algorithm via model selection.
\newblock {\em Journal of Machine Learning Research}, 20(162):1--56.

\bibitem[Aue et~al., 2009]{aue2009}
Aue, A., H{\"o}rmann, S., Horv{\'a}th, L., and Reimherr, M. (2009).
\newblock Break detection in the covariance structure of multivariate time
  series models.
\newblock {\em Ann. Statist.}, 37:4046--4087.

\bibitem[Bai and Perron, 1998]{bai1998estimating}
Bai, J. and Perron, P. (1998).
\newblock Estimating and testing linear models with multiple structural
  changes.
\newblock {\em Econometrica}, 66:47--78.

\bibitem[Bakirov et~al., 2006]{bakirov2006multivariate}
Bakirov, N.~K., Rizzo, M.~L., and Sz{\'e}kely, G.~J. (2006).
\newblock A multivariate nonparametric test of independence.
\newblock {\em Journal of Multivariate Analysis}, 97(8):1742--1756.

\bibitem[Boniece et~al., 2023]{boniece2022change}
Boniece, B.~C., Horv{\'a}th, L., and Jacobs, P.~M. (2023).
\newblock Change point detection in high dimensional data with u-statistics.
\newblock {\em TEST}, pages 1--53.

\bibitem[Carlstein, 1988]{carlstein1988}
Carlstein, E. (1988).
\newblock Nonparametric change-point estimation.
\newblock {\em The Annals of Statistics}, 16(1):188--197.

\bibitem[Carr et~al., 2017]{carr2017exceptional}
Carr, J.~R., Bell, H., Killick, R., and Holt, T. (2017).
\newblock Exceptional retreat of {N}ovaya {Z}emlya's marine-terminating outlet
  glaciers between 2000 and 2013.
\newblock {\em The Cryosphere}, 11(5):2149--2174.

\bibitem[Celisse et~al., 2018]{celisse2018new}
Celisse, A., Marot, G., Pierre-Jean, M., and Rigaill, G. (2018).
\newblock New efficient algorithms for multiple change-point detection with
  reproducing kernels.
\newblock {\em Computational Statistics \& Data Analysis}, 128:200--220.

\bibitem[Chakraborty and Zhang, 2021]{chakraborty2021high}
Chakraborty, S. and Zhang, X. (2021).
\newblock High-dimensional change-point detection using generalized homogeneity
  metrics.
\newblock {\em arXiv preprint arXiv:2105.08976}.

\bibitem[Chen and Friedman, 2017]{chen2017new}
Chen, H. and Friedman, J.~H. (2017).
\newblock A new graph-based two-sample test for multivariate and object data.
\newblock {\em Journal of the American statistical association},
  112(517):397--409.

\bibitem[Chen and Zhang, 2015]{chen2015graph}
Chen, H. and Zhang, N. (2015).
\newblock {Graph-based change-point detection}.
\newblock {\em The Annals of Statistics}, 43(1):139 -- 176.

\bibitem[Chen et~al., 2020]{ocd2020}
Chen, Y., Wang, T., and Samworth, R.~J. (2020).
\newblock {\em ocd: High-dimensional, multiscale online changepoint detection}.
\newblock {R} package version 1.1.

\bibitem[Chen et~al., 2022]{chen2020high}
Chen, Y., Wang, T., and Samworth, R.~J. (2022).
\newblock High-dimensional, multiscale online changepoint detection.
\newblock {\em Journal of the Royal Statistical Society: Series B (Statistical
  Methodology)}, 84(1):234--266.

\bibitem[Cho and Fryzlewicz, 2012]{cho2012a}
Cho, H. and Fryzlewicz, P. (2012).
\newblock Multiscale and multilevel technique for consistent segmentation of
  nonstationary time series.
\newblock {\em Stat. Sinica}, 22:207--229.

\bibitem[Cho and Fryzlewicz, 2015]{cho2015}
Cho, H. and Fryzlewicz, P. (2015).
\newblock Multiple-change-point detection for high dimensional time series via
  sparsified binary segmentation.
\newblock {\em Journal of the Royal Statistical Society: Series B (Statistical
  Methodology)}, 77(2):475--507.

\bibitem[Cho and Fryzlewicz, 2018]{hdbinseg2018}
Cho, H. and Fryzlewicz, P. (2018).
\newblock {\em hdbinseg: Change-Point Analysis of High-Dimensional Time Series
  via Binary Segmentation}.
\newblock {R} package version 1.0.1.

\bibitem[Cho and Fryzlewicz, 2023]{cho2021multiple}
Cho, H. and Fryzlewicz, P. (2023).
\newblock Multiple change point detection under serial dependence: Wild
  contrast maximisation and gappy schwarz algorithm.
\newblock {\em Journal of Time Series Analysis}, 45(3):479--494.

\bibitem[Cho and Kirch, 2022]{cho2020two}
Cho, H. and Kirch, C. (2022).
\newblock Two-stage data segmentation permitting multiscale change points,
  heavy tails and dependence.
\newblock {\em Annals of the Institute of Statistical Mathematics},
  74(4):653--684.

\bibitem[Cho and Kirch, 2023]{cho2021data}
Cho, H. and Kirch, C. (2023+).
\newblock Data segmentation algorithms: Univariate mean change and beyond.
\newblock {\em Econometrics and Statistics (to appear)}.

\bibitem[Chu et~al., 1995]{chu1995mosum}
Chu, C.-S.~J., Hornik, K., and Kaun, C.-M. (1995).
\newblock {MOSUM} tests for parameter constancy.
\newblock {\em Biometrika}, 82(3):603--617.

\bibitem[Chu and Chen, 2019]{chu2019}
Chu, L. and Chen, H. (2019).
\newblock {Asymptotic distribution-free change-point detection for multivariate
  and non-Euclidean data}.
\newblock {\em The Annals of Statistics}, 47(1):382 -- 414.

\bibitem[Davis et~al., 2018]{davis2018}
Davis, R.~A., Matsui, M., Mikosch, T., and Wan, P. (2018).
\newblock Applications of distance correlation to time series.
\newblock {\em Bernoulli}, 24(4):3087--3116.

\bibitem[Dette et~al., 2020]{dette2020multiscale}
Dette, H., Eckle, T., and Vetter, M. (2020).
\newblock Multiscale change point detection for dependent data.
\newblock {\em Scandinavian Journal of Statistics}, 47(4):1243--1274.

\bibitem[Diop and Kengne, 2021]{diop2021piecewise}
Diop, M.~L. and Kengne, W. (2021).
\newblock Piecewise autoregression for general integer-valued time series.
\newblock {\em Journal of Statistical Planning and Inference}, 211:271--286.

\bibitem[Eichinger and Kirch, 2018]{eichinger2018}
Eichinger, B. and Kirch, C. (2018).
\newblock A {MOSUM} procedure for the estimation of multiple random change
  points.
\newblock {\em Bernoulli}, 24:526--564.

\bibitem[Fan et~al., 2017]{fan2017multivariate}
Fan, Y., de~Micheaux, P.~L., Penev, S., and Salopek, D. (2017).
\newblock Multivariate nonparametric test of independence.
\newblock {\em Journal of Multivariate Analysis}, 153:189--210.

\bibitem[Fokianos and Pitsillou, 2017]{fokianos2017consistent}
Fokianos, K. and Pitsillou, M. (2017).
\newblock Consistent testing for pairwise dependence in time series.
\newblock {\em Technometrics}, 59(2):262--270.

\bibitem[Frick et~al., 2014]{frick2014}
Frick, K., Munk, A., and Sieling, H. (2014).
\newblock Multiscale change point inference.
\newblock {\em Journal of the Royal Statistical Society: Series B (Statistical
  Methodology)}, 76(3):495--580.

\bibitem[Fryzlewicz, 2014]{fryzlewicz2014wild}
Fryzlewicz, P. (2014).
\newblock Wild binary segmentation for multiple change-point detection.
\newblock {\em The Annals of Statistics}, 42(6):2243--2281.

\bibitem[Fryzlewicz and Subba~Rao, 2014]{psr2014}
Fryzlewicz, P. and Subba~Rao, S. (2014).
\newblock Multiple-change-point detection for auto-regressive conditional
  heteroscedastic processes.
\newblock {\em Journal of the Royal Statistical Society: Series B (Statistical
  Methodology)}, 76:903--924.

\bibitem[Gradshteyn and Ryzhik, 2014]{gradshteyn2014table}
Gradshteyn, I.~S. and Ryzhik, I.~M. (2014).
\newblock {\em {Table of Integrals, Series, and Products}}.
\newblock Academic Press.

\bibitem[Gretton et~al., 2012]{gretton2012kernel}
Gretton, A., Borgwardt, K.~M., Rasch, M.~J., Sch{\"o}lkopf, B., and Smola, A.
  (2012).
\newblock A kernel two-sample test.
\newblock {\em The Journal of Machine Learning Research}, 13(1):723--773.

\bibitem[Harchaoui et~al., 2009]{harchaoui2009}
Harchaoui, Z., Vallet, F., Lung-Yut-Fong, A., and Capp{\'e}, O. (2009).
\newblock A regularized kernel-based approach to unsupervised audio
  segmentation.
\newblock In {\em ICASSP}, pages 1665--1668.

\bibitem[Harel and Puri, 1989]{harel1989limiting}
Harel, M. and Puri, M.~L. (1989).
\newblock Limiting behavior of {U}-statistics, {V}-statistics, and one sample
  rank order statistics for nonstationary absolutely regular processes.
\newblock {\em Journal of Multivariate Analysis}, 30(2):181--204.

\bibitem[Haynes et~al., 2017]{haynes2017computationally}
Haynes, K., Fearnhead, P., and Eckley, I.~A. (2017).
\newblock A computationally efficient nonparametric approach for changepoint
  detection.
\newblock {\em Statistics and Computing}, 27(5):1293--1305.

\bibitem[Haynes and Killick, 2021]{haynes2021}
Haynes, K. and Killick, R. (2021).
\newblock {\em changepoint.np: Methods for nonparametric changepoint
  detection}.
\newblock {R} package version 1.0.3.

\bibitem[Hoeffding, 1994]{hoeffding1994probability}
Hoeffding, W. (1994).
\newblock Probability inequalities for sums of bounded random variables.
\newblock {\em The collected works of Wassily Hoeffding}, pages 409--426.

\bibitem[Hudecov{\'a}, 2013]{hudecova2013structural}
Hudecov{\'a}, S. (2013).
\newblock Structural changes in autoregressive models for binary time series.
\newblock {\em Journal of Statistical Planning and Inference},
  143(10):1744--1752.

\bibitem[Huskova and Slaby, 2001]{huvskova2001permutation}
Huskova, M. and Slaby, A. (2001).
\newblock Permutation tests for multiple changes.
\newblock {\em Kybernetika}, 37(5):605--622.

\bibitem[James and Matteson, 2015]{JSSv062i07}
James, N.~A. and Matteson, D.~S. (2015).
\newblock ecp: An {R} package for nonparametric multiple change point analysis
  of multivariate data.
\newblock {\em Journal of Statistical Software}, 62(7):1--25.

\bibitem[Jewell et~al., 2020]{jewell2020fast}
Jewell, S.~W., Hocking, T.~D., Fearnhead, P., and Witten, D.~M. (2020).
\newblock Fast nonconvex deconvolution of calcium imaging data.
\newblock {\em Biostatistics}, 21(4):709--726.

\bibitem[Killick and Eckley, 2014]{killick2014}
Killick, R. and Eckley, I.~A. (2014).
\newblock {changepoint}: An {R} package for changepoint analysis.
\newblock {\em Journal of Statistical Software}, 58:1--19.

\bibitem[Killick et~al., 2012]{killick2012}
Killick, R., Fearnhead, P., and Eckley, I.~A. (2012).
\newblock Optimal detection of changepoints with a linear computational cost.
\newblock {\em Journal of the American Statistical Association},
  107(500):1590--1598.

\bibitem[Kirch and Reckruehm, 2024]{kirch2024data}
Kirch, C. and Reckruehm, K. (2024).
\newblock Data segmentation for time series based on a general moving sum
  approach.
\newblock {\em Annals of the Institute of Statistical Mathematics}, pages
  1--29.

\bibitem[Korkas and Fryzlewicz, 2020]{wbsts2020}
Korkas, K. and Fryzlewicz, P. (2020).
\newblock {\em wbsts: Multiple Change-Point Detection for Nonstationary Time
  Series}.
\newblock {R} package version 2.1.

\bibitem[Korkas and Fryzlewicz, 2017]{korkas2017multiple}
Korkas, K.~K. and Fryzlewicz, P. (2017).
\newblock Multiple change-point detection for non-stationary time series using
  wild binary segmentation.
\newblock {\em Statistica Sinica}, pages 287--311.

\bibitem[Lavielle and Teyssiere, 2007]{lavielle2007adaptive}
Lavielle, M. and Teyssiere, G. (2007).
\newblock Adaptive detection of multiple change-points in asset price
  volatility.
\newblock In {\em Long Memory in Economics}, pages 129--156. Springer.

\bibitem[Leucht and Neumann, 2013]{leucht2013dependent}
Leucht, A. and Neumann, M.~H. (2013).
\newblock Dependent wild bootstrap for degenerate {U}- and {V}-statistics.
\newblock {\em Journal of Multivariate Analysis}, 117:257--280.

\bibitem[Li et~al., 2019]{li2019scan}
Li, S., Xie, Y., Dai, H., and Song, L. (2019).
\newblock Scan {B}-statistic for kernel change-point detection.
\newblock {\em Sequential Analysis}, 38(4):503--544.

\bibitem[Madrid~Padilla et~al., 2023]{misael2023change}
Madrid~Padilla, C.~M., Xu, H., Wang, D., Madrid~Padilla, O.~H., and Yu, Y.
  (2023).
\newblock Change point detection and inference in multivariate non-parametric
  models under mixing conditions.
\newblock In {\em Advances in Neural Information Processing Systems},
  volume~36, pages 21081--21134.

\bibitem[Matteson and James, 2014]{matteson2014}
Matteson, D.~S. and James, N.~A. (2014).
\newblock A nonparametric approach for multiple change point analysis of
  multivariate data.
\newblock {\em Journal of the American Statistical Association},
  109(505):334--345.

\bibitem[McGonigle and Cho, 2023a]{CptNonPar2023}
McGonigle, E.~T. and Cho, H. (2023a).
\newblock {\em CptNonPar: Nonparametric Change Point Detection for Multivariate
  Time Series}.
\newblock {R} package version 0.1.2.

\bibitem[McGonigle and Cho, 2023b]{mcgonigle2023robust}
McGonigle, E.~T. and Cho, H. (2023b).
\newblock Robust multiscale estimation of time-average variance for time series
  segmentation.
\newblock {\em Computational Statistics \& Data Analysis}, 179:107648.

\bibitem[Meier et~al., 2021]{meier2021mosum}
Meier, A., Kirch, C., and Cho, H. (2021).
\newblock mosum: A package for moving sums in change-point analysis.
\newblock {\em Journal of Statistical Software}, 97(1):1--42.

\bibitem[Messer et~al., 2014]{messer2014}
Messer, M., Kirchner, M., Schiemann, J., Roeper, J., Neininger, R., and
  Schneider, G. (2014).
\newblock A multiple filter test for the detection of rate changes in renewal
  processes with varying variance.
\newblock {\em Ann. Appl. Stat.}, 8:2027--2067.

\bibitem[Mokkadem, 1988]{mokkadem1988mixing}
Mokkadem, A. (1988).
\newblock Mixing properties of {ARMA} processes.
\newblock {\em Stochastic Processes and their Applications}, 29(2):309--315.

\bibitem[Padilla et~al., 2021]{padilla2021optimal}
Padilla, O. H.~M., Yu, Y., Wang, D., and Rinaldo, A. (2021).
\newblock Optimal nonparametric change point analysis.
\newblock {\em Electronic Journal of Statistics}, 15(1):1154--1201.

\bibitem[Padilla et~al., 2022]{padilla2022optimal2}
Padilla, O. H.~M., Yu, Y., Wang, D., and Rinaldo, A. (2022).
\newblock Optimal nonparametric multivariate change point detection and
  localization.
\newblock {\em IEEE Transactions on Information Theory}, 68(3):1922--1944.

\bibitem[Page, 1954]{page1954continuous}
Page, E.~S. (1954).
\newblock Continuous inspection schemes.
\newblock {\em Biometrika}, 41(1/2):100--115.

\bibitem[Preu{\ss} et~al., 2015]{preuss2015detection}
Preu{\ss}, P., Puchstein, R., and Dette, H. (2015).
\newblock Detection of multiple structural breaks in multivariate time series.
\newblock {\em Journal of the American Statistical Association},
  110(510):654--668.

\bibitem[Ramdas et~al., 2015]{ramdas2015decreasing}
Ramdas, A., Reddi, S.~J., P{\'o}czos, B., Singh, A., and Wasserman, L. (2015).
\newblock On the decreasing power of kernel and distance based nonparametric
  hypothesis tests in high dimensions.
\newblock {\em Proceedings of the AAAI Conference on Artificial Intelligence},
  29(1).

\bibitem[Rigaill and Marot, 2018]{kernseg2018}
Rigaill, G. and Marot, G. (2018).
\newblock {\em KernSeg: Kernel Based Segmentation}.
\newblock {R} package version 0.0.2.

\bibitem[Rosenberg and Hirschberg, 2007]{rosenberg2007v}
Rosenberg, A. and Hirschberg, J. (2007).
\newblock V-measure: A conditional entropy-based external cluster evaluation
  measure.
\newblock In {\em Proc. Conf. on Empirical Methods in Natural Language
  Processing and Computational Natural Language Learning}, pages 410--420.

\bibitem[Safikhani and Shojaie, 2022]{safikhani2022joint}
Safikhani, A. and Shojaie, A. (2022).
\newblock Joint structural break detection and parameter estimation in
  high-dimensional nonstationary {VAR} models.
\newblock {\em Journal of the American Statistical Association},
  117(537):251--264.

\bibitem[Sejdinovic et~al., 2013]{sejdinovic2013equivalence}
Sejdinovic, D., Sriperumbudur, B., Gretton, A., and Fukumizu, K. (2013).
\newblock Equivalence of distance-based and rkhs-based statistics in hypothesis
  testing.
\newblock {\em The Annals of Statistics}, pages 2263--2291.

\bibitem[Sz{\'e}kely et~al., 2007]{szekely2007}
Sz{\'e}kely, G.~J., Rizzo, M.~L., and Bakirov, N.~K. (2007).
\newblock Measuring and testing dependence by correlation of distances.
\newblock {\em The Annals of Statistics}, 35(6):2769--2794.

\bibitem[Tecuapetla-G{\'o}mez and Munk, 2017]{tecuapetla2017}
Tecuapetla-G{\'o}mez, I. and Munk, A. (2017).
\newblock Autocovariance estimation in regression with a discontinuous signal
  and m-dependent errors: A difference-based approach.
\newblock {\em Scandinavian Journal of Statistics}, 44(2):346--368.

\bibitem[Truong et~al., 2020]{truong2020selective}
Truong, C., Oudre, L., and Vayatis, N. (2020).
\newblock Selective review of offline change point detection methods.
\newblock {\em Signal Processing}, 167:107299.

\bibitem[van~den Burg and Williams, 2020]{van2020evaluation}
van~den Burg, G.~J. and Williams, C.~K. (2020).
\newblock An evaluation of change point detection algorithms.
\newblock {\em arXiv preprint arXiv:2003.06222}.

\bibitem[Vanegas et~al., 2022]{jula2022multiscale}
Vanegas, L.~J., Behr, M., and Munk, A. (2022).
\newblock Multiscale quantile segmentation.
\newblock {\em Journal of the American Statistical Association},
  117(539):1384--1397.

\bibitem[Wang et~al., 2021]{wang2021optimal}
Wang, D., Yu, Y., and Rinaldo, A. (2021).
\newblock {Optimal covariance change point localization in high dimensions}.
\newblock {\em Bernoulli}, 27(1):554 -- 575.

\bibitem[Wu, 2005]{wu2005nonlinear}
Wu, W.~B. (2005).
\newblock Nonlinear system theory: Another look at dependence.
\newblock {\em Proceedings of the National Academy of Sciences},
  102(40):14150--14154.

\bibitem[Xie et~al., 2019]{xie2019asynchronous}
Xie, L., Xie, Y., and Moustakides, G.~V. (2019).
\newblock Asynchronous multi-sensor change-point detection for seismic tremors.
\newblock In {\em 2019 IEEE International Symposium on Information Theory
  (ISIT)}, pages 787--791. IEEE.

\bibitem[Xu et~al., 2022]{changepoints2022}
Xu, H., Wang, D., Zhao, Z., and Yu, Y. (2022).
\newblock {\em changepoints: A Collection of ChangePoint Detection Methods}.
\newblock {R} package version 1.1.0.

\bibitem[Xu et~al., 2024]{xu2024change}
Xu, H., Wang, D., Zhao, Z., and Yu, Y. (2024).
\newblock Change point inference in high-dimensional regression models under
  temporal dependence.
\newblock {\em The Annals of Statistics (to appear)}.

\bibitem[Yousuf and Feng, 2022]{yousuf2022targeting}
Yousuf, K. and Feng, Y. (2022).
\newblock Targeting predictors via partial distance correlation with
  applications to financial forecasting.
\newblock {\em Journal of Business \& Economic Statistics}, 40(3):1007--1019.

\bibitem[Zhao et~al., 2022]{zhao2022segmenting}
Zhao, Z., Jiang, F., and Shao, X. (2022).
\newblock Segmenting time series via self-normalisation.
\newblock {\em Journal of the Royal Statistical Society Series B: Statistical
  Methodology}, 84(5):1699--1725.

\bibitem[Zhou, 2012]{zhou2012measuring}
Zhou, Z. (2012).
\newblock Measuring nonlinear dependence in time-series, a distance correlation
  approach.
\newblock {\em Journal of Time Series Analysis}, 33(3):438--457.

\bibitem[Zou et~al., 2014]{zou2014nonparametric}
Zou, C., Yin, G., Feng, L., and Wang, Z. (2014).
\newblock Nonparametric maximum likelihood approach to multiple change-point
  problems.
\newblock {\em The Annals of Statistics}, 42(3):970--1002.

\end{thebibliography}

\clearpage

\appendix

\appendixpage
\numberwithin{equation}{section}
\numberwithin{figure}{section}
\numberwithin{table}{section}
\numberwithin{theorem}{section}
\numberwithin{lemma}{section}

\section{Additional discussions about NP-MOJO}\label{sec:a}

\subsection{Computational complexity}
\label{sec:cost}

As briefly discussed in Section~\ref{sec:tuning}, we can perform a sequential update of $T_{\ell}(G, k)$ to enable efficient computation. 
By symmetry of the kernel $h$, we only need to calculate $h(Y_s , Y_t)$ for $(s, t)$ satisfying $1 \leq t \leq s \leq n$ and $|s - t| \leq 2G-\ell$, giving $O (n G)$ total computations for evaluating $h(Y_s , Y_t)$ for such $s$ and $t$. Then, writing 
\begin{align*}
\tstatkl &=   \frac{1}{(G-\ell)^2} \left\{ \sum_{s,t = k - G+1}^{k - \ell} h(Y_s, Y_t) + \sum_{s,t = k+1}^{k +G - \ell} h(Y_s, Y_t) -2 \sum_{s = k - G+1}^{k - \ell} \sum_{t = k +1}^{k +G - \ell} h(Y_s, Y_t)  \right\} \\
& =: T^{(1)}_{\ell}(G,k) + T^{(1)}_{\ell}(G,k+G) - 2 T^{(2)}_{\ell} (G,k),
\end{align*}
we can sequentially update $T^{(1)}_{\ell}(G, k)$ and $T^{(2)}_{\ell}(G, k)$. For example,
\begin{align*}
 T^{(1)}_{\ell}(G,k+1) =& \, T^{(1)}_{\ell}(G,k) -2 \sum_{s = k - G + 1}^{k-\ell} h(Y_{s}, Y_{k-G+1}) + 2 \sum_{s = k-G+2}^{k-\ell +1} h(Y_{s}, Y_{k-G+2}) \\
 & + h(Y_{k-G+1}, Y_{k-G+1}) - h(Y_{k-G+2}, Y_{k-G+2}),
\end{align*}
and a similar updating equation is available for $T^{(2)}_\ell(G, k)$. 
This update can be performed efficiently by pre-computing $\sum_{s = k - G + 1}^{k - \ell} h(Y_s, Y_{k - G + u})$ for all $k$ and $u = 1, 2$, which requires $O(n)$ computations.
In a similar fashion, the bootstrap replicates $T^{[r]}_{\ell}, \, 1 \le \ell \le R$, can also be computed using sequential updates in $O (nG)$ computational cost, giving the total cost for multi-lag NP-MOJO using the set of lags $\mc L$ as $O ( \vert \mc L \vert R nG)$. 
\subsection{Alternative weight function}
\label{sec:weight}

The following lemma describes the use of an additional weight function and kernel pair, supplementing Lemma~\ref{lemma:weight-int-identities} in the main text.

\begin{lemma}
\label{lemma:weight-int3} 
For any $\gamma \in (0,2)$, suppose that $d^{(j)}_\ell$ is obtained with
\begin{align*}
w_3(u, v) = C_3 (\gamma, p)^{-1}\left( \Vert u \Vert^2 + \Vert v \Vert^2 \right)^{-(\gamma+2p)/2} \text{ \ with \ } C_3(\gamma, p) = \frac{2 \pi^{p/2} \Gamma (1 - \gamma/2 ) }{\gamma 2^{\gamma} \Gamma((p+\gamma)/2 )}.   
\end{align*}  
If $\max_{0 \le j \le q} \max_{1 \le i \le p} \mathbb{E} ( \vert X_{1i}^{(j)}\vert^{\gamma} ) \le C < \infty$, then the function $h_3 : \mathbb{R}^{2p} \times \mathbb{R}^{2p} \to [ 0 , \infty )$ defined as $h_3 (x, y) =  \Vert x - y \Vert^\gamma $ for $x, y \in \mathbb{R}^{2p}$, satisfies
\begin{align*}
d^{(j)}_{\ell} & = 2  \mathbb{E} \left\{ h_3 \left(\tilde Y_1^{(j)}, Y_1^{(j-1)} \right) \right\} -  \mathbb{E} \left\{ h_3 \left( Y_1^{(j)},\tilde Y_1^{(j)} \right) \right\} - \mathbb{E} \left\{ h_3 \left(Y_1^{(j-1)}, \tilde Y_1^{(j-1)} \right) \right\}.
\end{align*}
\end{lemma}

The weight function $w_3$ was previously used in \cite{bakirov2006multivariate} in the context of independence testing and in \cite{matteson2014} for measuring changes in the marginal distribution of independent data (with $\ell = 0$).
In contrast to $w_1$ and $w_2$ given in Lemma~\ref{lemma:weight-int-identities}, $w_3$ is non-separable and non-integrable, and does not fulfil Assumption~\ref{assum:kernel}. 
As a consequence, we require an additional condition on the moments of $X_t$ for $d^{(j)}_\ell$ to be well-defined as well as for the consistency of NP-MOJO when it is applied with the kernel $h_3$. 

\subsection{\tcb{Multiscale extension}}\label{sec:multiscale}

In Section~\ref{sec:tuning}, we discussed the possibility of extending NP-MOJO to the multiple bandwidth setting. A multiple-bandwidth moving window procedure can be particularly beneficial in the presence of multiscale change points (a mixture of smaller changes over long time periods, and large changes over short time periods), when the single bandwidth approach may fail. Here, we describe explicitly how a multiscale version of NP-MOJO can be implemented.

We use the multi-lag NP-MOJO method combined with the `bottom-up’ merging as proposed by \cite{messer2014} (see also \cite{meier2021mosum}). Denote the range of bandwidths by $\mc G = \{G_r, 1 \leq r \leq R : G_1 < \ldots < G_R \}$, and let
$\widehat{\Theta} (G)$ denote the set of estimators detected using multi-lag NP-MOJO with some bandwidth $G \in \mc G$. Then, we accept all estimators in $\widehat{\Theta}(G_1)$ returned with the finest bandwidth $G_1$ to the set of final estimators $\widehat{\Theta}$ and, sequentially for $r = 2, \ldots , R$, accept $\widehat{\theta} \in \widehat{\Theta}(G_r)$ as a final change point if and only if $\min_{\theta \in \widehat{\Theta}} | \theta - \widehat{\theta} |\geq C G_r$, for some $C \in (0,1)$. That is, we only accept the estimators that do not detect the change points which have previously been detected at a finer scale. The algorithmic description of multiscale, multi-lag NP-MOJO is given in Algorithm~\ref{algo-multiscale} in Section~\ref{sec:algo}, \tcr{whilst simulations assessing its performance are provided in Section~\ref{sec:sim-multiscale}.}

\subsection{\tcr{Adaptive selection of the set of lags $\mc L$}}\label{sec:lag-select}

The multi-lag NP-MOJO approach gives a theoretically sound way for a practitioner wishing to detect changes at a given set of lags of interest $\mc L$, provided changes are detectable. However, in practice, a practitioner may be interested in running NP-MOJO for any lags that may possess change points, or may not have prior knowledge of the set of lags that may possess change points. To address this issue we suggest the following adaptive method to set $\mathcal{L}$, motivated by the bottom-up multiscale MOSUM algorithm:

\begin{enumerate}
\item[1.] Run the multi-lag NP-MOJO procedure for an initial set of lags $\tilde{\mathcal{L}} = \{ 0, 1 , \ldots , L \}$ for some $L >0$, obtaining the set of change point estimators $\widehat{\Theta}$. Set $\ell = L+1$.
\item[2.] Run the single-lag NP-MOJO procedure for the lag $\ell$, obtaining the set of change point estimators $\widehat{\Theta}_\ell$. 
\item[3.] For each $\widehat{\theta} \in \widehat{\Theta}_\ell$, add $\widehat{\theta}$ to $\widehat{\Theta}$ if $\widehat{\theta}$ satisfies:
\begin{equation*}
\min_{{\theta \in \Theta}}| \widehat{\theta} - \theta  | \geq cG  ,
\end{equation*}
where $c \in (0,1)$ is the lag-merging parameter from the multi-lag NP-MOJO procedure.
\item[4.] If no new change points were added to $\widehat{\Theta}$, stop the algorithm. Otherwise, set $\ell \leftarrow \ell +1$ and return to step 2 above.
\end{enumerate}

This provides a semi-automatic approach for choosing $\mathcal{L}$. However, we note that this does not preclude the possibility of a change point only occurring at some large positive lag, for example, with initial set $\tilde{ \mathcal{L}} = \{ 0, 1, , \ldots , L \}$ and the change point only occurs at some lag $\ell > L + 1$. In this scenario, the above approach would fail. In any case, we would not expect NP-MOJO to have detection power in such a setting, since we consider time series with short range dependence. The algorithmic description of the approach is given in Algorithm~\ref{algo-lag-select} in Section~\ref{sec:algo}, whilst a small simulation study assessing its performance is provided in Section~\ref{sec:sim-lag}.

\subsection{Algorithms}\label{sec:algo}

The algorithmic descriptions of the NP-MOJO procedure is summarised in Algorithms~\ref{algo-single-lag} and~\ref{algo-multi-lag}, corresponding to the single lag and multi-lag versions respectively. \tcb{The algorithmic description of multiscale, multi-lag NP-MOJO is given in Algorithm~\ref{algo-multiscale}}, \tcr{and the algorithm for NP-MOJO with adaptive lag selection is given in Algorithm~\ref{algo-lag-select}.}

\begin{algorithm}[]
\DontPrintSemicolon
 \KwIn{Multivariate time series $\{X_{t} \}_{t=1}^n$, bandwidth $G$, lag $\ell$, kernel $h$, selection parameter $\eta$, threshold level $\alpha$, bootstrap parameters $b_n$ and $R$ \;}
 \BlankLine

    \For{$k \in \{G,\ldots , n-G\}$}{
    \BlankLine
       Compute $\tstatkl$ \;
    }
      \For{$r \in \{1, ..., R\}$}{
       Compute $T^{[r]}_{\ell}$ \;
     }
    $\zeta_\ell (n,G)\leftarrow q_{1-\alpha} ( \{ T^{[r]}_{\ell}  \}_{r=1}^{R}) $ \;
                $\wh{\Theta}_\ell \leftarrow$ Set of change point estimators obtained with bandwidth $G$ and threshold $\zeta_\ell (n,G)$ according to~\eqref{eq:mosum:est} \;
      \For{$\wh \theta \in  \wh \Theta_{\ell}$}{
      \BlankLine
       Compute $s(\widehat{\theta} ) $ according to Equation~\eqref{eq:cpt-score}  \;
     }
 \KwOut{Estimated change locations $\wh{\Theta}_\ell$, change scores $\{ s ( \widehat{\theta} ) : \ \widehat{\theta} \in \wh{\Theta}_\ell \}$}
 \caption{Single-lag NP-MOJO algorithm}
 \label{algo-single-lag}
\end{algorithm}

\begin{algorithm}[]
\DontPrintSemicolon
\SetKwData{npMOJO}{NP-MOJO}
\SetKwData{sort}{SORT}
 \KwIn{Multivariate time series $\{X_{t} \}_{t=1}^n$, bandwidth $G$, set of lags $\mc L$, kernel $h$, selection parameter $\eta$, threshold level $\alpha$, bootstrap parameters $b_n$ and $R$, merge parameter $c$ \;}
Initialise $\wh \Theta \leftarrow  \emptyset$ \;
    \For{$\ell \in \mc L$}{
    \BlankLine
       $ \{ \wh{\Theta}_\ell , \{ s ( \widehat{\theta}) : \ \widehat{\theta} \in \wh{\Theta}_\ell \}  \}  \leftarrow \npMOJO(\{X_{t} \}_{t=1}^n, G, \ell, h , \eta , \alpha, b_n, R  ) $
    }
    \BlankLine
    $\widetilde{\Theta} \leftarrow  \bigcup_{\ell \in \mc L} \widehat{\Theta}_\ell$ \;
    $\widetilde{\Theta}_{1} \leftarrow   \widetilde{\Theta} $, $j \leftarrow 1$ \;
    \BlankLine
        \While{$ \widetilde{\Theta}_{j} \neq \emptyset$}{
        \BlankLine
       $\wh{\theta} \leftarrow \min \wt \Theta_{j} $  \;
       $\mc C_j \leftarrow \{ \tilde{\theta} \in \widetilde{\Theta}_{j} : \tilde{\theta} - \wh{\theta}  < cG \} $ \;
       $\wh{\theta}_j \leftarrow \argmax_{\tilde{\theta} \in \mc C_j}  s(\tilde{\theta})$ \;
       $\wh{\Theta} \leftarrow \wh{\Theta} \cup \{\wh{\theta}_j\}$ \;
       $\widetilde{\Theta}_{j+1} \leftarrow \widetilde{\Theta}_{j} \setminus \mc C_j$ \;
       $j \leftarrow j +1$
    }
 \KwOut{Estimated change locations $\wh{\Theta}$, estimated number of changes $\wh{q} = \vert \wh \Theta \vert$}
 \caption{Multi-lag NP-MOJO algorithm}
 \label{algo-multi-lag}
\end{algorithm}

\begin{algorithm}[]
\DontPrintSemicolon
\SetKwData{npMOJOmulti}{MULTI-LAG-NP-MOJO}
\SetKwData{sort}{SORT}
 \KwIn{Multivariate time series $\{X_{t} \}_{t=1}^n$, set of bandwidths $\mc G$, set of lags $\mc L$, kernel $h$, selection parameter $\eta$, threshold level $\alpha$, bootstrap parameters $b_n$ and $R$, lag merge parameter $c$, bandwidth merge parameter $C$ \;}
Initialise $\widehat{\Theta} \leftarrow \mc P \leftarrow \emptyset$ 
\BlankLine
\For{$G \in \mc G$}{
    $\widehat{\Theta} (G) \leftarrow \npMOJOmulti(\{X_{t} \}_{t=1}^n, G, \mc L, h , \eta , \alpha, b_n, R, c)$
\BlankLine

\lFor{$\wh{\theta} \in \wh{\Theta}(G)$}{ Add $(\wh{\theta},G)$ to $\mc P$}
}
\BlankLine

\For{$\wh{k} \in \mc P$ in increasing order with respect to $G$}{
    \lIf{\textup{$\min_{\wh{\theta} \in \widehat{\Theta}}| \wh{k}- \wh{\theta}  | \geq C G$}}{Add $\wh{k}$ to $\wh{\Theta}$}
}
\BlankLine
 \KwOut{Estimated change locations $\wh{\Theta}$, estimated number of changes $\wh{q} = \vert \wh \Theta \vert$}
 \caption{\tcb{Multiscale multi-lag NP-MOJO algorithm}}
 \label{algo-multiscale}
\end{algorithm}

\begin{algorithm}[]\label{algo-lag-select}
\DontPrintSemicolon
\SetKwData{npMOJOmulti}{MULTI-LAG-NP-MOJO}
\SetKwData{npMOJO}{NP-MOJO}
\SetKwData{sort}{SORT}
 \KwIn{Multivariate time series $\{X_{t} \}_{t=1}^n$, set of bandwidths $\mc G$, initial set of lags $\tilde{ \mc L}$, kernel $h$, selection parameter $\eta$, threshold level $\alpha$, bootstrap parameters $b_n$ and $R$, lag merge parameter $c$, bandwidth merge parameter $C$ \;}
\BlankLine
    $\wh{\Theta} \leftarrow \npMOJOmulti(\{X_{t} \}_{t=1}^n, G, \tilde{ \mc L}, h , \eta , \alpha, b_n, R, c)$
\BlankLine
$\ell \leftarrow \max \{ \tilde{\mc L} \} + 1 $ \\
$\text{FLAG} \leftarrow \text{TRUE}$ \\
\While{$\text{FLAG} = \text{TRUE}$ }{
   $\tilde{q} \leftarrow \vert \wh \Theta \vert$ 
   \BlankLine
    $ \{ \wh{\Theta}_\ell , \{ s ( \widehat{\theta}) : \ \widehat{\theta} \in \wh{\Theta}_\ell \}  \}  \leftarrow \npMOJO(\{X_{t} \}_{t=1}^n, G, \ell, h , \eta , \alpha, b_n, R  ) $ 
    \BlankLine
   \For{$\wh{k} \in \wh{\Theta}_\ell$}{
    \lIf{\textup{$\min_{\wh{\theta} \in \widehat{\Theta}}| \wh{k}- \wh{\theta}  | \geq c G$}}{Add $\wh{k}$ to $\wh{\Theta}$}
    }
    $\ell \leftarrow \ell + 1$ \\
    \lIf{$\vert \wh \Theta \vert = \tilde{q}$}{$\text{FLAG} \leftarrow \text{FALSE}$}
}

\BlankLine
 \KwOut{Estimated change locations $\wh{\Theta}$, estimated number of changes $\wh{q} = \vert \wh \Theta \vert$}
 \caption{\tcb{NP-MOJO with adaptive lag selection algorithm}}
 \label{algo-lag-select}
\end{algorithm}

\clearpage

\section{Complete simulation study}\label{sec:sim-study}

We examine the performance of NP-MOJO via a wide-ranging simulation study. For all experiments we simulate $1000$ replications. All tuning parameters are set as described in Section~\ref{sec:tuning}. 
We report the results from both single-lag and multi-lag NP-MOJO with the set of lags $\mc L = \{0, 1, 2\}$, which are denoted by NP-MOJO-$\ell$ and NP-MOJO-$\mc L$, respectively.

Where appropriate, we compare with competing methods for which R implementations are readily available.
In particular, we consider parametric methods which are designed specifically for detecting the particular types of changes we introduce in data generation, and their performance serve as a benchmark.
Information about their implementation is given in the relevant sections.

For nonparametric methods, we consider the E-Divisive approach of \cite{matteson2014} (R package \verb!ecp!, \cite{JSSv062i07}), the Kolmogorov-Smirnov-based CUSUM procedure (NWBS) of \cite{padilla2021optimal} implemented in the \verb!changepoints! R package \citep{changepoints2022}, the kernel-based method (KCPA) of \cite{celisse2018new} and \cite{arlot2019kernel} (\verb!KernSeg!, \cite{kernseg2018}), and the computationally efficient extension of \cite{zou2014nonparametric} proposed by \cite{haynes2017computationally} (\verb!changepoint.np!, \cite{haynes2021}), referred to as cpt.np. 
In their implementation, we mostly follow the settings recommended by the authors. 
For E-Divisive and cpt.np, we set the minimum segment length to be $30$ and for the former, we use the same settings as NP-MOJO for the number of bootstrap replications $R$ and level $\alpha$. 
For cpt.np, we use the MBIC penalty for declaring change points and $10$ quantiles at which to estimate the cdf.
For KCPA, we use the Gaussian kernel with bandwidth given by the standard deviation, and calculate the penalty using the slope heuristic as recommended in \cite{arlot2019kernel}.
We note that all four methods are developed for detecting changes in the marginal distribution from independent data.
NWBS and cpt.np are univariate methods so their performance is not considered in the multivariate scenarios.
Throughout, $\mbf 0$ denotes a vector of zeros (and analogously for other non-zero values) and $\mbf I$ an identity matrix, whose dimensions are determined by the context. 

\subsection{Size comparison}\label{sec:size}

We assess the performance of NP-MOJO and nonparametric change point methods when there does not exist any change point in the time series. 
Unless stated otherwise, the time series is univariate ($p = 1$) and $\vep_t \sim_{\text{i.i.d.}} \mc N(0, \sigma_{\vep}^2)$ with $\sigma_{\vep} = 1$.
In all scenarios, we set $n = 1000$.

\begin{enumerate}[label=(N\arabic*)] 
\item \label{model-n1} $X_t = \vep_t$.
\item \label{model-n2} $X_t = \vep_t$ where $\vep_t$ are i.i.d. $t_{5}$-distributed random variables.
\item \label{model-n3}  
$X_t = 0.7 X_{t-1} + \vep_t$.
\item \label{model-n4}  
$X_t =  \vep_t + 0.9\vep_{t-1} + 0.8\vep_{t-2} + 0.7 \vep_{t-3} + 0.6 \vep_{t-4}$.
\item \label{model-n5}  
$X_t = \sigma_t \vep_t$ where $\sigma_t^2 =0.5 + 0.4 X_{t-1}^2$.
\item \label{model-n6}  
$X_t = A X_{t-1} + \vep_t$ with $p = 2$ where $\vep_t \sim_{\text{i.i.d.}}\mc N_2(\mbf 0, \mbf I)$  and $A = [A_{ii'}] \in \mathbb{R}^{2 \times 2}$ has $A_{11} = A_{22} = 0.4, A_{12} = A_{21} = -0.2$.
\item \label{model-n7}  
$X_t = A X_{t-1} + \vep_t$ with $p = 5$ where $\vep_t \sim_{\text{i.i.d.}}\mc N_5(\mbf 0,\mbf I)$ and $A = [A_{ii'}] \in \mathbb{R}^{5 \times 5}$ has $A_{ii'} = 0.3^{|i-i'|}$.
\end{enumerate}

Table~\ref{null-table} reports the proportion of realisations where change points are falsely detected. 
The single-lag NP-MOJO controls the size well across all scenarios.
As expected, the multi-lag extension tends to return more spurious estimators but it shows reasonably good size performance.
KCPA does not tend to return spurious estimators even when $\{X_t\}_{t = 1}^n$ is serially correlated.
On the other hand, E-Divisive, NWBS and cpt.np suffer from the presence of temporal dependence as they are calibrated for independent data. 
In the case of cpt.np, it tends to return spurious estimators even when the data is independently generated.

\begingroup
{
\setlength{\tabcolsep}{3pt}
\setlength{\LTcapwidth}{\textwidth}
{\small
\begin{longtable}{c ccccccc}
\caption{Size comparison: we report the size, the proportion of realisations where change points are falsely detected when $q = 0$ out of $1000$ realisations.}
\label{null-table}
\endfirsthead
\endhead
\toprule	
 Size & \multicolumn{7}{c}{Model}  \\ 
 \cmidrule(lr){1-1} 
  \cmidrule(lr){2-8} 
  Method              & \ref{model-n1}     & \ref{model-n2}    &\ref{model-n3}     & \ref{model-n4}    & \ref{model-n5} & \ref{model-n6} & \ref{model-n7}  \\  \cmidrule(lr){1-1} 
  \cmidrule(lr){2-8} 
NP-MOJO-$0$            & 0.043  & 0.050  & 0.123  & 0.104 & 0.064 & 0.045 & 0.021    \\ 
NP-MOJO-$1$         & 0.061  & 0.058 & 0.116  & 0.100 & 0.043 & 0.053 & 0.016    \\ 
NP-MOJO-$2$         & 0.059  & 0.065  & 0.138  & 0.116 & 0.082 & 0.064 & 0.026    \\ 
NP-MOJO-$\mc L$         & 0.114  & 0.114  & 0.172  & 0.140 & 0.125 & 0.089 & 0.033    \\ 
\cmidrule(lr){2-8} 
E-Divisive             & 0.109  & 0.112  & 1.000  & 1.000 & 0.167 & 0.631 & 0.999   \\ 
KCPA              & 0.005  & 0.005   &  0.055 &  0.011 &  0.005 & 0.003 & 0.000  \\
NWBS           & 0.049  & 0.037  & 0.841  & 0.791 & 0.103 & -- & -- 
\\
cpt.np            & 0.286  & 0.313  & 1.000  & 1.000 & 0.695 & -- & -- \\
      \bottomrule
\end{longtable}}}
\endgroup

\subsection{Detection comparison}\label{sec:sim-detect}
\label{sec:power}

We investigate NP-MOJO in its change point detection performance in a variety of change point scenarios. 
Where relevant, we compare NP-MOJO with the relevant parametric change point detection methods, in addition to the nonparametric ones considered in Section~\ref{sec:size}, and their performance serves as a benchmark.

For each scenario, we report the distribution of the error in estimating the number of change points.
For single lag NP-MOJO, this refers to the distribution of $\wh q_\ell - q_\ell$ (recall the definition of $q_\ell$ given in Section~\ref{sec:theory}) over the $1000$ realisations, while for the multi-lag NP-MOJO and other methods, the distribution of $\wh q - q$ is reported.
We also report the covering metric (CM, \citeauthor{arbelaez2010contour}, \citeyear{arbelaez2010contour}) and V-measure (VM, \citeauthor{rosenberg2007v}, \citeyear{rosenberg2007v}) of the segmentation defined by the set of estimated change points.
Let $\mc P = \{ A_j \}_{j = 1}^{q + 1}$ denote the partition of $\{1, \ldots, n\}$ defined by the true change locations $\{\theta_j \}_{j=1}^{q}$, i.e.\ $A_j = \{ \theta_{j-1}+1, \ldots , \theta_j \}$.
Similarly we denote by $\wh{\mc P} = \{\wh{A}_j\}_{j = 1}^{\wh q + 1}$ the partition defined by a set of estimated change points.
Then, CM is defined as
\begin{align*}
\text{CM}(\wh{\mc P}, \mc P) = \frac{1}{n} \sum_{A \in {\mc P}} |A| \max_{\wh{A} \in {\wh{\mc P}}} \left\{ \frac{|A \cap \wh{A}\vert}{\vert A \cup \wh{A}\vert}  \right\},
\end{align*}
and advocated as an evaluation metric for comparing change point detection algorithms \citep{van2020evaluation}.
VM is similarly calculated using the conditional entropy of the resulting segmentation. 
Both the CM and VM take values between $0$ and $1$, with a value of $1$ indicating a perfect segmentation. For each measure, we report its average over the $1000$ realisations.

\subsubsection{Changes in mean}
\label{sec:sim:mean}

We generate time series under the model
\begin{align}
\label{eq:mean:model}
X_t = \sum_{j = 0}^q \mu_j \mathbb{I}{\{\theta_j + 1 \le t \le \theta_j\}} + \vep_t, \, 1 \le t \le n, 
\end{align}
with $n = 1000$, $q = 3$, $(\theta_1, \theta_2, \theta_3) = (250, 500, 750)$ and $(\mu_0, \mu_1, \mu_2, \mu_3) = (0, {1}, 0, { 1})$.
The error sequence $\{\vep_t\}_{t = 1}^n$ is simulated according to models \ref{model-n1}--\ref{model-n4} from Section~\ref{sec:size}, and then is standardised such that $\Var(\vep_t) = 1$; we refer to the corresponding scenarios as (A1)--(A4).
To these scenarios, in addition to the nonparametric methods considered in Section~\ref{sec:size}, we apply the pruned exact linear time (PELT) method \citep{killick2012} implemented in the \verb!changepoint! R package \citep{killick2014} and WCM.gSa \citep{cho2021multiple} implemented in \cite{breakfast}.
While both detect multiple mean shifts in univariate time series, PELT is proposed for independent data while WCM.gSa handles autocorrelations under an AR model.

In addition, we consider a multivariate scenario:
\begin{enumerate}[label=(A5)] 
\item \label{model-a5} Setting $p = 10$, $X_t$ follows~\eqref{eq:mean:model} with $\vep_t \sim_{\text{i.i.d.}} \mc N_{10}(\mbf 0, \mbf I)$ and $(\mu_0, \mu_1, \mu_2, \mu_3) = (\mbf 0, \bm\Delta, \mbf 0, \bm\Delta)$, where $\bm\Delta$ has its first $5$ coordinates set to $0.5$ and the rest to $0$. 
\end{enumerate}
 
The results are reported in Table~\ref{multcpt-a}.
In general, NP-MOJO accurately detects the number and locations of change points across all scenarios regardless of the choice of the lag, as the changes in the mean are detectable at all lags.
In the independent settings (A1) and (A2), its performance is comparable to PELT
while in the presence of serial dependence under (A3) and (A4), it performs as well as WCM.gSa.
Among the nonparametric methods, NP-MOJO and KCPA outperform E-Divisive, NWBS and cpt.np
and NP-MOJO tends to perform better than KCPA, either marginally or significantly, particularly in the multivariate setting in~\ref{model-a5}.
As noted in Section~\ref{sec:size}, E-Dvisive, NWBS and cpt.np suffer from the departure from the independence assumption.

\begingroup
{\small
\setlength{\tabcolsep}{3pt}
\setlength{\LTcapwidth}{\textwidth}
\begin{longtable}{c c ccccc cc}
\caption{(A1)--\ref{model-a5}: we report the distribution of the estimated number of change points and the average CM and VM over 1000 realisations.
The modal value of $\wh q - q$ in each row is given in bold.
Also, the best performance for each metric is underlined for each scenario.}
\label{multcpt-a}
\endfirsthead
\endhead
\toprule	
&& \multicolumn{5}{c}{$\wh{q} - q$} &   &   \\ 
 Model       & Method      &  $\leq-2$     & $-1$        & $\mathbf{0}$  & $1$    & $\geq 2$    & CM & VM       \\ 
 \cmidrule(lr){1-2} \cmidrule(lr){3-7} \cmidrule(lr){8-9}
(A1)    &  NP-MOJO-$0$   &  0.000      & 0.019   &  \textbf{0.976}    & 0.005       &    0.000          & 0.958  & 0.942    \\
&  NP-MOJO-$1$     &  0.000    & 0.003     & \textbf{0.997}   & 0.000          &  0.000          & 0.971   & 0.955      \\
      &  NP-MOJO-$2$   &  0.000      &    0.002   & \textbf{0.997}  & 0.001     &   0.000          &  0.971 & 0.955    \\ 
            &  NP-MOJO-$\mc L$   &  0.000      &    0.001   & \textbf{0.999}  & 0.000     &   0.000          &  0.970 & 0.953    \\ \cmidrule(lr){3-7} \cmidrule(lr){8-9}
      &  E-Divisive      &  0.000      &   0.000    & \textbf{0.912}  & 0.070        &    0.018         & 0.975 & 0.965   \\
      &  KCPA      &  0.000      & 0.000      & \textbf{0.971}  &  0.028      &     0.001        &  0.977  & 0.963 \\      
      &  NWBS      &  0.000      &   0.000    &  \textbf{0.955} &   0.028     &    0.017        &  0.971  & 0.956 \\
      &  cpt.np      & 0.000       &  0.000     &  \textbf{0.788}  & 0.184       &   0.028          & 0.964   & 0.955 \\ 
      &  PELT      &  0.000      & 0.000      & \underline{\textbf{1.000}}  &   0.000     &  0.000           & \underline{0.983}   & \underline{0.970} \\
      &  WCM.gSa      &  0.000      &  0.000     &  \textbf{0.972}  & 0.021        &    0.007         & 0.980   & 0.969 \\ \cmidrule(lr){1-2}    \cmidrule(lr){3-7} \cmidrule(lr){8-9}
(A2)    &  NP-MOJO-$0$ &  0.000      & 0.002    &  \textbf{0.998}    &  0.000       & 0.000             &  0.974 &  0.958      \\
&  NP-MOJO-$1$  & 0.000     & 0.000     & \underline{\textbf{1.000}}   & 0.000          &   0.000         & 0.977   & 0.962        \\
      &  NP-MOJO-$2$ &   0.000     &  0.000     & \textbf{0.999}  & 0.001        &   0.000          & 0.976  & 0.961      \\
       &  NP-MOJO-$\mc L$   &  0.000      &    0.000   & \underline{\textbf{1.000}}  & 0.000     &   0.000          &  0.976 & 0.961    \\  \cmidrule(lr){3-7} \cmidrule(lr){8-9}
      &  E-Divisive      &  0.000      &   0.000    &  \textbf{0.913}  &    0.058    &   0.029          & 0.977 &  0.969 \\
      &  KCPA      &  0.000      &   0.000    & \textbf{0.978} &  0.021      & 0.001            & \underline{0.983}   &  \underline{0.972} \\      
      &  NWBS      &  0.000      & 0.000      & \textbf{0.970}  &   0.015     &   0.015          &  0.979  & 0.967 \\
      &  cpt.np      &  0.000      &  0.000     & \textbf{0.739}  & 0.206        &  0.055           & 0.960   & 0.954 \\ 
      &  PELT      &  0.000      &  0.000     & \underline{\textbf{1.000}}  &    0.000    &    0.000         & \underline{0.983}   & 0.970 \\
      &  WCM.gSa      &  0.000      &  0.000     & \textbf{0.973}  & 0.016        &     0.011        &  0.980  & 0.969\\     \cmidrule(lr){1-2} \cmidrule(lr){3-7} \cmidrule(lr){8-9}
(A3)     &  NP-MOJO-$0$    &    0.000    &  0.000  &  {\textbf{0.999}}    &  0.001      &    0.000          &   \underline{0.986} & 0.978  \\
&  NP-MOJO-$1$  &  0.000    &  0.000    & \textbf{0.997}   &    0.003     &  0.000          & 0.984   & 0.974        \\
      &  NP-MOJO-$2$  & 0.000       &  0.000     & {\textbf{0.997}}  &  0.003      &  0.000           & 0.984  & 0.973     \\
       &  NP-MOJO-$\mc L$   &  0.000      &    0.000   & \underline{\textbf{1.000}}  & 0.000     &   0.000          &  0.984 & 0.975    \\ \cmidrule(lr){3-7} \cmidrule(lr){8-9}
      &  E-Divisive      &  0.000      &   0.000    & 0.001  & 0.000       &  \textbf{0.999}           &  0.413 & 0.675  \\
      &  KCPA      &  0.000      &  0.000     &  \textbf{0.724} &    0.151            &  0.125  & 0.959 & 0.962 \\      
      &  NWBS      &  0.000      &   0.000    & 0.000   & 0.000        &  \textbf{1.000}           & 0.438   & 0.662 \\
      &  cpt.np      &   0.000     &  0.000     & 0.002  &  0.009      &      \textbf{0.989}       & 0.655   & 0.779 \\ 
      &  PELT      & 0.000       &   0.000    & 0.233  &  0.244      &      \textbf{0.523}       &  0.885  & 0.914 \\
      &  WCM.gSa      &    0.000    &   0.000    & \textbf{0.949}  &  0.027     &  0.024           &  0.985  & \underline{0.981} \\    \cmidrule(lr){1-2} \cmidrule(lr){3-7} \cmidrule(lr){8-9}
(A4)    &  NP-MOJO-$0$  &  0.000      &  0.000  &   {\textbf{0.996}}   &    0.004    &  0.000            & 0.980  &  0.969     \\
&  NP-MOJO-$1$  &  0.000    & 0.000     & {\textbf{0.997}}   & 0.003          &  0.000          &  0.978  & 0.966         \\
      &  NP-MOJO-$2$ &  0.000    & 0.000     & \textbf{0.997}   & 0.003          &  0.000          &  0.977  & 0.964      \\ 
       &  NP-MOJO-$\mc L$   &  0.000      &    0.000   & \underline{\textbf{1.000}}  & 0.000     &   0.000          &  0.979 & 0.967    \\ \cmidrule(lr){3-7} \cmidrule(lr){8-9}
      &  E-Divisive      &   0.000     &   0.000    &  0.000 &  0.000      &  \textbf{1.000}           &  0.416 & 0.674  \\
      &  KCPA      &  0.000      &  0.000     &  \textbf{0.910}  &    0.062    &    0.028         & 0.978   & 0.972 \\      
      &  NWBS      &  0.000      &  0.000     & 0.001  &  0.000      &  \textbf{0.999}           & 0.437   &  0.658\\
      &  cpt.np      &   0.000     & 0.000     & 0.000  & 0.006       &     \textbf{0.994}        & 0.642   & 0.769 \\ 
      &  PELT      &   0.000     &  0.000     & 0.309  &   0.272     &  \textbf{0.419}           & 0.905   & 0.923 \\
      &  WCM.gSa      &  0.000      &  0.000     &  \textbf{0.987}  &  0.010      &  0.003           & \underline{0.985}   & \underline{0.977} \\ \cmidrule(lr){1-2} \cmidrule(lr){3-7} \cmidrule(lr){8-9} 
          {\ref{model-a5}}     &  NP-MOJO-$0$  &  0.001      &  0.013  &   {\textbf{0.986}}   &    0.006    &  0.000            & 0.971  &  0.957     \\
&  NP-MOJO-$1$  &  0.000    & 0.005     & \textbf{0.995}   & 0.000         &  0.000          &  {0.976}  & {0.962}         \\
      &  NP-MOJO-$2$ &  0.000    & 0.004     & {\textbf{0.996}}   & 0.000          &  0.000          &  {0.976} & {0.961}      \\ 
       &  NP-MOJO-$\mc L$   &  0.000      &    0.003   & \underline{\textbf{0.997}}  & 0.000     &   0.000          &  0.975 & 0.961    \\ \cmidrule(lr){3-7} \cmidrule(lr){8-9}
      &  E-Divisive      &   0.000     &   0.000    &  \bf{0.913} &  0.072      &  0.015          &  \underline{0.978} & \underline{0.969} \\
      &  KCPA      &  \bf{1.000}      &  0.000     &  0.000  &    0.000    &    0.000         & 0.250   & 0.000 \\
      \bottomrule
\end{longtable}}
\endgroup

\subsubsection{Changes in second-order moments}
\label{sec:second:mean}

We first consider the scenarios where $X_t$ undergoes changes in variance or covariance which are detectable at all lags, with $n = 1000$, $q = 3$ and $(\theta_1, \theta_2, \theta_3) = (250, 500, 750)$.
\begin{enumerate}[label=(B\arabic*)] 
\item \label{model-b1} $X_t = \sum_{j = 0}^q \sigma_j \mathbb{I}{\{\theta_j + 1 \le t \le \theta_{j + 1}\}} \cdot \vep_t$ where $\vep_t \sim_{\text{i.i.d.}} \mc N(0, 1)$ and
$(\sigma_0, \sigma_1, \sigma_2, \sigma_3) = (0.5, 1, 0.5, 1)$.

\item \label{model-b2}  $X_t = \sum_{j = 0}^q \sigma_j \mathbb{I}{\{\theta_j + 1 \le t \le \theta_{j + 1}\}} \cdot \vep_t$ with $\sigma_j$ chosen as in~\ref{model-b1} where $\vep_{t} \sim_{\text{i.i.d.}} t_5/\sqrt{5/3}$.
\item \label{model-b3}  $X_t = 0.4X_{t-1} + \sum_{j = 0}^q \sigma_j \mathbb{I}{\{\theta_j + 1 \le t \le \theta_{j + 1}\}} \cdot \vep_t$ where $\vep_t \sim_{\text{i.i.d.}} \mc N(0, 1)$ and $(\sigma_0, \sigma_1, \sigma_2, \sigma_3) = (1, 2, 1, 2)$.
\item \label{model-b4} $X_t = \sum_{j = 0}^q \Sigma_j^{1/2} \mathbb{I}{\{\theta_j + 1 \le t \le \theta_{j + 1}\}} \cdot \vep_t$ with $p = 2$, where $\vep_t \sim_{\text{i.i.d.}} \mc N_2(\mbf 0, \mbf I)$, $\Sigma_0 = \Sigma_2 = \mbf I$ and $\Sigma_1 = \Sigma_3 = \begin{psmallmatrix}1 & 0.9\\ 0.9 & 1\end{psmallmatrix}$.
\item \label{model-b5} As in~\ref{model-b4} where $\vep_t = (\vep_{1t}, \vep_{2t})^\top$ generated with $\vep_{it} \sim_{\text{i.i.d.}} t_5$. 
\item \label{model-b6} \tcb{$X_t = \sum_{j = 0}^q \Sigma_j^{1/2} \mathbb{I}{\{\theta_j + 1 \le t \le \theta_{j + 1}\}} \cdot \vep_t$ with $p = 5$, where $\vep_t \sim_{\text{i.i.d.}} \mc N_{10}(\mbf 0, \mbf I)$, $\Sigma_0 = \Sigma_2 = \mbf I$ and $\Sigma_1 = \Sigma_3$ has $i, j$-th entry given by $0.7^{|i-j|}$.}
\end{enumerate}

In addition to the nonparametric competitors, we consider the wavelet-based WBS approach (WBSTS) of \cite{korkas2017multiple}, implemented in the R package \verb!wbsts! \citep{wbsts2020}, when $p = 1$, and the sparsified binary segmentation (SBS) \citep{cho2015}, implemetend in the R package \verb!hdbinseg! \citep{hdbinseg2018} when $p > 1$, both of which are developed for detecting changes in the second-order structure of time series. 
The results are reported in Table~\ref{multcpt-b}. 
NP-MOJO consistently outperforms the competing nonparametric methods in all metrics. 
It is competitive with WBSTS and SBS which specifically seek changes in the second-order structure and in fact, NP-MOJO performs better in estimating $q$ when the data is non-Gaussian in model~\ref{model-b2}.

\begingroup
{\small
\setlength{\tabcolsep}{3pt}
\setlength{\LTcapwidth}{\textwidth}
\begin{longtable}{c c ccccc cc}
\caption{\ref{model-b1}--\ref{model-b5}: we report the distribution of the estimated number of change points and the average CM and VM over 1000 realisations.
The modal value of $\wh q - q$ in each row is given in bold.
Also, the best performance for each metric is underlined for each scenario.}
\label{multcpt-b}
\endfirsthead
\endhead
\toprule	
&& \multicolumn{5}{c}{$\wh{q} - q$} &   &   \\ 
 Model       & Method      &  $\leq-2$     & $-1$        & $\mathbf{0}$  & $1$    & $\geq 2$    & CM & VM       \\ 
 \cmidrule(lr){1-2} \cmidrule(lr){3-7} \cmidrule(lr){8-9}
 \ref{model-b1}    &  NP-MOJO-$0$      & 0.000       &  0.052  & \bf{0.930}    & 0.017    &     0.001          & 0.942  & 0.928  \\
&  NP-MOJO-$1$        &  0.000   &   0.008  & \bf{0.986}  &        0.006 & 0.000      & 0.965 & \underline{0.949}   \\
      &  NP-MOJO-$2$      & 0.000 &  0.008  & {\bf{0.988}}  & 0.004        &  0.000              &  \underline{0.966} & \underline{0.949} \\ 
       &  NP-MOJO-$\mc L$   &  0.000      &    0.006   & \underline{\textbf{0.994}}  & 0.000     &   0.000          &  0.965 & 0.948    \\ \cmidrule(lr){3-7} \cmidrule(lr){8-9}
      &  E-Divisive      &  0.003     &  0.008  & \bf{0.896} &    0.069   & 0.024 & 0.946 & 0.934 \\
      &  KCPA      &  0.007    &  0.000  & \bf{0.955} & 0.033       &  0.005      & {0.965}  & \underline{0.949} \\      
      &  NWBS      & \bf{0.429}     & 0.093   & 0.364 &   0.089   &   0.025    & 0.616   & 0.558 \\
      &  cpt.np      & 0.000     & 0.000   &  \bf{0.676} &     0.214   &  0.110   & 0.943 & 0.936 \\ 
      &  WBSTS      &   0.000 &  0.000    & {\bf{0.978}} &     0.021  &  0.001        &  0.960 & 0.941 \\  \cmidrule(lr){1-2} \cmidrule(lr){3-7} \cmidrule(lr){8-9} 
                \ref{model-b2}  &  NP-MOJO-$0$      &  0.005      & 0.133  & \bf{0.839}     &  0.023    & 0.000           & 0.912 & 0.905 \\
&  NP-MOJO-$1$        &  0.000    &  0.044   & {\bf{0.945}}    & 0.011 & 0.000 & {0.944} & {0.929} \\
      &  NP-MOJO-$2$      & 0.000     & 0.033    & {\bf{0.956}}    & 0.011     & 0.000   & {0.945}  & {0.929} \\  
      &  NP-MOJO-$\mc L$   &  0.000      &    0.012   & \underline{\textbf{0.988}}  & 0.000     &   0.000          &  \underline{0.950} & \underline{0.932}    \\ \cmidrule(lr){3-7} \cmidrule(lr){8-9}
      &  E-Divisive             &  0.035  & 0.039 & \bf{0.814} &  0.096  &0.016      & 0.910 & 0.902 \\
      &  KCPA      &  0.100     & 0.003   & \bf{0.863} & 0.032     &   0.002      & 0.904  & 0.882  \\      
      &  NWBS      & \bf{0.559}    &  0.136  &0.212   & 0.064    & 0.029    & 0.510 & 0.423 \\
      &  cpt.np      & 0.001  &  0.000 & \bf{0.615}  & 0.269   & 0.115   & 0.924  &0.915  \\ 
      &  WBSTS     &  0.000 & 0.002  & \bf{0.693}  & 0.230    & 0.075  & 0.905  & 0.894 \\  \cmidrule(lr){1-2} \cmidrule(lr){3-7} \cmidrule(lr){8-9}  
              \ref{model-b3}    &  NP-MOJO-$0$      &  0.025      & 0.121  & \bf{0.840}     &  0.014    & 0.000           & 0.905 & 0.899 \\
&  NP-MOJO-$1$        &  0.000    &  0.024   & \bf{0.962}    & 0.014   & 0.000 & 0.953 & \underline{0.937} \\
      &  NP-MOJO-$2$      & 0.000     & 0.035    & \bf{0.953}    & 0.012     & 0.000   &0.949  & 0.934 \\ 
      &  NP-MOJO-$\mc L$   &  0.000      &    0.013   & \underline{\textbf{0.987}}  & 0.000     &   0.000          &  0.953 & 0.936    \\ \cmidrule(lr){3-7} \cmidrule(lr){8-9}
      &  E-Divisive             &  0.000  & 0.000 & 0.148 &  0.178  & \bf{0.674}      & 0.774 & 0.813 \\
      &  KCPA      &  0.163     & 0.004   & \bf{0.739} & 0.071     &   0.023      & 0.858  & 0.833  \\      
      &  NWBS      & 0.085    &  0.036  & 0.110   & 0.118    & \bf{0.651}    & 0.657 & 0.700 \\
      &  cpt.np      & 0.000  &  0.000 & 0.046  & 0.105   & \bf{0.849}   & 0.789  & 0.831 \\ 
      &  WBSTS     &  0.000 & 0.000  & {\bf{0.979}}  & 0.021    & 0.000  & \underline{0.954}  & {0.934} \\ \cmidrule(lr){1-2} \cmidrule(lr){3-7} \cmidrule(lr){8-9}  
        \ref{model-b4}    &  NP-MOJO-$0$      &  0.000      & 0.000  & \underline{\bf{1.000}}     &  0.000    & 0.000           & \underline{0.981} & \underline{0.967} \\
&  NP-MOJO-$1$        &  0.000    &  0.031   & \bf{0.963}    & 0.006   & 0.000 & 0.965 & 0.953 \\
      &  NP-MOJO-$2$      & 0.000     & 0.015    & \bf{0.976}    & 0.009     & 0.000   & 0.969  & 0.955 \\ 
       &  NP-MOJO-$\mc L$   &  0.000      &    0.000   & \underline{\textbf{1.000}}  & 0.000     &   0.000          &  0.979 & 0.965    \\ \cmidrule(lr){3-7} \cmidrule(lr){8-9}
      &  E-Divisive       & \bf{0.529}       &  0.168   & 0.256   & 0.032   &  0.015         & 0.557 & 0.506  \\
      &  KCPA      &  0.077      & 0.000   & \bf{0.909} & 0.014     &   0.000       & 0.935  & 0.915 \\  
      &  SBS    & 0.044 & 0.000 &  {\bf{0.942}} & 0.014    & 0.000  & {0.949} & {0.939} \\ 
      \cmidrule(lr){1-2} \cmidrule(lr){3-7} \cmidrule(lr){8-9} 
      \ref{model-b5}    &  NP-MOJO-$0$      & 0.000       & 0.001  &   {\bf{0.997}}   &   0.002   & 0.000            & \underline{0.974} & \underline{0.959}  \\
&  NP-MOJO-$1$        & 0.005    & 0.121    & \bf{0.867}     & 0.007   &  0.000 & 0.931 & 0.927 \\
      &  NP-MOJO-$2$      &   0.006  & 0.103   & \bf{0.884}   &   0.007  &  0.000  & 0.935 & 0.929 \\  &  NP-MOJO-$\mc L$   &  0.000      &    0.001   & \underline{\textbf{0.999}}  & 0.000     &   0.000          &  0.973 & 0.958    \\ \cmidrule(lr){3-7} \cmidrule(lr){8-9}
      &  E-Divisive      &  \bf{0.670}      &  0.189   &  0.101  & 0.032   &  0.008         & 0.431 & 0.335 \\
      &  KCPA      &  0.322      &  0.000  &  {\bf{0.662}}     &   0.015        &  0.001 & 0.775 & 0.725 \\      
      &  SBS     & \bf{0.614} &  0.003 &  0.377 & 0.006    & 0.000  & 0.653  & 0.660 \\
            \cmidrule(lr){1-2} \cmidrule(lr){3-7} \cmidrule(lr){8-9} 
      \tcb{\ref{model-b6}}    &  NP-MOJO-$0$      & 0.000       & 0.003  &   \underline{\bf{0.997}}   &   0.000   & 0.000            & \underline{0.982} & \underline{0.970}  \\
&  NP-MOJO-$1$        & 0.169    & 0.347    & \bf{0.482}     & 0.002   &  0.000 & 0.801 & 0.829 \\
      &  NP-MOJO-$2$      &   0.136  & 0.343   & \bf{0.517}   &   0.004  &  0.000  & 0.817 & 0.844 \\  &  NP-MOJO-$\mc L$   &  0.000      &    0.003   & \underline{\textbf{0.997}}  & 0.000     &   0.000          &  0.981 & 0.969    \\ \cmidrule(lr){3-7} \cmidrule(lr){8-9}
      &  E-Divisive      &  \textbf{0.552}      &  0.237   &  0.154  & 0.040   &  0.017         & 0.493 & 0.445 \\
      &  KCPA      &  0.262      &  0.000  &  {\bf{0.736}}     &   0.002  &  0.000 & 0.820 & 0.779 \\      
      &  SBS     & 0.321 &  0.001 &  \textbf{0.660} & 0.018    & 0.000  & 0.806  & 0.812 \\
      \bottomrule
\end{longtable}}
\endgroup

\subsubsection{Changes in temporal dependence}
\label{sec:temporal}

We consider the scenarios where the autocorrelations or the (conditional) variance 
of the data change. Unless stated otherwise $q = 2$, $(\theta_1, \theta_2) = (333, 667)$ and $\vep_t \sim_{\text{i.i.d.}} \mc N(0, 1)$.
\begin{enumerate}[(C1)] 
\item \label{model-c1} $X_t = X^{(j)}_t = a_j  X^{(j)}_{t-1} + \vep_t$ for $\theta_j + 1 \le t \le \theta_{j + 1}$, where $(a_0, a_1, a_2) = (-0.8, 0.8, -0.8)$.
\item \label{model-c2} $X_t = \vep_t + \sum_{j = 0}^q b_j \mathbb{I}{\{\theta_j + 1 \le t \le \theta_{j + 1}\}} \cdot \vep_{t-2}$, where $(b_0, b_1, b_2) = (-0.7, 0.7, -0.7)$.
\item \label{model-c3} $X_t = X^{(j)}_t = \sigma^{(j)}_t \vep_t$ with $(\sigma^{(j)}_t)^2 = \omega_j + \alpha_j (X^{(j)}_{t - 1})^2 + \beta_j (\sigma^{(j)}_{t - 1})^2 )$ for $\theta_j + 1 \le t \le \theta_{j + 1}$, $q = 1$, $\theta_1 = 500$, $(\omega_0, \alpha_0, \beta_0) = (0.01, 0.7, 0.2)$ and $(\omega_1, \alpha_1, \beta_1) = (0.01, 0.2, 0.7)$.
\item \label{model-c4} $X_t = X^{(j)}_t = A_j X^{(j)}_{t-1} + \vep_t$ for $\theta_j + 1 \le t \le \theta_{j + 1}$,  where $A_0 = A_2 = \begin{psmallmatrix}0.5 & 0.1\\ 0.1 & 0.5\end{psmallmatrix}$ and $A_1 = \begin{psmallmatrix}-0.5 & 0.1\\ 0.1 & -0.5\end{psmallmatrix}$, $\vep_t \sim_{\text{i.i.d.}} \mc N_2(\mbf 0, \mbf I)$.
\item \label{model-c5} $X_t = \vep_t + \sum_{j = 0}^q B_j \mathbb{I}{\{\theta_j + 1 \le t \le \theta_{j + 1}\}}  \vep_{t-1} $, where $B_0 = B_2 = \begin{psmallmatrix}1 & 0.1 \\ 0.1 & 1\end{psmallmatrix}$ and $B_1 = \begin{psmallmatrix}-1 & 0.1\\ 0.1 & -1\end{psmallmatrix}$, $\vep_t \sim_{\text{i.i.d.}} \mc N_2(\mbf 0, \mbf I)$.
\item \tcb{\label{model-c6}As in~\ref{model-c5}, except $p=5$ (diagonal entries of $B_0$ and $B_2$ given by $1$, diagonal entries of $B_1$ equal to $-1$, all off-diagonal entries given by $0.1$). }
\end{enumerate}
Model \ref{model-c1} was studied in \cite{korkas2017multiple}, while models similar to \ref{model-c4} and \ref{model-c5} were considered in \cite{preuss2015detection}. 
In all models, except \ref{model-c3}, changes are present only in the joint distribution of $X_t$ and its lagged values. Therefore, we exclude the nonparametric methods considered in Section~\ref{sec:size} which have detection power against changes in marginal distribution only. Specifically, $(q_0, q_1, q_2) = (0, 0, 2)$ in~\ref{model-c2} and $(q_0, q_1, q_2) = (0, 2, 0)$ in~\ref{model-c4}, ~\ref{model-c5}, and~\ref{model-c6}. Accordingly, in reporting the results returned by NP-MOJO-$\ell$ for $\ell = 0, 1, 2$, we report the distribution of $\wh q_\ell - q_\ell$ and report CM and VM for NP-MOJO-$\ell$ with $q_\ell = 2$ only, see Table~\ref{multcpt-c}. 
In scenarios~\ref{model-c1}--~\ref{model-c5}, we observe that NP-MOJO performs similarly or superior to the competing method in both detection and estimation accuracy. In scenarios~\ref{model-c1}, SBS outperforms NP-MOJO, likely as it is designed for detection of change points in high dimensions. As expected, we do not detect all $q$ change points from NP-MOJO-$\ell$ for which $q_\ell < q$, but the multi-lag extension successfully aggregates the estimators from multiple lags.

\begingroup
{\small
\setlength{\tabcolsep}{3pt}
\setlength{\LTcapwidth}{\textwidth}
\begin{longtable}{c c ccccc cc}
\caption{\ref{model-c1}--\ref{model-c5}: we report the distribution of the estimated number of change points and the average CM and VM over 1000 realisations.
The modal value of $\wh q - q$ in each row is given in bold.
Also, the best performance for each metric is underlined for each scenario.}
\label{multcpt-c}
\endfirsthead
\endhead
\toprule	
&& \multicolumn{5}{c}{$\wh{q} - q$ / $\wh{q}_\ell - q_\ell$} &   &   \\ 
 Model       & Method      &  $-2$     & $-1$        & $\mathbf{0}$  & $1$    & $\geq 2$    & CM & VM       \\ 
 \cmidrule(lr){1-2} \cmidrule(lr){3-7} \cmidrule(lr){8-9}
      \ref{model-c1}    &  NP-MOJO-$0$     &  --     &  -- & \bf{0.851}   & 0.140     &     0.009          &  -- & --   \\
&  NP-MOJO-$1$        &  0.000   &   0.002  & \bf{0.956}  &        0.042 & 0.000      & {0.978} & {0.961}   \\
     &  NP-MOJO-$2$      & -- &   -- & \bf{0.836}  & 0.149        &  0.015              & --   & -- \\ 
      &  NP-MOJO-$\mc L$   &  0.000      &    0.002   & \underline{\textbf{0.986}}  & 0.012     &   0.000          &  \underline{0.980} & \underline{0.963}    \\ \cmidrule(lr){3-7} \cmidrule(lr){8-9}
     &  WBSTS      &   0.000 &  0.000    & {\bf{0.414}} &     0.299  &  0.287        &  0.904 & 0.900 \\  \cmidrule(lr){1-2} \cmidrule(lr){3-7} \cmidrule(lr){8-9} 
       \ref{model-c2}    &  NP-MOJO-$0$      &   --     & --   & \textbf{0.952}    & 0.047 &  0.001          & -- & --  \\
&  NP-MOJO-$1$        & --    & --    & \textbf{0.930}  &    0.068 & 0.002     & -- &  --  \\
      &  NP-MOJO-$2$      & 0.001  &  0.054  &  \bf{0.908}  & 0.036        &  0.001              &   {0.949} & \underline{0.926} \\ 
      &  NP-MOJO-$\mc L$   &  0.001      &    0.051   & \underline{\textbf{0.942}}  & 0.006     &   0.000          &  \underline{0.950} & \underline{0.926}    \\
      \cmidrule(lr){3-7} \cmidrule(lr){8-9}
      &  WBSTS      &   0.007 &  0.021    & {\bf{0.899}} &     0.062  &  0.011        &  0.896 & 0.852 \\  \cmidrule(lr){1-2} \cmidrule(lr){3-7} \cmidrule(lr){8-9} 
      \ref{model-c3}    &  NP-MOJO-$0$  &  --      &  0.409  & \textbf{0.533}    & 0.056     &     0.002          & 0.744  & 0.484       \\
&  NP-MOJO-$1$  & --    &   0.236  &  {\textbf{0.682}}  &  0.081  &    0.001 & 0.819 &  0.633         \\
     &  NP-MOJO-$2$  & -- &  0.299  & \textbf{{0.626}}  & 0.073       &  0.002    &  {0.787} &{0.571}     \\
     &  NP-MOJO-$\mc L$ &  --      &    0.210   & \underline{\textbf{0.727}}  & 0.062     &   0.001          &  \underline{0.823} & \underline{0.645}      \\ \cmidrule(lr){3-7} \cmidrule(lr){8-9}
     &  WBSTS   & -- & 0.003 & 0.025 & 0.054 & \textbf{0.918} & 0.662 & 0.487    \\  \cmidrule(lr){1-2} \cmidrule(lr){3-7} \cmidrule(lr){8-9} 
       \ref{model-c4}    &  NP-MOJO-$0$      & --  &  --  & \textbf{0.904}    & 0.090    &   0.006          &  -- &  -- \\
&  NP-MOJO-$1$        &  0.004   &   0.159  & \bf{0.783}  &   0.051 & 0.003      & \underline{0.907} & \underline{0.893}   \\
      &  NP-MOJO-$2$      & -- &  --  & \textbf{{0.888}}  & 0.107   &  0.005 & --  & --\\
      &  NP-MOJO-$\mc L$   &  0.004      &    0.165   &  {\textbf{0.818}}  & 0.013     &   0.000          &  \underline{0.907} &  {0.891}    \\ \cmidrule(lr){3-7} \cmidrule(lr){8-9}
      &  SBS      &   0.070 &  0.000    & \underline{\bf{0.911}} & 0.019  &  0.000        &  0.903 & 0.875 \\  \cmidrule(lr){1-2} \cmidrule(lr){3-7} \cmidrule(lr){8-9} 
       \ref{model-c5}    &  NP-MOJO-$0$      & --      & --   & \textbf{0.939}    & 0.058   & 0.003    & --  & --  \\
&  NP-MOJO-$1$        &  0.000   &   0.011  & \bf{0.952}  &  0.035 & 0.002      &  {0.974} & \underline{0.957}   \\
      &  NP-MOJO-$2$      & -- & --   & \textbf{0.926}  & 0.073    &  0.001             & --  & -- \\ 
      &  NP-MOJO-$\mc L$   &  0.000      &    0.012   & \underline{\textbf{0.979}}  & 0.009     &   0.000          &  \underline{0.976} & \underline{0.957}    \\ \cmidrule(lr){3-7} \cmidrule(lr){8-9}
      &  SBS      &   0.006 &  0.000    &  {\bf{0.961}} &     0.033  &  0.000        &  0.967 & 0.942 \\  \cmidrule(lr){1-2} \cmidrule(lr){3-7} \cmidrule(lr){8-9} 
       \ref{model-c6}    &  NP-MOJO-$0$      & --      & --   & \textbf{0.969}    & 0.031   & 0.000    & --  & --  \\
&  NP-MOJO-$1$        &  0.048   &   0.299  & \bf{0.636}  &  0.017 & 0.000      &  {0.851} & {0.838}   \\
      &  NP-MOJO-$2$      & -- & --   & \textbf{0.960}  & 0.040    &  0.000             & --  & -- \\ 
      &  NP-MOJO-$\mc L$   &  0.047      &    0.299   & {\textbf{0.649}}  & 0.005     &   0.000          &  {0.850} & {0.836}    \\ \cmidrule(lr){3-7} \cmidrule(lr){8-9}
      &  SBS      &   0.000 &  0.000    &  \underline{\bf{0.889}} &     0.105  &  0.006        &  \underline{0.977} & \underline{0.964} \\ 

      \bottomrule
\end{longtable}}
\endgroup

\subsubsection{Changes in higher-order moments}\label{sec:sims-high}

We simulate scenarios where there are changes in stochastic properties beyond the first two moments. 
In what follows, we have $q = 2$ and $(\theta_1, \theta_2) = (333, 667)$.
\begin{enumerate}[(D1)] 
\item \label{model-d1} $X_t \sim_{\text{i.i.d.}} \mc N(0, 1)$ for $t \leq \theta_1$ and $t \geq \theta_2 + 1$, and $X_t \sim_{\text{ i.i.d.}} t_{2.5}/ \sqrt{5}$ for $\theta_1 + 1 \leq t \leq \theta_2$.
\item \label{model-d2} $X_t \sim_{\text{ i.i.d.}} 0.5 + (\chi^2_1 - 1)/ 2\sqrt{2}$ for $t \leq \theta_1$ and $t \geq \theta_2 + 1$, and $X_t \sim_{\text{ i.i.d.}} \mc N(0.5, 0.5^2)$ for $\theta_1 + 1 \leq t \leq \theta_2$.
\item \label{model-d3} $X_t = 0.4X_{t-1} + \vep_t$ where $\vep_t \sim_{\text{i.i.d.}} \mc N(0, 0.5^2)$ for $t \leq \theta_1$ and $t \geq \theta_2 + 1$, and $ \vep_t \sim_{\text{i.i.d.}} \text{Exponential}(0.5) - 0.5$ for $\theta_1 + 1 \leq t \leq \theta_2$.
\item \label{model-d4} \tcb{$X_t \sim_{\text{i.i.d.}} \mc N_{10}(\mbf{0}, \mbf{I})$ for $t \leq \theta_1$ and $t \geq \theta_2 + 1$, and $X_t = (Y_t, Z_t)$, $Y_t \sim_{\text{i.i.d.}} \mc N_{3}(\mbf{0}, \mbf{I})$, $Z_t = (Z_{1t}, \ldots, Z_{7t})^\mathsf{T}$ where  $Z_{it} \sim_{\text{i.i.d.}}  t_{2.5}/\sqrt{5}$ for $\theta_1 + 1 \leq t \leq \theta_2$.}
\end{enumerate}

Model~\ref{model-d1} is taken from \cite{padilla2021optimal}, where $\mathbb{E}(X_t) = 0$ and $\Var(X_t) = 1$ for all $t$ and changes occur in the tail of the distribution. 
Model~\ref{model-d2} is a variation of a scenario studied in \cite{arlot2019kernel}, where $\mathbb{E}(X_t) = 0.5$ and $\Var(X_t) = 0.25$ for all $t$ with changes in the tail behaviour. 
Model~\ref{model-d3} considers changes in higher order moments but allows the data to be serially correlated. \tcb{Model~\ref{model-d4} is a higher-dimensional version of Model~\ref{model-d2} with changes in a subset of the variables.}
The results are reported in Table~\ref{multcpt-high}, from which we see that the multi-lag NP-MOJO procedure gives the strongest overall performance, particularly in the serially correlated model~\ref{model-d3}. 
KCPA performs the best from the competing methods \tcb{(and is the best overall in~\ref{model-d4})}, and NWBS tends to under-detect the change points while cpt.np over-detects them.

\begingroup
{\small
\setlength{\tabcolsep}{3pt}
\setlength{\LTcapwidth}{\textwidth}
\begin{longtable}{c c ccccc cc}
\caption{\ref{model-d1}--\ref{model-d3}: we report the distribution of the estimated number of change points and the average CM and VM over 1000 realisations.
The modal value of $\wh q - q$ in each row is given in bold.
Also, the best performance for each metric is underlined for each scenario.}
\label{multcpt-high}
\endfirsthead
\endhead
\toprule	
&& \multicolumn{5}{c}{$\wh{q} - q$} &   &   \\ 
 Model       & Method      &  $-2$     & $-1$        & $\mathbf{0}$  & $1$    & $\geq 2$    & CM & VM       \\ 
 \cmidrule(lr){1-2} \cmidrule(lr){3-7} \cmidrule(lr){8-9}
\ref{model-d1}    &  NP-MOJO-$0$      & 0.000       &  0.069  & {\bf{0.892}}    & 0.037     &     0.002          & {0.933}  & {0.904}  \\
 &  NP-MOJO-$1$        &  0.003   &   0.134  & \bf{0.810}  &        0.053 & 0.000      & 0.902 & 0.874   \\
       &  NP-MOJO-$2$      & 0.000 &  0.128  & \bf{0.823}  & 0.049        &  0.000              &  0.905 & 0.878 \\ 
       &  NP-MOJO-$\mc L$   &  0.000      &    0.034   & \underline{\textbf{0.960}}  & 0.006     &   0.000          &  \underline{0.942} & \underline{0.909}    \\ \cmidrule(lr){3-7} \cmidrule(lr){8-9}
       &  E-Divisive      &  0.113     &  0.086  & \bf{0.699} &    0.079   & 0.023 & 0.832 & 0.770 \\
       &  KCPA      &  0.086    &  0.002  &  {\bf{0.890}} & 0.019       &  0.003      & {0.909}  & {0.853} \\      
       &  NWBS      & \bf{0.496}     & 0.070   & 0.339 &   0.076   &   0.019    & 0.582   & 0.394 \\
       &  cpt.np      & 0.006     & 0.004   &  \bf{0.592} &     0.276   &  0.122   & 0.896 & 0.864 \\ 
        \cmidrule(lr){1-2} \cmidrule(lr){3-7} \cmidrule(lr){8-9}
               \ref{model-d2}      &  NP-MOJO-$0$      &  0.000      & 0.005  & {\bf{0.981}}     &  0.014    & 0.000           & \underline{0.970} & {0.944} \\
&  NP-MOJO-$1$        &  0.000    &  0.126 & \bf{0.824}   & 0.049   & 0.001 & 0.904 & 0.874 \\
      &  NP-MOJO-$2$      & 0.001     & 0.105    & \bf{0.831}    & 0.060 & 0.003   & 0.909  & 0.879 \\ &  NP-MOJO-$\mc L$   &  0.000      &    0.003   & \underline{\textbf{0.993}}  & 0.004     &   0.000          &   {0.964} &  {0.934}    \\ \cmidrule(lr){3-7} \cmidrule(lr){8-9}
      &  E-Divisive &  0.000  & 0.000 & \bf{0.894} &  0.058  & 0.048     & 0.956 & {0.931} \\
      &  KCPA      &  0.104     & 0.001   & \bf{0.880} & 0.014    &   0.001      & 0.897  & 0.835  \\      
      &  NWBS      & 0.350    &  0.000  & \bf{0.508}   & 0.101    & 0.041    & 0.731 & 0.596 \\
      &  cpt.np      & 0.000  &  0.000 & \bf{0.741}  & 0.184  & 0.075   & 0.962  & \underline{0.950}  \\  
     \cmidrule(lr){1-2} \cmidrule(lr){3-7} \cmidrule(lr){8-9}
          \ref{model-d3}   &  NP-MOJO-$0$      &  0.003      & 0.139  & {\bf{0.809}}     & 0.049      &   0.000         & {0.899}  &  {0.872} \\
&  NP-MOJO-$1$        &0.006     &  0.155   &  \bf{0.792}   & 0.047    & 0.000 & {0.892} &  {0.864}\\
     &  NP-MOJO-$2$      & 0.021     & 0.248   &  \bf{0.685}  &   0.045  & 0.001   &  0.848 & 0.819 \\ 
     &  NP-MOJO-$\mc L$   &  0.002      &    0.082   & \underline{\textbf{0.914}}  & 0.002     &   0.000          &  \underline{0.917} & \underline{0.885}    \\ \cmidrule(lr){3-7} \cmidrule(lr){8-9}
     &  E-Divisive       & 0.005       &  0.002   & 0.072   & 0.118   &  \textbf{0.803}         & 0.681 & 0.707  \\
     &  KCPA      &  0.441      & 0.012   & \bf{0.481} & 0.052     &   0.014       & 0.667  & 0.500 \\      
     &  NWBS      & 0.047     & 0.015   & 0.139 &   0.124   &   \bf{0.675}    & 0.680   & 0.676 \\
     &  cpt.np      & 0.000     & 0.000   & 0.045 &     0.055   &  \textbf{0.900}   & 0.726 & 0.756 \\    \cmidrule(lr){1-2} \cmidrule(lr){3-7} \cmidrule(lr){8-9}
          \ref{model-d4}   &  NP-MOJO-$0$      &  0.000      & 0.016  & {\bf{0.979}}     & 0.005      &   0.000         & {0.972}  &  {0.950} \\
&  NP-MOJO-$1$        & 0.009     &  0.211   &  \bf{0.774}   & 0.006    & 0.000 & {0.890} &  {0.871}\\
     &  NP-MOJO-$2$      & 0.008     & 0.207   &  \bf{0.775}  &   0.010  & 0.000   &  0.893 & 0.875 \\ 
     &  NP-MOJO-$\mc L$   &  0.000      &    0.015   & {\textbf{0.985}}  & 0.000     &   0.000          &  {0.971} & {0.948}    \\ \cmidrule(lr){3-7} \cmidrule(lr){8-9}
     &  E-Divisive       & 0.000       &  0.000   &\textbf{0.908}   & 0.069   &  {0.023}         & 0.970 & 0.952  \\
     &  KCPA      &  0.000      & 0.000   & \underline{\bf{1.000}} & 0.000     &   0.000       & \underline{0.991}  & \underline{0.979} \\      
      \bottomrule
\end{longtable}}
\endgroup

\subsection{\tcb{Varying the tuning parameters}}

In this section, we perform simulations with different values of tuning parameters to those considered in the main simulations, to assess the method's sensitivity to some of the choices of tuning parameters. In all simulations, we consider the following subset of the scenarios used in Section \ref{sec:sim}: \ref{model-n3}, (A1), \ref{model-b5}, \ref{model-c1}, and \ref{model-d3}.

\subsubsection{\tcb{Kernel choice}}
\label{sec:sim-study:kernel}

As in \cite{matteson2014} and \cite{arlot2019kernel}, we must choose a kernel for computing the test statistic, which will affect the performance of the change point detection method. Here, we analyse NP-MOJO's performance when using kernel $h_1$ (the recommended kernel in \cite{arlot2019kernel}, c.f. Lemma~\ref{lemma:weight-int-identities} in the main text) and kernel $h_3$ (the kernel used in \cite{matteson2014}, c.f. Lemma~\ref{lemma:weight-int3}). All other tuning parameters are identical to those described in Section 4 of the main text, with the value $\gamma = 1$ used for the kernel $h_3$ (as in \cite{matteson2014}). The results for $h_1$ and $h_3$ mirror the results of KCPA and E-Divisive from the main simulation study respectively, and demonstrate the superiority in performance of kernel $h_2$. After $h_2$, kernel $h_1$ is the best performing, with $h_3$ generally the worst.

\begingroup
{\small
\setlength{\tabcolsep}{3pt}
\setlength{\LTcapwidth}{\textwidth}
\begin{longtable}{c c ccccc cc}
\caption{We report the distribution of the estimated number of change points and the average CM and VM over 1000 realisations when kernel $h_1$ is used. The modal value of $\wh q - q$ in each row is given in bold.}
\label{table-h1}
\endfirsthead
\endhead
\toprule	
&& \multicolumn{5}{c}{$\wh{q} - q$ / $\wh{q}_\ell - q_\ell$} &   &   \\ 
 Model       & Method      &  $\leq-2$     & $-1$        & $\mathbf{0}$  & $1$    & $\geq 2$    & CM & VM       \\ 
 \cmidrule(lr){1-2} \cmidrule(lr){3-7} \cmidrule(lr){8-9}
\ref{model-n3}   &  NP-MOJO-$0$   &  --      & --   &  \textbf{0.827}    & 0.153     &    0.020       & --  & --   \\
&  NP-MOJO-$1$     &  --   & --     & \textbf{0.809}   & 0.167         &  0.024          & --   & --      \\
      &  NP-MOJO-$2$   &  --      &    --   & \textbf{0.792}  & 0.183     &   0.025          &  -- & --    \\ 
            &  NP-MOJO-$\mc L$   &  --      &    --  & \textbf{0.784}  & 0.188     &   0.028          &  --& --    \\  \cmidrule(lr){1-2} \cmidrule(lr){3-7} \cmidrule(lr){8-9}
(A1) &  NP-MOJO-$0$ &  0.000      & 0.002    &  \textbf{0.998}    &  0.000       & 0.000             &  0.972 &  0.956     \\
&  NP-MOJO-$1$  & 0.000     & 0.000     & {\textbf{1.000}}   & 0.000          &   0.000         & 0.975   & 0.960        \\
      &  NP-MOJO-$2$ &   0.000     &  0.000     & \textbf{1.000}  & 0.000        &   0.000          & 0.975  & 0.960     \\
       &  NP-MOJO-$\mc L$   &  0.000      &    0.000   & {\textbf{1.000}}  & 0.000     &   0.000          &  0.975 & 0.959    \\     \cmidrule(lr){1-2} \cmidrule(lr){3-7} \cmidrule(lr){8-9}
\ref{model-b5}     &  NP-MOJO-$0$    &    0.169    &  0.363  &  {\textbf{0.452}}    &  0.016      &    0.000         &  0.779 & 0.802  \\
&  NP-MOJO-$1$  &  \textbf{0.888}    &  0.104    & 0.008   &    0.000     &  0.000          & 0.382   & 0.279        \\
      &  NP-MOJO-$2$  & \textbf{0.868}       &  0.112     & 0.019  &  0.001     &  0.000      & 0.407  & 0.327    \\
       &  NP-MOJO-$\mc L$   &  0.168      &    0.365   & {\textbf{0.467}}  & 0.000     &   0.000          &  0.780 & 0.803    \\ \cmidrule(lr){1-2} \cmidrule(lr){3-7} \cmidrule(lr){8-9}
\ref{model-c1}    &  NP-MOJO-$0$  &  --      &  --  &   {\textbf{0.797}}   &    0.187    &  0.016            & --  &  --     \\
&  NP-MOJO-$1$  &  0.117    & 0.375     & {\textbf{0.450}}   & 0.055          &  0.003          &  0.734  & 0.692         \\
      &  NP-MOJO-$2$ &  --    & --     & \textbf{0.758}   & 0.218          &  0.024          &  --  & --      \\ 
       &  NP-MOJO-$\mc L$   &  0.114     &    0.398   & {\textbf{0.481}}  & 0.007     &   0.000          &  0.731 & 0.687    \\ \cmidrule(lr){1-2} \cmidrule(lr){3-7} \cmidrule(lr){8-9} 
    \ref{model-d3}     &  NP-MOJO-$0$  &  0.074      &  0.386  &   \textbf{0.506}  &    0.033    &  0.001            & 0.771  &  0.735     \\
&  NP-MOJO-$1$  &  0.191   & \textbf{0.513}   & \textbf{0.274}   & 0.021        &  0.001          &  {0.657}  & {0.591}         \\
      &  NP-MOJO-$2$ &  0.235   & \textbf{0.502}   & 0.237  & 0.025         &  0.001        &  {0.629} & {0.549}      \\ 
       &  NP-MOJO-$\mc L$   &  0.068     &    0.396  & \textbf{0.530} & 0.006     &   0.000      &  0.775 & 0.739  \\
      \bottomrule
\end{longtable}}
\endgroup

\begingroup
{\small
\setlength{\tabcolsep}{3pt}
\setlength{\LTcapwidth}{\textwidth}
\begin{longtable}{c c ccccc cc}
\caption{We report the distribution of the estimated number of change points and the average CM and VM over 1000 realisations when kernel $h_3$ is used. The modal value of $\wh q - q$ in each row is given in bold.}
\label{table-h3}
\endfirsthead
\endhead
\toprule	
&& \multicolumn{5}{c}{$\wh{q} - q$ / $\wh{q}_\ell - q_\ell$} &   &   \\ 
 Model       & Method      &  $\leq-2$     & $-1$        & $\mathbf{0}$  & $1$    & $\geq 2$    & CM & VM       \\ 
 \cmidrule(lr){1-2} \cmidrule(lr){3-7} \cmidrule(lr){8-9}
\ref{model-n3}   &  NP-MOJO-$0$   &  --      & --   &  \textbf{0.801}    & 0.175      &    0.024          & --  & --   \\
&  NP-MOJO-$1$     &  --   & --     & \textbf{0.796}   & 0.175         &  0.029         & --   & --      \\
      &  NP-MOJO-$2$   &  --      &    --   & \textbf{0.773}  & 0.198     &   0.029          &  -- & --    \\ 
            &  NP-MOJO-$\mc L$   &  --      &    --  & \textbf{0.761}  & 0.205     &   0.034          &  --& --    \\  \cmidrule(lr){1-2} \cmidrule(lr){3-7} \cmidrule(lr){8-9}
(A1) &  NP-MOJO-$0$ &  0.000      & 0.000    &  \underline{\textbf{1.000}}    &  0.000       & 0.000             &  0.975 &  0.959      \\
&  NP-MOJO-$1$  & 0.000     & 0.000     & {\textbf{1.000}}   & 0.000          &   0.000         & 0.976   & 0.960        \\
      &  NP-MOJO-$2$ &   0.000     &  0.000     & \textbf{1.000}  & 0.000        &   0.000          & 0.976  & 0.960     \\
       &  NP-MOJO-$\mc L$   &  0.000      &    0.000   & \underline{\textbf{1.000}}  & 0.000     &   0.000          &  0.975 & 0.960    \\     \cmidrule(lr){1-2} \cmidrule(lr){3-7} \cmidrule(lr){8-9}
\ref{model-b5}     &  NP-MOJO-$0$    &    \textbf{0.677}    &  0.232  &  0.091   &  0.000     &    0.000          &   0.512 & 0.487  \\
&  NP-MOJO-$1$  &  \textbf{0.970}    &  0.029    & {0.001}   &    0.000     &  0.000          & 0.324   & 0.166        \\
      &  NP-MOJO-$2$  & \textbf{0.951}       &  0.045     & 0.003 &  0.001     &  0.000      & 0.334  & 0.185     \\
       &  NP-MOJO-$\mc L$   &  \textbf{0.677}      &    0.235  & 0.088  & 0.000     &   0.000          &  0.514 & 0.490    \\ \cmidrule(lr){1-2} \cmidrule(lr){3-7} \cmidrule(lr){8-9}
\ref{model-c1}    &  NP-MOJO-$0$  &  --      &  --  &   {\textbf{0.771}}   &    0.209    &  0.020           & --  &  --     \\
&  NP-MOJO-$1$  & \textbf{0.486}   & 0.340     & 0.164   & 0.010         &  0.000          &  {0.513}  & {0.348}           \\
      &  NP-MOJO-$2$ &  --    & --     & \textbf{0.753}   & 0.220      &  0.027          &  --  & --      \\ 
       &  NP-MOJO-$\mc L$   &  \textbf{0.476}      &    0.375   & 0.148  & 0.001     &   0.000          &  0.511 & 0.349     \\ \cmidrule(lr){1-2} \cmidrule(lr){3-7} \cmidrule(lr){8-9} 
   \ref{model-d3}     &  NP-MOJO-$0$  &  \textbf{0.537}      &  0.379  &   0.076  &    0.008    &  0.000            & 0.486  &  0.306     \\
&  NP-MOJO-$1$  &  \textbf{0.518}   & 0.390    & {0.085}   & 0.007         &  0.000          &  {0.491}  & {0.317}         \\
      &  NP-MOJO-$2$ &  \textbf{0.559}    & 0.354   & 0.080  & 0.007         &  0.000        &  {0.477} & {0.287}      \\ 
       &  NP-MOJO-$\mc L$   &  \textbf{0.470}      &    0.433  & 0.093  & 0.004     &   0.000      &  0.511 & 0.351  \\
      \bottomrule
\end{longtable}}
\endgroup

\subsubsection{\tcb{Bootstrap parameter $b_n$}}
\label{sec:sim-study:bn}

The bootstrap parameter $b_n$ affects the level of dependence of the bootstrapped multiplier sequence, with larger values of $b_n$ giving higher levels of dependence. In previous simulations, we set $b_n = 1.5 n^{1/3}$; here, we assess the performance of NP-MOJO using $b_n = n^{1/3}$ and $b_n = 2 n^{1/3}$. Table~\ref{table-b10} shows the results when $b_n = n^{1/3}$, and Table~\ref{table-b20} shows the results with $b_n = 2n^{1/3}$. Firstly, as we would expect, the larger the value of $b_n$, the more conservative NP-MOJO is, due to the increased dependence in the bootstrap sequence. The results, to be compared with those in the main text where $b_n = 1.5 n^{1/3}$, show that NP-MOJO is largely insensitive to the choice of $b_n$, with the method performing similarly over the range of bandwidths $b_n \in \{ n^{1/3}, 1.5n^{1/3}, 2n^{1/3} \}$. 

\begingroup
{\small
\setlength{\tabcolsep}{3pt}
\setlength{\LTcapwidth}{\textwidth}
\begin{longtable}{c c ccccc cc}
\caption{We report the distribution of the estimated number of change points and the average CM and VM over 1000 realisations when $b_n = n^{1/3}$. The modal value of $\wh q - q$ in each row is given in bold.}
\label{table-b10}
\endfirsthead
\endhead
\toprule	
&& \multicolumn{5}{c}{$\wh{q} - q$ / $\wh{q}_\ell - q_\ell$} &   &   \\ 
 Model       & Method      &  $\leq-2$     & $-1$        & $\mathbf{0}$  & $1$    & $\geq 2$    & CM & VM       \\ 
 \cmidrule(lr){1-2} \cmidrule(lr){3-7} \cmidrule(lr){8-9}
\ref{model-n3}   &  NP-MOJO-$0$   &  --      & --   &  \textbf{0.840}    & 0.140     &    0.020       & --  & --   \\
&  NP-MOJO-$1$     &  --   & --     & \textbf{0.852}   & 0.133         &  0.015          & --   & --      \\
      &  NP-MOJO-$2$   &  --      &    --   & \textbf{0.825}  & 0.156     &   0.019        &  -- & --    \\ 
            &  NP-MOJO-$\mc L$   &  --      &    --  & \textbf{0.791}  & 0.187     &   0.022          &  --& --    \\  \cmidrule(lr){1-2} \cmidrule(lr){3-7} \cmidrule(lr){8-9}
(A1) &  NP-MOJO-$0$ &  0.000      & 0.004    &  \textbf{0.988}    &  0.008       & 0.000             &  0.961 &  0.943     \\
&  NP-MOJO-$1$  & 0.000     & 0.000     & {\textbf{1.000}}   & 0.000          &   0.000         & 0.972   & 0.955        \\
      &  NP-MOJO-$2$ &   0.000     &  0.000     & \textbf{0.999}  & 0.001        &   0.000          & 0.972  & 0.955     \\
       &  NP-MOJO-$\mc L$   &  0.000      &    0.000   & {\textbf{1.000}}  & 0.000     &   0.000          &  0.969 & 0.952    \\     \cmidrule(lr){1-2} \cmidrule(lr){3-7} \cmidrule(lr){8-9}
\ref{model-b5}     &  NP-MOJO-$0$    &    0.000    &  0.000  &  {\textbf{0.998}}    &  0.002      &    0.000          &  0.975 & 0.960  \\
&  NP-MOJO-$1$  &  0.001    &  0.056    & \textbf{0.926}   &    0.017     &  0.000          & 0.947   & 0.936        \\
      &  NP-MOJO-$2$  & 0.001       &  0.044     & {\textbf{0.935}}  &  0.020     &  0.000      & 0.951  & 0.938     \\
       &  NP-MOJO-$\mc L$   &  0.000      &    0.000   & {\textbf{1.000}}  & 0.000     &   0.000          &  0.973 & 0.957    \\ \cmidrule(lr){1-2} \cmidrule(lr){3-7} \cmidrule(lr){8-9}
\ref{model-c1}    &  NP-MOJO-$0$  &  --      &  --  &   {\textbf{0.806}}   &    0.169    &  0.025            & --  &  --     \\
&  NP-MOJO-$1$  &  0.000    & 0.000     & {\textbf{0.934}}   & 0.066          &  0.000          &  0.975  & 0.959         \\
      &  NP-MOJO-$2$ &  --    & --     & \textbf{0.795}   & 0.177          &  0.028          &  --  & --      \\ 
       &  NP-MOJO-$\mc L$   &  0.000      &    0.000   & {\textbf{0.983}}  & 0.017     &   0.000          &  0.980 & 0.962    \\ \cmidrule(lr){1-2} \cmidrule(lr){3-7} \cmidrule(lr){8-9} 
    \ref{model-d3}     &  NP-MOJO-$0$  &  0.002      &  0.106  &   {\textbf{0.826}}   &    0.066    &  0.000            & 0.908  &  0.879     \\
&  NP-MOJO-$1$  &  0.003   & 0.107     & \textbf{0.830}   & 0.060         &  0.000          &  {0.905}  & {0.876}         \\
      &  NP-MOJO-$2$ &  0.014    & 0.198    & {\textbf{0.725}}   & 0.062         &  0.001        &  {0.863} & {0.834}      \\ 
       &  NP-MOJO-$\mc L$   &  0.002      &    0.058   & {\textbf{0.932}}  & 0.008     &   0.000      &  0.923 & 0.888  \\
      \bottomrule
\end{longtable}}
\endgroup

\begingroup
{\small
\setlength{\tabcolsep}{3pt}
\setlength{\LTcapwidth}{\textwidth}
\begin{longtable}{c c ccccc cc}
\caption{We report the distribution of the estimated number of change points and the average CM and VM over 1000 realisations when $b_n = 2n^{1/3}$. The modal value of $\wh q - q$ in each row is given in bold.}
\label{table-b20}
\endfirsthead
\endhead
\toprule	
&& \multicolumn{5}{c}{$\wh{q} - q$ / $\wh{q}_\ell - q_\ell$} &   &   \\ 
 Model       & Method      &  $\leq-2$     & $-1$        & $\mathbf{0}$  & $1$    & $\geq 2$    & CM & VM       \\ 
 \cmidrule(lr){1-2} \cmidrule(lr){3-7} \cmidrule(lr){8-9}
\ref{model-n3}   &  NP-MOJO-$0$   &  --      & --   &  \textbf{0.888}    & 0.103      &    0.009          & --  & --   \\
&  NP-MOJO-$1$     &  --   & --     & \textbf{0.901}   & 0.092         &  0.007          & --   & --      \\
      &  NP-MOJO-$2$   &  --      &    --   & \textbf{0.880}  & 0.112     &   0.008          &  -- & --    \\ 
            &  NP-MOJO-$\mc L$   &  --      &    --  & \textbf{0.854}  & 0.136     &   0.010          &  --& --    \\  \cmidrule(lr){1-2} \cmidrule(lr){3-7} \cmidrule(lr){8-9}
(A1) &  NP-MOJO-$0$ &  0.001      & 0.040    &  \textbf{0.954}    &  0.005       & 0.000             &  0.953 &  0.939      \\
&  NP-MOJO-$1$  & 0.000     & 0.011     & {\textbf{0.989}}   & 0.000          &   0.000         & 0.969   & 0.953        \\
      &  NP-MOJO-$2$ &   0.000     &  0.006     & \textbf{0.993}  & 0.001        &   0.000          & 0.979  & 0.954     \\
       &  NP-MOJO-$\mc L$   &  0.000      &    0.004   & \textbf{0.996}  & 0.000     &   0.000          &  0.970 & 0.953    \\     \cmidrule(lr){1-2} \cmidrule(lr){3-7} \cmidrule(lr){8-9}
\ref{model-b5}     &  NP-MOJO-$0$    &    0.000    &  0.003  &  {\textbf{0.995}}    &  0.002      &    0.000          &   {0.974} & 0.959  \\
&  NP-MOJO-$1$  &  0.020    &  0.169    & \textbf{0.806}   &    0.005     &  0.000          & 0.913   & 0.916        \\
      &  NP-MOJO-$2$  & 0.016       &  0.167     & {\textbf{0.811}}  &  0.006     &  0.000      & 0.916  & 0.918     \\
       &  NP-MOJO-$\mc L$   &  0.000      &    0.002   & {\textbf{0.998}}  & 0.000     &   0.000          &  0.973 & 0.958    \\ \cmidrule(lr){1-2} \cmidrule(lr){3-7} \cmidrule(lr){8-9}
\ref{model-c1}    &  NP-MOJO-$0$  &  --      &  --  &   {\textbf{0.857}}   &    0.134    &  0.009            & --  &  --     \\
&  NP-MOJO-$1$  & 0.000   & 0.002     & \textbf{0.965}   & 0.033         &  0.000          &  {0.979}  & {0.962}           \\
      &  NP-MOJO-$2$ &  --    & --     & \textbf{0.859}   & 0.133      &  0.008          &  --  & --      \\ 
       &  NP-MOJO-$\mc L$   &  0.000      &    0.002   & {\textbf{0.988}}  & 0.010     &   0.000          &  0.980 & 0.963     \\ \cmidrule(lr){1-2} \cmidrule(lr){3-7} \cmidrule(lr){8-9} 
    \ref{model-d3}     &  NP-MOJO-$0$  &  0.003      &  0.173  &   {\textbf{0.782}}   &    0.042    &  0.000            & 0.891  &  0.866     \\
&  NP-MOJO-$1$  &  0.007   & 0.185     & \textbf{0.765}   & 0.043         &  0.000          &  {0.883}  & {0.857}         \\
      &  NP-MOJO-$2$ &  0.025    & 0.270    & {\textbf{0.670}}   & 0.035         &  0.000        &  {0.840} & {0.811}      \\ 
       &  NP-MOJO-$\mc L$   &  0.001     &    0.101   & {\textbf{0.897}}  & 0.001     &   0.000          &  0.911 & 0.880  \\
      \bottomrule
\end{longtable}}
\endgroup

\subsubsection{\tcb{Varying $\eta$}}
\label{sec:sim-study:eta}

Next, we check the sensitivity of NP-MOJO to the choice of $\eta$, the parameter that governs how large the local environment is when deciding if a local maximum should be declared a true change point. We run the method with $\eta = 0 .2$ and $\eta = 0.6$ with all other tuning parameters identical to those  described in the main simulation study (where $\eta = 0.4$). The results for $\eta = 0.2$ are given in Table~\ref{table-eta2}, and the results for $\eta = 0.6$ are given in Table~\ref{table-eta6}, demonstrating that that multi-lag NP-MOJO is generally insensitive to the choice of $\eta$. When $\eta = 0.2$, corresponding to a less strict change point acceptance rule, the single-lag NP-MOJO procedure can overestimate the true number of changes. However, the aggregation of multi-lag NP-MOJO is able to correct for the overestimation and correctly estimate the true number of change points: see for example the result for Model (A1).

\begingroup
{\small
\setlength{\tabcolsep}{3pt}
\setlength{\LTcapwidth}{\textwidth}
\begin{longtable}{c c ccccc cc}
\caption{We report the distribution of the estimated number of change points and the average CM and VM over 1000 realisations when $\eta = 0.2$. The modal value of $\wh q - q$ in each row is given in bold.}
\label{table-eta2}
\endfirsthead
\endhead
\toprule	
&& \multicolumn{5}{c}{$\wh{q} - q$ / $\wh{q}_\ell - q_\ell$} &   &   \\ 
 Model       & Method      &  $\leq-2$     & $-1$        & $\mathbf{0}$  & $1$    & $\geq 2$    & CM & VM       \\ 
 \cmidrule(lr){1-2} \cmidrule(lr){3-7} \cmidrule(lr){8-9}
\ref{model-n3}   &  NP-MOJO-$0$   &  --      & --   &  \textbf{0.867}    & 0.122     &    0.011       & --  & --   \\
&  NP-MOJO-$1$     &  --   & --     & \textbf{0.891}   & 0.099        &  0.010         & --   & --      \\
      &  NP-MOJO-$2$   &  --      &    --   & \textbf{0.863}  & 0.126     &   0.011          &  -- & --    \\ 
            &  NP-MOJO-$\mc L$   &  --      &    --  & \textbf{0.825}  & 0.161     &   0.014          &  --& --    \\  \cmidrule(lr){1-2} \cmidrule(lr){3-7} \cmidrule(lr){8-9}
(A1) &  NP-MOJO-$0$ &  0.000      & 0.015    &  \textbf{0.756}    &  0.202       & 0.027             &  0.947 &  0.934     \\
&  NP-MOJO-$1$  & 0.000     & 0.003     & {\textbf{0.851}}   & 0.138          &   0.008         & 0.963   & 0.949        \\
      &  NP-MOJO-$2$ &   0.000     &  0.003     & \textbf{0.868}  & 0.120        &   0.009          & 0.964  & 0.949     \\
       &  NP-MOJO-$\mc L$   &  0.000      &    0.002   & {\textbf{0.998}}  & 0.000     &   0.000          &  0.969 & 0.952    \\     \cmidrule(lr){1-2} \cmidrule(lr){3-7} \cmidrule(lr){8-9}
\ref{model-b5}     &  NP-MOJO-$0$    &    0.000    &  0.000  &  {\textbf{0.883}}    &  0.113      &    0.004          &  0.969 & 0.955  \\
&  NP-MOJO-$1$  &  0.005    &  0.108    & \textbf{0.729}   &    0.147     &  0.011         & 0.923   & 0.922        \\
      &  NP-MOJO-$2$  & 0.005       &  0.092     & {\textbf{0.726}}  &  0.170     &  0.007      & 0.928  & 0.925     \\
       &  NP-MOJO-$\mc L$   &  0.000      &    0.000   & {\textbf{1.000}}  & 0.000     &   0.000          &  0.973 & 0.958    \\ \cmidrule(lr){1-2} \cmidrule(lr){3-7} \cmidrule(lr){8-9}
\ref{model-c1}    &  NP-MOJO-$0$  &  --      &  --  &   {\textbf{0.845}}   &    0.135    &  0.020            & --  &  --     \\
&  NP-MOJO-$1$  &  --  & 0.002     & {\textbf{0.820}}   & 0.157          &  0.021          &  0.968  & 0.952         \\
      &  NP-MOJO-$2$ &  --    & --     & \textbf{0.839}   & 0.139         &  0.022          &  --  & --      \\ 
       &  NP-MOJO-$\mc L$   &  --      &    0.002   & {\textbf{0.984}}  & 0.014     &   0.000          &  0.980 & 0.962    \\ \cmidrule(lr){1-2} \cmidrule(lr){3-7} \cmidrule(lr){8-9} 
    \ref{model-d3}     &  NP-MOJO-$0$  &  0.003      &  0.125  &   {\textbf{0.516}}   &    0.294    &  0.062            & 0.889  &  0.863     \\
&  NP-MOJO-$1$  &  0.004   & 0.138     & \textbf{0.505}   & 0.285         &  0.068          &  {0.880}  & {0.854}         \\
      &  NP-MOJO-$2$ &  0.022    & 0.209    & {\textbf{0.471}}   & 0.237         &  0.061        &  {0.837} & {0.809}      \\ 
       &  NP-MOJO-$\mc L$   &  0.002      &    0.081   & {\textbf{0.906}}  & 0.011     &   0.000      &  0.916 & 0.883  \\
      \bottomrule
\end{longtable}}
\endgroup

\begingroup
{\small
\setlength{\tabcolsep}{3pt}
\setlength{\LTcapwidth}{\textwidth}
\begin{longtable}{c c ccccc cc}
\caption{We report the distribution of the estimated number of change points and the average CM and VM over 1000 realisations when $\eta = 0.6$. The modal value of $\wh q - q$ in each row is given in bold.}
\label{table-eta6}
\endfirsthead
\endhead
\toprule	
&& \multicolumn{5}{c}{$\wh{q} - q$ / $\wh{q}_\ell - q_\ell$} &   &   \\ 
 Model       & Method      &  $\leq-2$     & $-1$        & $\mathbf{0}$  & $1$    & $\geq 2$    & CM & VM       \\ 
 \cmidrule(lr){1-2} \cmidrule(lr){3-7} \cmidrule(lr){8-9}
\ref{model-n3}   &  NP-MOJO-$0$   &  --      & --   &  \textbf{0.867}    & 0.122     &    0.011       & --  & --   \\
&  NP-MOJO-$1$     &  --   & --     & \textbf{0.891}   & 0.100         &  0.009          & --   & --      \\
      &  NP-MOJO-$2$   &  --      &    --   & \textbf{0.863}  & 0.128     &   0.009          &  -- & --    \\ 
            &  NP-MOJO-$\mc L$   &  --      &    --  & \textbf{0.825}  & 0.161     &   0.014          &  --& --    \\  \cmidrule(lr){1-2} \cmidrule(lr){3-7} \cmidrule(lr){8-9}
(A1) &  NP-MOJO-$0$ &  0.000      & 0.018   &  \textbf{0.982}    &  0.000       & 0.000             &  0.959 &  0.942     \\
&  NP-MOJO-$1$  & 0.000     & 0.003     & {\textbf{0.997}}   & 0.000          &   0.000         & 0.971   & 0.955        \\
      &  NP-MOJO-$2$ &   0.000     &  0.003     & \textbf{0.997}  & 0.000        &   0.000          & 0.971  & 0.954     \\
       &  NP-MOJO-$\mc L$   &  0.000      &    0.002   & {\textbf{0.998}}  & 0.000     &   0.000          &  0.970 & 0.953    \\     \cmidrule(lr){1-2} \cmidrule(lr){3-7} \cmidrule(lr){8-9}
\ref{model-b5}     &  NP-MOJO-$0$    &    0.000    &  0.000  &  {\textbf{1.000}}    &  0.000      &    0.000          &  0.975 & 0.960  \\
&  NP-MOJO-$1$  &  0.005    &  0.131    & \textbf{0.864}   &    0.000     &  0.000          & 0.930   & 0.926        \\
      &  NP-MOJO-$2$  & 0.005       &  0.111     & {\textbf{0.883}}  &  0.001     &  0.000      & 0.935  & 0.929     \\
       &  NP-MOJO-$\mc L$   &  0.000      &    0.000   & {\textbf{1.000}}  & 0.000     &   0.000          &  0.974 & 0.958    \\ \cmidrule(lr){1-2} \cmidrule(lr){3-7} \cmidrule(lr){8-9}
\ref{model-c1}    &  NP-MOJO-$0$  &  --      &  --  &   {\textbf{0.845}}   &    0.139    &  0.016            & --  &  --     \\
&  NP-MOJO-$1$  &  0.000    & 0.002     & {\textbf{0.985}}   & 0.013          &  0.000          &  0.981  & 0.963         \\
      &  NP-MOJO-$2$ &  --    & --     & \textbf{0.839}   & 0.146          &  0.015          &  --  & --      \\ 
       &  NP-MOJO-$\mc L$   &  0.000      &    0.002   & {\textbf{0.985}}  & 0.013     &   0.000          &  0.980 & 0.963    \\ \cmidrule(lr){1-2} \cmidrule(lr){3-7} \cmidrule(lr){8-9} 
    \ref{model-d3}     &  NP-MOJO-$0$  &  0.003      &  0.139  &   {\textbf{0.848}}   &    0.010    &  0.000            & 0.902  &  0.874     \\
&  NP-MOJO-$1$  &  0.004   & 0.155     & \textbf{0.834}   & 0.007         &  0.000          &  {0.894}  & {0.866}         \\
      &  NP-MOJO-$2$ &  0.022    & 0.252    & {\textbf{0.716}}   & 0.010         &  0.000        &  {0.848} & {0.818}      \\ 
       &  NP-MOJO-$\mc L$   &  0.002      &    0.081   & {\textbf{0.914}}  & 0.003     &   0.000      &  0.917 & 0.884  \\
      \bottomrule
\end{longtable}}
\endgroup

\subsubsection{\tcb{Varying $c$}}
\label{sec:sim-study:c}

Lastly, we investigate the sensitivity of NP-MOJO to the choice of $c$, the parameter that dictates how closely clustered change points across different lags should be. We run the method with $c = 0.8$ and $c = 1.2$ with all other tuning parameters identical to those  described in the main simulation study (where $c=1$). As with the parameter $\eta$, smaller values of $c$ give rise to more detected change points, and larger values give fewere changes. The results for $c = 0.8$ are given in Table~\ref{table-c8}, and the results for $c = 1.2$ are given in Table~\ref{table-c12}. As in the case with varying $\eta$, increasing $c$ results in fewer number of change points detected. The results show multi-lag NP-MOJO is robust to the choice of $c$ and that the recommended choice of $c=1$ comprises a good balance between over and under-estimation of the true number of change points. 

\begingroup
{\small
\setlength{\tabcolsep}{3pt}
\setlength{\LTcapwidth}{\textwidth}
\begin{longtable}{c c ccccc cc}
\caption{We report the distribution of the estimated number of change points and the average CM and VM over 1000 realisations when $c = 0.8$. The modal value of $\wh q - q$ in each row is given in bold.}
\label{table-c8}
\endfirsthead
\endhead
\toprule	
&& \multicolumn{5}{c}{$\wh{q} - q$ / $\wh{q}_\ell - q_\ell$} &   &   \\ 
 Model       & Method      &  $\leq-2$     & $-1$        & $\mathbf{0}$  & $1$    & $\geq 2$    & CM & VM       \\ 
 \cmidrule(lr){1-2} \cmidrule(lr){3-7} \cmidrule(lr){8-9}
\ref{model-n3}  &  NP-MOJO-$\mc L$   &  --      &    --  & \textbf{0.828}  & 0.156     &   0.016          &  --& --    \\  \cmidrule(lr){1-2} \cmidrule(lr){3-7} \cmidrule(lr){8-9}
(A1) &  NP-MOJO-$\mc L$   &  0.000      &    0.001   & {\textbf{0.999}}  & 0.000     &   0.000          &  0.970 & 0.953    \\     \cmidrule(lr){1-2} \cmidrule(lr){3-7} \cmidrule(lr){8-9}
\ref{model-b5}     &   NP-MOJO-$\mc L$   &  0.000      &    0.001   & {\textbf{0.998}}  & 0.001     &   0.000          &  0.973 & 0.958    \\ \cmidrule(lr){1-2} \cmidrule(lr){3-7} \cmidrule(lr){8-9}
\ref{model-c1}    &  NP-MOJO-$\mc L$   &  0.000      &    0.001   & {\textbf{0.928}}  & 0.071    &   0.000          &  0.972 & 0.958    \\ \cmidrule(lr){1-2} \cmidrule(lr){3-7} \cmidrule(lr){8-9} 
    \ref{model-d3}     &  NP-MOJO-$\mc L$   &  0.002      &    0.081   & {\textbf{0.903}}  & 0.014     &   0.000      &  0.917 & 0.885  \\
      \bottomrule
\end{longtable}}
\endgroup

\begingroup
{\small
\setlength{\tabcolsep}{3pt}
\setlength{\LTcapwidth}{\textwidth}
\begin{longtable}{c c ccccc cc}
\caption{We report the distribution of the estimated number of change points and the average CM and VM over 1000 realisations when $c = 1.2$. The modal value of $\wh q - q$ in each row is given in bold.}
\label{table-c12}
\endfirsthead
\endhead
\toprule	
&& \multicolumn{5}{c}{$\wh{q} - q$ / $\wh{q}_\ell - q_\ell$} &   &   \\ 
 Model       & Method      &  $\leq-2$     & $-1$        & $\mathbf{0}$  & $1$    & $\geq 2$    & CM & VM       \\ 
 \cmidrule(lr){1-2} \cmidrule(lr){3-7} \cmidrule(lr){8-9}
\ref{model-n3}  &  NP-MOJO-$\mc L$   &  --      &    --  & \textbf{0.828}  & 0.159     &   0.013          &  --& --    \\  \cmidrule(lr){1-2} \cmidrule(lr){3-7} \cmidrule(lr){8-9}
(A1) &  NP-MOJO-$\mc L$   &  0.000      &    0.003   & {\textbf{0.997}}  & 0.000     &   0.000          &  0.970 & 0.953    \\     \cmidrule(lr){1-2} \cmidrule(lr){3-7} \cmidrule(lr){8-9}
\ref{model-b5}     &   NP-MOJO-$\mc L$   &  0.000      &    0.002   & {\textbf{0.998}}  & 0.000     &   0.000          &  0.973 & 0.958    \\ \cmidrule(lr){1-2} \cmidrule(lr){3-7} \cmidrule(lr){8-9}
\ref{model-c1}    &  NP-MOJO-$\mc L$   &  0.000      &    0.002   & {\textbf{0.998}}  & 0.000     &   0.000          &  0.982 & 0.964    \\ \cmidrule(lr){1-2} \cmidrule(lr){3-7} \cmidrule(lr){8-9} 
    \ref{model-d3}     &  NP-MOJO-$\mc L$   &  0.002      &    0.087   & {\textbf{0.911}}  & 0.000     &   0.000      &  0.917 & 0.885  \\
      \bottomrule
\end{longtable}}
\endgroup

\subsection{\tcb{Scalability of NP-MOJO}}
\label{sec:sim-study:n}

In this section we run simulations with sample \tcr{sizes $n=500$, $n=2000$, and $n = 10000$} to assess how the performance of NP-MOJO scales with $n$. As in Section~\ref{sec:sim-detect}, changes are equispaced (unevenly spaced settings are studied in Section~\ref{sec:sim-multiscale}). \tcr{For each scenario (except the null model~\ref{model-n3}), when $n=500$, there is 1 less change than when $n=1000$, when $n=2000$, there are 2 more changes, and when $n=10000$, there are 12 more changes. All tuning parameters are identical to those described in Section 4 of the main text, except when $n=2000$, we set $G = n/8$, and when $n=10000$, we set $G = n/25$, due to the smaller distance between change points.}

\tcr{The results for $n=500$, $n=2000$, and $n=10000$ are reported in Tables~\ref{n500},~\ref{n2000}, and~\ref{n10000} respectively. Where appropriate, when $n=10000$, the competing nonparametric methods E-Divisive and NWBS, and reference method SBS were not included due to their prohibitively slow running times (see Figure~\ref{fig:speed} of the main text for run time comparisons).} We see that NP-MOJO generally has improved performance with increasing $n$. For example, for single-lag NP-MOJO in the null model~\ref{model-n3} with heavily correlated errors, the proportion of times where $\hat{q} \neq 0$ decreases as $n$ increases. In terms of each method's relative performances, the results here closely align with those in the main simulation study.

\begingroup
{\small
\setlength{\tabcolsep}{3pt}
\setlength{\LTcapwidth}{\textwidth}
\begin{longtable}{c c ccccc cc}
\caption{We report the distribution of the estimated number of change points and the average CM and VM over 1000 realisations when $n=500$. The modal value of $\wh q - q$ in each row is given in bold.
 Also, the best performance for each metric is underlined for each scenario.}
\label{n500}
\endfirsthead
\endhead
\toprule	
&& \multicolumn{5}{c}{$\wh{q} - q$} &   &   \\ 
 Model       & Method      &  $\leq-2$     & $-1$        & $\mathbf{0}$  & $1$    & $\geq 2$    & CM & VM       \\ 
 \cmidrule(lr){1-2} \cmidrule(lr){3-7} \cmidrule(lr){8-9}
\ref{model-n3}   &  NP-MOJO-$0$   &  --      & --   &  \textbf{0.842}    & 0.143      &    0.015          & --  & --   \\
&  NP-MOJO-$1$     &  --   & --     & \textbf{0.856}   & 0.130         &  0.014          & --   & --      \\
      &  NP-MOJO-$2$   &  --      &    --   & \textbf{0.843}  & 0.144     &   0.013          &  -- & --    \\ 
            &  NP-MOJO-$\mc L$   &  --      &    --  & \textbf{0.794}  & 0.187     &   0.019          &  --& --   \\ \cmidrule(lr){3-7} \cmidrule(lr){8-9}
      &  E-Divisive      & --     &  --  & 0.000  & 0.001        &    \textbf{0.999}         & -- & --   \\
      &  KCPA      &  --    & --     & \textbf{0.896}  &  0.048      &     0.056        &  --  & -- \\      
      &  NWBS      &  --      &  --    &  0.003 &   0.003     &    \textbf{0.994}        &  --  & --\\
      &  cpt.np      & --     &  --     &  0.000  & 0.000       &   \bf{1.000}          & --   & -- \\  \cmidrule(lr){1-2} \cmidrule(lr){3-7} \cmidrule(lr){8-9}
(A1) &  NP-MOJO-$0$ &  0.005      & 0.190    &  \textbf{0.773}    &  0.031       & 0.001             &  0.889 &  0.866      \\
&  NP-MOJO-$1$  & 0.000     & 0.062     & {\textbf{0.896}}   & 0.042          &   0.000         & 0.938   & 0.912        \\
      &  NP-MOJO-$2$ &   0.000     &  0.061     & \textbf{0.896}  & 0.043        &   0.000          & 0.939  & 0.912     \\
       &  NP-MOJO-$\mc L$   &  0.000      &    0.037   & {\textbf{0.960}}  & 0.003     &   0.000          &  0.947 & 0.917    \\  \cmidrule(lr){3-7} \cmidrule(lr){8-9}
      &  E-Divisive      &  0.000      &   0.000    &  \textbf{0.886}  &    0.091    &   0.023         & 0.961 &  0.941 \\
      &  KCPA      &  0.004     &   0.000   & \textbf{0.974} &  0.019     & 0.003           & {0.966}   &  {0.939} \\      
      &  NWBS      &  0.005    & 0.000     & \textbf{0.923}  &   0.030    &   0.042          &  0.954  & 0.925 \\
      &  cpt.np      &  0.000      &  0.000    & \textbf{0.862}  & 0.121        &  0.017           & 0.958   & 0.937 \\
      &  PELT      &  0.000     &  0.000     & \underline{\textbf{1.000}}  &    0.000    &    0.000         & \underline{0.976}   & \underline{0.952} \\
      &  WCM.gSa      &  0.000      &  0.000     & \textbf{0.958}  & 0.037        &     0.005        &  0.971  & 0.948 \\     \cmidrule(lr){1-2} \cmidrule(lr){3-7} \cmidrule(lr){8-9}
\ref{model-b5}     &  NP-MOJO-$0$    &    0.008    &  0.179  &  {\textbf{0.791}}    &  0.022      &    0.000          &   {0.900} & 0.880  \\
&  NP-MOJO-$1$  &  0.282    &  \textbf{0.506}    & {0.206}   &    0.005     &  0.001  & 0.620   & 0.534        \\
      &  NP-MOJO-$2$  & 0.266       &  \textbf{0.487}     & {{0.242}}  &  0.004     &  0.001      & 0.635  & 0.551     \\
       &  NP-MOJO-$\mc L$   &  0.006      &    0.169   & \underline{\textbf{0.820}}  & 0.005     &   0.000          &  \underline{0.905} & \underline{0.884}    \\ \cmidrule(lr){3-7} \cmidrule(lr){8-9}
      &  E-Divisive      &  \textbf{0.620}      &   0.131   & 0.196  & 0.041       &  0.012       &  0.461 & 0.237  \\
      &  KCPA      &  \textbf{0.747}      &  0.003    &  {0.232} &    0.017           &  0.001  & 0.491 & 0.235 \\      
      &  SBS      &  \textbf{0.802}      &   0.036   & 0.162   & 0.000       &  0.000         & 0.437   & 0.162 \\  \cmidrule(lr){1-2} \cmidrule(lr){3-7} \cmidrule(lr){8-9}
\ref{model-c1}    &  NP-MOJO-$0$  &  --      &  --  &   {\textbf{0.779}}   &    0.203    &  0.018            & --  &  --     \\
&  NP-MOJO-$1$  &  --    & 0.021     & \underline{\textbf{0.823}}   & 0.146          &  0.010          &  \underline{0.939}  & \underline{0.888}         \\
      &  NP-MOJO-$2$ &  --   & --     & \textbf{0.768}   & 0.213          &  0.019          &  --  & --      \\ 
       &  NP-MOJO-$\mc L$   &  --     &    0.018   & {\textbf{0.809}}  & 0.170     &   0.003          &  {0.929} & {0.880}    \\ \cmidrule(lr){3-7} \cmidrule(lr){8-9}
      &  WBSTS      &   --  &   0.000    &  \textbf{0.627} &  0.262      &  \textbf{0.111}           &  0.918 & 0.864  \\ \cmidrule(lr){1-2} \cmidrule(lr){3-7} \cmidrule(lr){8-9} 
    \ref{model-d3}     &  NP-MOJO-$0$  &  --      &  0.323  &   {\textbf{0.642}}   &    0.034    &  0.001            & 0.797  &  0.584     \\
&  NP-MOJO-$1$  &  --   & 0.286     & \textbf{0.672}   & 0.039         &  0.003          &  {0.813}  & {0.617}         \\
      &  NP-MOJO-$2$ & --    & 0.375    & {\textbf{0.583}}   & 0.038         &  0.004        &  {0.766} & {0.527}      \\ 
       &  NP-MOJO-$\mc L$   &  --      &    0.191   & \underline{\textbf{0.768}}  & 0.039     &   0.002          &  \underline{0.847} & 0.688    \\ \cmidrule(lr){3-7} \cmidrule(lr){8-9}
      &  E-Divisive      &   --    &   0.007    &  0.210 &  0.195     &  \bf{0.588}          &  0.692 &  0.626 \\
      &  KCPA      &  --      &  0.215     &  \textbf{0.729}  &    0.040    &    0.016         & 0.842   & 0.674 \\
            &  NWBS      &  --      & 0.044     &  0.197 &   0.146     &    \textbf{0.613}      & 0.666   & 0.568  \\
      &  cpt.np      & --    &  0.000     & 0.145  & 0.204    &   \bf{0.651}          &  0.747  &  \underline{0.690} \\ 
      \bottomrule
\end{longtable}}
\endgroup

\begingroup
{\small
\setlength{\tabcolsep}{3pt}
\setlength{\LTcapwidth}{\textwidth}
\begin{longtable}{c c ccccc cc}
\caption{We report the distribution of the estimated number of change points and the average CM and VM over 1000 realisations when $n=2000$. The modal value of $\wh q - q$ in each row is given in bold.
Also, the best performance for each metric is underlined for each scenario.}
\label{n2000}
\endfirsthead
\endhead
\toprule	
&& \multicolumn{5}{c}{$\wh{q} - q$} &   &   \\ 
 Model       & Method      &  $\leq-2$     & $-1$        & $\mathbf{0}$  & $1$    & $\geq 2$    & CM & VM       \\ 
 \cmidrule(lr){1-2} \cmidrule(lr){3-7} \cmidrule(lr){8-9}
\ref{model-n3}   &  NP-MOJO-$0$   &  --      & --   &  \textbf{0.903}    & 0.090      &    0.007          & --  & --   \\
&  NP-MOJO-$1$     &  --   & --     & \textbf{0.900}   & 0.090         &  0.010          & --   & --      \\
      &  NP-MOJO-$2$   &  --      &    --   & \textbf{0.887}  & 0.103     &   0.010          &  -- & --    \\ 
            &  NP-MOJO-$\mc L$   &  --      &    --  & \textbf{0.860}  & 0.128     &   0.012          &  --& --   \\ \cmidrule(lr){3-7} \cmidrule(lr){8-9}
      &  E-Divisive      & --     &  --  & 0.000  & 0.000        &    \textbf{1.000}         & -- & --   \\
      &  KCPA      &  --    & --     & \textbf{0.975}  &  0.016      &     0.009        &  --  & -- \\      
      &  NWBS      &  --      &  --    &  0.001 &   0.001     &    \textbf{0.998}        &  --  & --\\
      &  cpt.np      & --     &  --     &  0.000 & 0.000       &   \bf{1.000}         & --   & -- \\  \cmidrule(lr){1-2} \cmidrule(lr){3-7} \cmidrule(lr){8-9}
(A1) &  NP-MOJO-$0$ &  0.001      & 0.058    &  \textbf{0.940}    &  0.001       & 0.000             &  0.957 &  0.956      \\
&  NP-MOJO-$1$  & 0.001     & 0.017     & {\textbf{0.982}}   & 0.000          &   0.000         & 0.972   & 0.966        \\
      &  NP-MOJO-$2$ &   0.000     &  0.014     & \textbf{0.986}  & 0.000        &   0.000          & 0.972  & 0.966     \\
       &  NP-MOJO-$\mc L$   &  0.000      &    0.008   & {\textbf{0.992}}  & 0.000     &   0.000          &  0.972 & 0.965    \\  \cmidrule(lr){3-7} \cmidrule(lr){8-9}
      &  E-Divisive      &  0.000      &   0.000    &  \textbf{0.905}  &    0.067    &   0.028         & 0.980 &  0.976 \\
      &  KCPA      &  0.000     &   0.000   & \textbf{0.977} &  0.020     & 0.003           & {0.980}   &  {0.974} \\      
      &  NWBS      &  0.000    & 0.000     & \textbf{0.966}  &   0.021    &   0.013          &  0.976  & 0.970 \\
      &  cpt.np      &  0.000      &  0.000    & \textbf{0.775}  & 0.186        &  0.039           & 0.972   & 0.970 \\
      &  PELT      &  0.000     &  0.000     & \underline{\textbf{1.0000}}  &    0.000    &    0.000         & \underline{0.985}   & \underline{0.980} \\
      &  WCM.gSa      &  0.000      &  0.000     & \textbf{0.962}  & 0.011        &     0.027        &  0.982  & 0.977 \\     \cmidrule(lr){1-2} \cmidrule(lr){3-7} \cmidrule(lr){8-9}
\ref{model-b5}     &  NP-MOJO-$0$    &    0.000    &  0.010  &  {\textbf{0.990}}    &  0.000      &    0.000          &   \underline{0.979} & \underline{0.973}  \\
&  NP-MOJO-$1$  &  0.004    &  0.066    & \textbf{0.930}   &    0.000     &  0.000          & 0.962   & 0.962        \\
      &  NP-MOJO-$2$  & 0.002       &  0.053     & {\textbf{0.945}}  &  0.000     &  0.000      & 0.965  & 0.963     \\
       &  NP-MOJO-$\mc L$   &  0.000      &    0.005   & \underline{\textbf{0.995}}  & 0.000     &   0.000          &  0.978 & 0.972    \\ \cmidrule(lr){3-7} \cmidrule(lr){8-9}
      &  E-Divisive      &  \textbf{0.664}      &   0.116   & 0.172  & 0.044       &  0.004         &  0.540 & 0.580  \\
      &  KCPA      &  0.068      &  0.000    &  \textbf{0.921} &    0.009           &  0.002  & 0.936 & 0.926 \\      
      &  SBS      &  0.431      &   0.001   & \textbf{0.568}   & 0.000       &  0.000           & 0.751   & 0.781 \\  \cmidrule(lr){1-2} \cmidrule(lr){3-7} \cmidrule(lr){8-9}
\ref{model-c1}    &  NP-MOJO-$0$  &  --      &  --  &   {\textbf{0.874}}   &    0.117    &  0.009            & --  &  --     \\
&  NP-MOJO-$1$  &  0.000    & 0.000     & {\textbf{0.996}}   & 0.004          &  0.000          &  0.984  & 0.976         \\
      &  NP-MOJO-$2$ &  --    & --     & \textbf{0.868}   & 0.123          &  0.009          &  --  & --      \\ 
       &  NP-MOJO-$\mc L$   &  0.000      &    0.000   & \underline{\textbf{0.999}}  & 0.001     &   0.000          &  \underline{0.984} & \underline{0.976}    \\ \cmidrule(lr){3-7} \cmidrule(lr){8-9}
      &  WBSTS      &   0.000    &   0.000    &  0.281 &  0.312      & \textbf{0.407}        &  0.914 & 0.933  \\ \cmidrule(lr){1-2} \cmidrule(lr){3-7} \cmidrule(lr){8-9} 
    \ref{model-d3}     &  NP-MOJO-$0$  &  0.012      &  0.151  &   {\textbf{0.815}}   &    0.022    &  0.000            & 0.914  &  0.919     \\
&  NP-MOJO-$1$  &  0.016   & 0.156     & \textbf{0.798}   & 0.030         &  0.000          &  {0.909}  & {0.916}         \\
      &  NP-MOJO-$2$ &  0.042    & 0.216    & {\textbf{0.708}}   & 0.034         &  0.000        &  {0.875} & {0.893}      \\ 
       &  NP-MOJO-$\mc L$   &  0.007      &    0.098   & \underline{\textbf{0.895}}  & 0.000     &   0.000         &  \underline{0.923} & \underline{0.923}    \\ \cmidrule(lr){3-7} \cmidrule(lr){8-9}
      &  E-Divisive      &   0.002    &   0.001    &  0.037 &  0.046     &  \textbf{0.914}          &  0.709 & 0.793 \\
      &  KCPA      &  0.302      &  0.004     &  \textbf{0.610}  &    0.069    &    0.015         & 0.751   & 0.715 \\
            &  NWBS      &  0.047      &  0.028    &  0.107 &   0.132     &    \textbf{0.686}      & 0.730   & 0.790 \\
      &  cpt.np      & 0.000     & 0.000      &  0.011  & 0.022       &   \bf{0.967}          &   0.755 & 0.829 \\ 
      \bottomrule
\end{longtable}}
\endgroup

\begingroup
{\small
\setlength{\tabcolsep}{3pt}
\setlength{\LTcapwidth}{\textwidth}
\begin{longtable}{c c ccccc cc}
\caption{We report the distribution of the estimated number of change points and the average CM and VM over 1000 realisations when $n=10000$. The modal value of $\wh q - q$ in each row is given in bold.
 Also, the best performance for each metric is underlined for each scenario.}
\label{n10000}
\endfirsthead
\endhead
\toprule	
&& \multicolumn{5}{c}{$\wh{q} - q$} &   &   \\ 
 Model       & Method      &  $\leq-2$     & $-1$        & $\mathbf{0}$  & $1$    & $\geq 2$    & CM & VM       \\ 
 \cmidrule(lr){1-2} \cmidrule(lr){3-7} \cmidrule(lr){8-9}
\ref{model-n3}   &  NP-MOJO-$0$   &  --      & --   &  \textbf{0.961}    & 0.037      &    0.002          & --  & --   \\
&  NP-MOJO-$1$     &  --   & --     & \textbf{0.956}   & 0.043         &  0.001          & --   & --      \\
      &  NP-MOJO-$2$   &  --      &    --   & \textbf{0.962}  & 0.037     &   0.001          &  -- & --    \\ 
            &  NP-MOJO-$\mc L$   &  --      &    --  & \textbf{0.943}  & 0.054    &   0.003          &  --& --   \\ \cmidrule(lr){3-7} \cmidrule(lr){8-9}
      &  KCPA      &  --    & --     & 0.000  &  \textbf{0.891}      &     0.109        &  --  & -- \\      
      &  cpt.np      & --     &  --     &  0.000  & 0.000       &   \bf{1.000}          & --   & -- \\  \cmidrule(lr){1-2} \cmidrule(lr){3-7} \cmidrule(lr){8-9}
(A1) &  NP-MOJO-$0$ &  0.001      & 0.015    &  \textbf{0.984}    &  0.000       & 0.000             &  0.979 &  0.982      \\
&  NP-MOJO-$1$  & 0.000     & 0.001     & {\textbf{0.999}}   & 0.000          &   0.000         & 0.985   & 0.986        \\
      &  NP-MOJO-$2$ &   0.000     &  0.000     & \textbf{1.000}  & 0.000        &   0.000          & 0.985  & 0.986     \\
       &  NP-MOJO-$\mc L$   &  0.000      &    0.000  & {\textbf{1.000}}  & 0.000    &   0.000          &  0.984 & 0.985    \\  \cmidrule(lr){3-7} \cmidrule(lr){8-9}
      &  KCPA      &  0.000     &   0.000   & \textbf{0.972} &  0.026     & 0.002           & {0.988}   &  {0.989} \\      
      &  cpt.np      &  0.000      &  0.000    & \textbf{0.669}  & 0.257        &  0.074           & 0.983   & 0.987 \\
      &  PELT      &  0.000   &  0.000     & \underline{\textbf{1.000}}  &    0.000    &    0.000         & \underline{0.991}   & \underline{0.991} \\
      &  WCM.gSa      &  0.000      &  0.000    & \textbf{0.951}  & 0.013        &     0.036       &  0.990  & 0.990 \\     \cmidrule(lr){1-2} \cmidrule(lr){3-7} \cmidrule(lr){8-9}
\ref{model-b5}     &  NP-MOJO-$0$    &    0.000    &  0.001  &  \textbf{0.999}    &  0.000     &    0.000          &   {0.989} & 0.989  \\
&  NP-MOJO-$1$  &  0.001    &  0.008    & \bf{0.991}   &    0.000     &  0.000  & 0.985  & 0.987        \\
      &  NP-MOJO-$2$  & 0.001       & 0.006    & \bf{0.993}  &  0.000     &  0.000      & 0.985  & 0.987     \\
       &  NP-MOJO-$\mc L$   &  0.000      &    0.000  & \underline{\textbf{1.000}}  & 0.000     &   0.000          &  {0.988} & {0.988}    \\ \cmidrule(lr){3-7} \cmidrule(lr){8-9}
      &  KCPA      & 0.000      &  0.000    &  \textbf{0.967} &    0.032    &  0.001  & \underline{0.992} & \underline{0.992}\\  \cmidrule(lr){1-2} \cmidrule(lr){3-7} \cmidrule(lr){8-9}
\ref{model-c1}    &  NP-MOJO-$0$  &  --      &  --  &   {\textbf{0.925}}   &    0.072    &  0.003           & --  &  --     \\
&  NP-MOJO-$1$  &  0.001    & 0.000    & {\textbf{0.998}}   & 0.002          &  0.000          &  \underline{0.990}  & \underline{0.990}         \\
      &  NP-MOJO-$2$ &  --   & --     & \textbf{0.927}   & 0.071         &  0.002          &  --  & --      \\ 
       &  NP-MOJO-$\mc L$   &  0.000    &    0.000   & \underline{\textbf{1.000}}  & 0.000    &   0.000          &  \underline{0.990} & \underline{0.990}    \\ \cmidrule(lr){3-7} \cmidrule(lr){8-9}
      &  WBSTS      &   0.000 &   0.000   &  0.059 &  0.152      &  \textbf{0.789}           &  0.940 & 0.970  \\ \cmidrule(lr){1-2} \cmidrule(lr){3-7} \cmidrule(lr){8-9} 
    \ref{model-d3}     &  NP-MOJO-$0$  &  0.134      &  0.284  &   {\textbf{0.569}}   &    0.013    &  0.000            & 0.924  &  0.960     \\
&  NP-MOJO-$1$  &  0.176   & 0.294     & \textbf{0.525}   & 0.005         &  0.          &  {0.917}  & {0.957}         \\
      &  NP-MOJO-$2$ & \textbf{0.413}    & 0.312    & 0.272   & 0.003         &  0.000        &  {0.864} & {0.939}      \\ 
       &  NP-MOJO-$\mc L$   &  0.061     &    0.235   & {\textbf{0.704}}  & 0.000     &   0.000          &  0.938 & 0.963    \\ \cmidrule(lr){3-7} \cmidrule(lr){8-9}
      &  KCPA      &  0.000      &  0.000     &  \underline{\textbf{0.872}}  &    0.112   &    0.016         & \underline{0.965}   & \underline{0.972} \\
      &  cpt.np      & 0.000    &  0.000    & 0.000  & 0.000    &   \bf{1.000}          &  0.803  &  0.907 \\ 
      \bottomrule
\end{longtable}}
\endgroup

\subsection{\tcb{Multiscale NP-MOJO}}\label{sec:sim-multiscale}

In this section we investigate the performance of the multiscale, multi-lag NP-MOJO procedure discussed in Section~\ref{sec:multiscale}. We use identical tuning parameters as in the main simulation study, and set the bottom-up merging parameter $C = 0.8$. Following \cite{mcgonigle2023robust}, we generate $\mc G$ as a sequence of Fibonacci numbers, with $\mc G = \{G_m, \, 1 \le m \le 4: \, G_1 < \ldots < G_4\}$ where $G_m = G_{m - 1} + G_{m - 2}$ for $m \ge 2$ with $G_0 = G_1 = 60$. We use a relatively large finest bandwidth $G_1$ since, as observed in \cite{mcgonigle2023robust}, bottom-up merging has a tendency to return false positives as it accepts all estimators from the finest bandwidth. We \tcr{first consider the following scenarios, all with $n=1000$}:

\begin{enumerate}[label=(M\arabic*)] 
\item[(B2)] Model \ref{model-b2} from the main simulation study.
\item[(D3)] Model \ref{model-d3} from the main simulation study.
\item \label{model-m1} $X_t$ follows the change in mean Model~\ref{eq:mean:model}, with $q = 3$, $(\theta_1, \theta_2, \theta_3) = (80, 250, 600)$, $(\mu_0, \mu_1, \mu_2, \mu_3) = (0, 1.6, 0.6, 1.2)$, and $\vep_t = 0.3 \vep_{t-1} + W_t$, $W_t \sim_{\text{i.i.d.}} \mc N(0, 1-0.3^2)$. 
\item \label{model-m2} $X_t = X^{(j)}_t = a_j  X^{(j)}_{t-1} + \vep_t$ for $\theta_j + 1 \le t \le \theta_{j + 1}$, where $q = 2$, $(\theta_1, \theta_2) = (500,900)$, $(a_0, a_1, a_2) = (0.3, 0.8, -0.8)$.
\item \label{model-m3} $X_t = \sum_{j = 0}^q \mu_j \mathbb{I}{\{\theta_j + 1 \le t \le \theta_j\}} + \vep_t + \sum_{j = 0}^q B_j \mathbb{I}{\{\theta_j + 1 \le t \le \theta_{j + 1}\}}  \vep_{t-1} $, where $q = 2$, $(\theta_1, \theta_2) = (150,500)$, $(\mu_0, \mu_1, \mu_2) = (\mbf{0}, \mbf{0.7}, \mbf{0.7})$, $B_0 = B_1 = \begin{psmallmatrix}1 & 0.1 \\ 0.1 & 1\end{psmallmatrix}$ and $B_2 = \begin{psmallmatrix}-1 & 0.1\\ 0.1 & -1\end{psmallmatrix}$, $\vep_t \sim_{\text{i.i.d.}} \mc N_2(\mbf 0, \mbf I)$.
\end{enumerate}

The first two scenarios \tcr{contain evenly-spaced change points and are taken from the main simulation study, in order to investigate the performance of the multiscale NP-MOJO in non-multiscale scenarios}. The last three scenarios incorporate multiscale change point scenarios \tcr{with uneven spacing between neighbouring segments}. Model~\ref{model-m1} considers mean change points with autocorrelated noise, Model~\ref{model-m2} has two changes in the parameter of an AR(1) process, and Model~\ref{model-m3} has one change in a vector moving average process and one change in mean vector. Results are given in Table~\ref{table-multiscale}. In the equispaced change point scenarios~\ref{model-b2} and~\ref{model-d3}, the performance of multiscale NP-MOJO is similar to single scale NP-MOJO. In the multiscale scenarios, the multiscale extension of NP-MOJO shows good adaptivity, enabling it to detect both larger changes over shorter time scales, and smaller changes over longer time scales. For example, in scenario~\ref{model-m2}, the first change point is harder to detect, and is best suited to be detected at bandwidth $G = 300$, whilst the second change point is easier to detect, and is best-suited to be detected at bandwidth $G=60$ due to it occurring near the end of the time series.
\begingroup
{\small
\setlength{\tabcolsep}{3pt}
\setlength{\LTcapwidth}{\textwidth}
\begin{longtable}{c ccccc cc}
\caption{We report the distribution of the estimated number of change points and the average CM and VM over 1000 realisations for the multiscale NP-MOJO procedure when $n=1000$. The modal value of $\wh q - q$ in each row is given in bold.}
\label{table-multiscale}
\endfirsthead
\endhead
\toprule	
& \multicolumn{5}{c}{$\wh{q} - q$} &   &   \\ 
 Model          &  $\leq-2$     & $-1$        & $\mathbf{0}$  & $1$    & $\geq 2$    & CM & VM       \\ 
 \cmidrule(lr){1-1} \cmidrule(lr){2-6} \cmidrule(lr){7-8}
\ref{model-b2}     &  0.000      & 0.006   &  \textbf{0.978}    & 0.016     &    0.000       & 0952 & 0.934
      \\
       \cmidrule(lr){1-1} \cmidrule(lr){2-6} \cmidrule(lr){7-8}
\ref{model-d3}     &  0.000      & 0.002   &  \textbf{0.922}    & 0.076     &    0.000       & 0.929 & 0.894
      \\
       \cmidrule(lr){1-1} \cmidrule(lr){2-6} \cmidrule(lr){7-8}
\ref{model-m1}     &  0.001      & 0.126   &  \textbf{0.856}    & 0.017     &    0.000       & 0.894 & 0.878
      \\
       \cmidrule(lr){1-1} \cmidrule(lr){2-6} \cmidrule(lr){7-8}
\ref{model-m2}     &  0.000      & 0.054   &  \textbf{0.802}    & 0.139    &    0.005       & 0.874 & 0.815
      \\
       \cmidrule(lr){1-1} \cmidrule(lr){2-6} \cmidrule(lr){7-8}
\ref{model-m3}     &  0.000      & 0.107   &  \textbf{0.854}    & 0.039     &    0.000       & 0.929 & 0.901
      \\\bottomrule
\end{longtable}}
\endgroup

\tcr{We next consider the following three multiscale change point scenarios with $n=2000$:}

\begin{enumerate}[label=(M\arabic*),resume]
\item \label{model-m4} $X_t$ follows the change in mean model as in~\ref{model-m1}, with $q = 5$, \\$(\theta_1, \ldots, \theta_5) = (500, 1000, 1150, 1550, 1900)$, $(\mu_0, \mu_1, \ldots ,  \mu_6) = (0, 0.9, 2.2, 1.1, 0, 1.5)$.
\item \label{model-m5} $X_t$ has changes in both mean and dependence with $q = 6$: mean changes $(\theta_1, \theta_2, \theta_3 , \theta_4) = (100, 200, 600, 1400)$, $(\mu_0, \mu_1, \ldots ,  \mu_5) = (0, 1.5, 0, 0.9, -0.3)$, and $\vep_t = \vep^{(j)}_t = a_j  X^{(j)}_{t-1} + \vep_t$ for $\theta'_j + 1 \le t \le \theta'_{j + 1}$ $(\theta'_1, \theta'_2) = (1000, 1800)$, $(a_0, a_1, a_2) = (-0.7, 0.7, -0.8)$.
\item \label{model-m6} $X_t$ follows the change in covariance model as in~\ref{model-b4}, with $q = 5$, $(\theta_1, \theta_2, \theta_3, \theta_4, \theta_5) = (150, 300, 800, 1300, 1600)$, and $\Sigma_1 = \Sigma_3 = \Sigma_5 =  \begin{psmallmatrix}1 & 0.6\\ 0.6 & 1\end{psmallmatrix}$, $\Sigma_2 = \begin{psmallmatrix}1 & -0.6\\ -0.6 & 1\end{psmallmatrix}$, and $\Sigma_4 = \Sigma_6 = \begin{psmallmatrix}1 & -0.2\\ -0.2 & 1 \end{psmallmatrix}$.
\end{enumerate}

\tcr{Model~\ref{model-m4} follows the same data generating process as Model~\ref{model-m1}, whilst Model~\ref{model-m5} has two changes in the parameter of an AR(1) process and four changes in the mean. Model~\ref{model-m6} has five changes in the covariance matrix of a Gaussian bivariate time series. Results are given in Table~\ref{table-multiscale2}.}

\begingroup
{\small
\setlength{\tabcolsep}{3pt}
\setlength{\LTcapwidth}{\textwidth}
\begin{longtable}{c ccccc cc}
\caption{We report the distribution of the estimated number of change points and the average CM and VM over 1000 realisations for the multiscale NP-MOJO procedure when $n=2000$. The modal value of $\wh q - q$ in each row is given in bold.}
\label{table-multiscale2}
\endfirsthead
\endhead
\toprule	
& \multicolumn{5}{c}{$\wh{q} - q$} &   &   \\ 
 Model          &  $\leq-2$     & $-1$        & $\mathbf{0}$  & $1$    & $\geq 2$    & CM & VM       \\ 
 \cmidrule(lr){1-1} \cmidrule(lr){2-6} \cmidrule(lr){7-8}
\ref{model-m4}     &  0.008      & 0.154   &  \textbf{0.821}    & 0.015    &    0.000      & 0.942 & 0.945
      \\
       \cmidrule(lr){1-1} \cmidrule(lr){2-6} \cmidrule(lr){7-8}
\ref{model-m5}     &  0.000     & 0.061   &  \textbf{0.856}    & 0.080     &    0.003       & 0.936 & 0.946
      \\
       \cmidrule(lr){1-1} \cmidrule(lr){2-6} \cmidrule(lr){7-8}
\ref{model-m6}     &  0.001     & 0.080   &  \textbf{0.913}    & 0.006     &    0.000       & 0.956 & 0.955
      \\\bottomrule
\end{longtable}}
\endgroup

\tcr{Lastly, we consider the following three scenarios with $n=10000$ to investigate the ability of multiscale NP-MOJO's ability to detect multiscale change points in longer time series.  Each scenario contains change points occurring over both short and long time periods.}

\begin{enumerate}[label=(M\arabic*),resume]
\item \label{model-m7} $X_t$ follows the change in mean Model~\ref{eq:mean:model}, with $q = 8$, \\$(\theta_1, \ldots, \theta_8) = (1000, 2000, 2150, 2800, 3650, 4650, 5150, 5550)$, \\$(\mu_0, \mu_1, \ldots ,  \mu_9) = (0, 1, 2.6, 1.1, 0, 1, -0.2, 1, 0)$, and $\vep_t  \sim_{\text{i.i.d.}} \mc N(0, 1)$. 
\item \label{model-m8} $X_t = X^{(j)}_t = a_j  X^{(j)}_{t-1} +  b_j X^{(j)}_{t-2} + \vep_t$ for $\theta_j + 1 \le t \le \theta_{j + 1}$, where $q = 5$, $(\theta_1, \ldots, \theta_5) = (1000, 1400, 5000, 9000, 9400)$, $(a_0, a_1, a_2, a_3, a_4, a_5, a_6) = (0.8, -0.8, 0.8, -0.2, -0.2, -0.2)$, $(b_0, b_1, b_2, b_3, b_4, b_5, b_6) = (-0.2, -0.2, -0.2, 0.6, -0.6, 0.6)$.
\item \label{model-m9}  $X_t = \vep_t + \sum_{j = 0}^q B_j \mathbb{I}{\{\theta_j + 1 \le t \le \theta_{j + 1}\}}  \vep_{t-1} $, where $q = 5$, $(\theta_1, \ldots, \theta_5) = \\ (600, 2000, 4000, 5300, 5600, 8000)$, $B_0 = B_2 = B_6 = \begin{psmallmatrix}-1 & 0.4 \\ 0.5 & -1\end{psmallmatrix}$ and $B_1 = B_3 = B_5 =  \begin{psmallmatrix}1 & 0.4\\ 0.4 & 1\end{psmallmatrix}$, $B_4 = \begin{psmallmatrix}1.8 & 0.1 \\ 0.1 & 1.8 \end{psmallmatrix}$, $\vep_t \sim_{\text{i.i.d.}} \mc N_2(\mbf 0, \mbf I)$.
\end{enumerate}

\tcr{Model~\ref{model-m7} follows a change in mean model with independent noise with highly uneven distribution of change points. Model~\ref{model-m8} follows an AR(2) process that exhibits changes at both even and odd lags, whilst Model~\ref{model-m9} follows a vector moving average process with change at lag $\ell = 1$. Results are given in Table~\ref{table-multiscale3}, where multiscale NP-MOJO demonstrates good adaptivity in the multiscale change point setting in longer times series. Overall, the multiscale extension to multi-lag NP-MOJO shows promising performance and provides a natural avenue for future research.}

\begingroup
{\small
\setlength{\tabcolsep}{3pt}
\setlength{\LTcapwidth}{\textwidth}
\begin{longtable}{c ccccc cc}
\caption{We report the distribution of the estimated number of change points and the average CM and VM over 1000 realisations for the multiscale NP-MOJO procedure when $n=10000$. The modal value of $\wh q - q$ in each row is given in bold.}
\label{table-multiscale3}
\endfirsthead
\endhead
\toprule	
& \multicolumn{5}{c}{$\wh{q} - q$} &   &   \\ 
 Model          &  $\leq-2$     & $-1$        & $\mathbf{0}$  & $1$    & $\geq 2$    & CM & VM       \\ 
 \cmidrule(lr){1-1} \cmidrule(lr){2-6} \cmidrule(lr){7-8}
\ref{model-m7}     &  0.001      & 0.032   &  \textbf{0.950}    & 0.017    &    0.000      & 0.989 & 0.988
      \\
       \cmidrule(lr){1-1} \cmidrule(lr){2-6} \cmidrule(lr){7-8}
\ref{model-m8}     &  0.001     & 0.102   &  \textbf{0.676}    & 0.184     &    0.037       & 0.919 & 0.938
      \\
       \cmidrule(lr){1-1} \cmidrule(lr){2-6} \cmidrule(lr){7-8}
\ref{model-m9}     &  0.036     & 0.250   &  \textbf{0.661}    & 0.052     &    0.001       & 0.955 & 0.972
      \\\bottomrule
\end{longtable}}
\endgroup

\subsection{\tcr{NP-MOJO with adaptive lag selection}}\label{sec:sim-lag}

In this section we perform a small simulation study to illustrate the effectiveness of the NP-MOJO procedure with adaptive lag selection, as discussed in Section~\ref{sec:lag-select}.  We use the adaptive lag selection method with initial set of lags $\tilde{\mc L} = \{ 0, 1, 2\}$, and consider scenarios~\ref{model-n3}, \ref{model-c1}, and  \ref{model-c2}, from the main simulation study. Recall that scenario~\ref{model-n3} follows an AR(1) process with no change points, \ref{model-c1} follows an AR(1) process with two change points at all odd lags, and \ref{model-c2} follows an MA(2) process with two change points at lag $2$ only.

The results are given in Table~\ref{table-adaptive-lag}. In general, the adaptive lag version of multi-lag NP-MOJO performs similarly to the version using the fixed set of lags $\mc L = \{ 0, 1, 2 \}$ (c.f. the result for scenario~\ref{model-n3} in Table~\ref{null-table}, and results for scenarios \ref{model-c1} and \ref{model-c2} in Table~\ref{multcpt-c}). Under the null model~\ref{model-n3} there is a slight increase in the number of false positives, whilst in scenarios~\ref{model-c1} and~\ref{model-c2} there is a very slight degradation in performance in terms of the segmentation scores. Overall the adaptive lag selection methods provide users of the method a semi-automatic way of selecting the set of lags, whilst offering similar performance to the fixed lag method.

\begingroup
{\small
\setlength{\tabcolsep}{3pt}
\setlength{\LTcapwidth}{\textwidth}
\begin{longtable}{c ccccc cc}
\caption{We report the distribution of the estimated number of change points and the average CM and VM over 1000 realisations for the multi-lag NP-MOJO procedure with adaptive lag selection. The modal value of $\wh q - q$ in each row is given in bold.}
\label{table-adaptive-lag}
\endfirsthead
\endhead
\toprule	
& \multicolumn{5}{c}{$\wh{q} - q$} &   &   \\ 
 Model          &  $\leq-2$     & $-1$        & $\mathbf{0}$  & $1$    & $\geq 2$    & CM & VM       \\ 
 \cmidrule(lr){1-1} \cmidrule(lr){2-6} \cmidrule(lr){7-8}
\ref{model-n3}     &  --      &  -- &  \textbf{0.802}    & 0.177    &    0.021     & -- & --
      \\
       \cmidrule(lr){1-1} \cmidrule(lr){2-6} \cmidrule(lr){7-8}
       \ref{model-c1}     &  0.000     & 0.002   &  \textbf{0.985}    & 0.013     &    0.000       & 0.980  & 0.962

      \\
       \cmidrule(lr){1-1} \cmidrule(lr){2-6} \cmidrule(lr){7-8}
\ref{model-c2}     &  0.001     & 0.051   &  \textbf{0.943}    & 0.005     &    0.000       & 0.948 & 0.924
      \\\bottomrule
\end{longtable}}
\endgroup

\clearpage

\section{Proofs of main results}\label{sec:proofs}

\subsection{\tcb{Proof of Lemma~\ref{lemma:discrepancy}}}

For any $\ell \geq 0$, if $(X_1^{(j)}, X^{(j)}_{1+\ell}) \stackrel{d}{=} (X_1^{(j-1)} , X_{1+\ell}^{(j-1)})$, then $\phi^{(j)}_\ell(u, v) = \phi^{(j-1)}_\ell(u, v)$ for all $(u, v)$, which implies that $\vert \phi^{(j)}_\ell (u,v) - \phi^{(j-1)}_\ell (u,v) \vert^2 \equiv 0$, and hence $d^{(j)}_\ell = 0$ in~\eqref{eq:test-stat-char}. Conversely, suppose that $d_{\ell}^{(j)} = 0$. Then, $\phi^{(j)}_\ell (u,v) - \phi^{(j-1)}_\ell (u,v) = 0$ a.e.\ since $w(u,v) > 0$ for $u, v \neq 0$, and hence $(X_1^{(j)}, X^{(j)}_{1+\ell}) \stackrel{d}{=} (X_1^{(j-1)} , X_{1+\ell}^{(j-1)})$.

\subsection{Proof of Lemma~\ref{lemma:weight-int-identities}}

We first consider the integrand term in~\eqref{eq:test-stat-char} involving the characteristic functions. We have that
\begin{align*}
\l\vert \phixj - \phiyj \r\vert^2 &= \phixj \overline{\phixj} + \phiyj \overline{\phiyj} \\
& \quad - \phixj \overline{\phiyj} - \overline{\phixj} \phiyj \\
& =: A + B - C - D.
\end{align*}
Then,
\begin{align*}
 A &= \mathbb{E} \left \{ \exp \l( \imath \langle u, X_1^{(j)} \rangle + \imath \langle v, X^{(j)}_{1+\ell} \rangle \r) \right\}  \mathbb{E} \left\{ \exp (-\imath \langle u,  \wt X^{(j)}_1 \rangle -\imath \langle v, \wt X^{(j)}_{1+\ell} \rangle ) \right\}  \\
 &= \mathbb{E} \left \{ \exp \l( \imath \langle u, X^{(j)}_1 - \wt X^{(j)}_{1} \rangle  + \imath \langle v,  X^{(j)}_{1+\ell} - \wt X^{(j)}_{1+\ell} \rangle   \r) \right\} .
\end{align*}
In a similar fashion,
\begin{align*}
 &B = \mathbb{E} \left \{ \exp \l( \imath \langle u, X_1^{(j-1)} - \wt X^{(j-1)}_{1} \rangle + \imath \langle v,  X_{1+\ell}^{(j-1)} - \wt X^{(j-1)}_{1+\ell} \rangle  \r) \right\}  , \\
 &C = \mathbb{E} \left \{ \exp \l( \imath \langle u, \wt X^{(j)}_1 - X^{(j-1)}_{1} \rangle + \imath \langle v,  \wt X^{(j)}_{1+\ell} - X^{(j-1)}_{1+\ell} \rangle   \r) \right\}  , \\
 &D =  \mathbb{E} \left \{ \exp \l( -\imath \langle u, \wt X^{(j)}_1 - X^{(j-1)}_{1} \rangle - \imath \langle v,  \wt X^{(j)}_{1+\ell} - X^{(j-1)}_{1+\ell} \rangle   \r) \right \} .
 \end{align*}
 Note that since $d^{(j)}_\ell$ is real, any term of the form $\exp(\imath z)$ with $z \in \mathbb{R}$ can be replaced by $\cos z$. Therefore, we have that $C=D$, and we can re-write the integral~\eqref{eq:test-stat-char} in terms of cosines as
\begin{align*}
d_{\ell}^{(j)} =& \, \int_{\mathbb{R}^{p}} \int_{\mathbb{R}^{p}}  \mathbb{E} ( \text{COS} (u,v) ) w(u,v) du dv,    \text{ \ where}
\\
\text{COS}(u,v) =& \, \cosux \cosvxj   \\
& + \cosuy \cosvyj \\
& - 2 \cosuxy \cosvxyj .
\end{align*}
Under the assumptions of Lemma~\ref{lemma:weight-int-identities}~\ref{lemma:weight-int1}, for weight $w_1$ we obtain 
\begin{align*}
d_{\ell}^{(j)} &=\int_{\mathbb{R}^{p}} \int_{\mathbb{R}^{p}} \mathbb{E} \{ \text{COS} (u,v) \}  w_1(u,v) du dv \\[1ex]
&= \int_{\mathbb{R}^{p}} \int_{\mathbb{R}^{p}} \mathbb{E} \left\{  \cosux \cosvxj   \right\} w_1(u,v) du dv \\[1ex]
&+\int_{\mathbb{R}^{p}} \int_{\mathbb{R}^{p}}  \mathbb{E} \left\{  \cosuy \cosvyj    \right\} w_1(u,v) du dv \\[1ex]
&-  \int_{\mathbb{R}^{p}} \int_{\mathbb{R}^{p}} 2 \mathbb{E} \left\{  \cosuxy \cosvxyj    \right\} w_1(u,v) du dv \\[1ex]
&=  \mathbb{E} \left\{ h_1 \l( Y_1^{(j)},\wt Y_1^{(j)} \r) \right\} +  \mathbb{E} \left\{ h_1 \l(Y_1^{(j-1)}, \wt Y_1^{(j-1)} \r) \right\} - 2 \mathbb{E} \left\{ h_1 \l( \wt Y_1^{(j)}, Y_1^{(j-1)} \r) \right\}.
\end{align*}
The integral and expectation can be swapped by applying Fubini's theorem, due to finiteness of the expectation. The final line follows from an application of Lemma~\ref{integral-lemma}. An analogous argument for Lemma~\ref{lemma:weight-int-identities}~\ref{lemma:weight-int2}, using Lemma~\ref{integral-lemma2}, yields
\begin{align*}
d_{\ell}^{(j)} = \mathbb{E} \left\{ h_2 \l( Y_1^{(j)}, \wt Y_1^{(j)} \r) \right\} +  \mathbb{E} \left\{ h_2 \l(Y_1^{(j-1)}, \wt Y_1^{(j-1)} \r) \right\} - 2 \mathbb{E} \left\{ h_2 \l(\wt Y_1^{(j)}, Y_1^{(j-1)} \r) \right\}
\end{align*}
for weight $w_2$. To prove Lemma~\ref{lemma:weight-int-identities}~\ref{lemma:weight-int3} for weight $w_3$, we re-write the integral~\eqref{eq:test-stat-char} to obtain
\begin{align*}
d_{\ell}^{(j)} &= \int_{\mathbb{R}^{p}} \int_{\mathbb{R}^{p}}\mathbb{E} \{ \text{COS} (u,v) \}  w_3(u,v) du dv \\[1ex]
&=  \int_{\mathbb{R}^{p}} \int_{\mathbb{R}^{p}} 2 \mathbb{E} \left\{ 1 - \cosuxy \cosvxyj     \right\} w_3(u,v) du dv \\[1ex]
&- \int_{\mathbb{R}^{p}} \int_{\mathbb{R}^{p}}  \mathbb{E} \left\{ 1 - \cosux \cosvxj    \right\} w_3(u,v) du dv \\[1ex]
&- \int_{\mathbb{R}^{p}} \int_{\mathbb{R}^{p}} \mathbb{E} \left\{ 1 - \cosuy \cosvyj    \right\} w_3(u,v) du dv \\[1ex]
&= 2  \mathbb{E} \left\{ h_3 \l(\wt Y_1^{(j)}, Y_1^{(j-1)} \r) \right\} -  \mathbb{E} \left\{ h_3 \l( Y_1^{(j)},\wt Y_1^{(j)} \r) \right\} - \mathbb{E} \left\{ h_3 \l(Y_1^{(j-1)}, \wt Y_1^{(j-1)} \r) \right\}.
\end{align*}
The expectation can be swapped with the integral using Fubini's theorem, since $\mathbb{E} ( \Vert X_1^{(j)} \Vert^\gamma ) < \infty$. The final line follows from Lemma~\ref{integral-lemma3}. 

\subsection{Proof of Theorem~\ref{thm:consistency}}

The proof proceeds in three steps. 
Step~1 derives a bound on $\max_{G \leq k \leq n-G}  \vert \tstatkl -  \mc D_\ell (G,k) \vert$, with which Step~2 shows that exactly one change point is detected within $(G - \ell)$ time points from each $\theta_j, \, j \in \mc I_\ell$, and no other estimator is detected.
Then Step~3 derives the rate of estimation.

\paragraph{Step 1. } For any $G \le k \le n - G$, we have
\begin{align*}
T_\ell (G,k) - \mc D_\ell (G, k) = \left\{ T_\ell (G,k) - \mathbb{E} (\tstatkl) \right\} + \left\{ \mathbb{E} (\tstatkl) - \mc D_\ell (G,k) \right\}.  
\end{align*}
Lemma~\ref{uniform-expectation} shows that $\vert \mathbb{E} \{ \tstatkl \} - \mc D_\ell (G,k)  \vert = O (G^{-1/2})$, whilst Lemma~\ref{lemma:unif-prob} shows that
 for any $z \ge 1/\sqrt{G - \ell}$,

\begin{align*}
\p \l( \max_{G \leq k \leq n-G} \l\vert \tstatkl    -  \mathbb{E} \{ \tstatkl \} \r\vert > z \r)
\\ 
\leq  6 n G^2 \exp \l( - c_1  z^\gamma G^{\gamma} \r)    +  12nG \exp \l( - c_2  z^2 G \r).
\end{align*}
Therefore, combining the results of Lemma~\ref{uniform-expectation} and Lemma~\ref{lemma:unif-prob}, we obtain $\p(\mc E_{\ell, n}) \to 1$ as $n \to \infty$, where
\begin{align}
\label{eq:Op-bound}
\mc E_{\ell, n} &= \l\{ \max_{G \leq k \leq n-G} \l\vert \tstatkl - \mc D_\ell(G, k) \r\vert \leq \frac{c_0}{2} \sqrt{ \frac{\log (n)}{G}} \r\}
\end{align}
for large enough constant $c_0 > \sqrt{2 c_2^\prime}$. 

In the following steps, all the arguments are conditional on $\mc E_{\ell, n}$.

\paragraph{Step 2.} Consider $k$ satisfying $\min_{j \in \mc I_\ell} \vert k - \theta_{j} \vert \geq G - \ell$, for which $\mc D_\ell(G, k) = 0$.
Then provided that $c_\zeta > c_0/2$, we have
\begin{align*}
 \max_{k : \, \min_{j \in \mc I_\ell} |k - \theta_{j} \vert \geq G - \ell} \tstatkl  \leq  \max_{G \leq k \leq n-G} \l\vert \tstatkl - \mc D_\ell (G,k) \r\vert \leq \frac{c_0}{2}\sqrt{\frac{\log (n)}{G}}  <  \zeta_\ell (n,G).
\end{align*}
Therefore, no change point is detected more than $(G-\ell)$ time points away from any $\theta_j, \, j \in \mc I_\ell$, i.e.\ $\min_{\wh\theta \in \wh{\Theta}_\ell} \vert \wh{\theta} - \theta_{j} \vert < G-\ell$. 
We now consider some $\theta_j, \, j \in \mc I_\ell$.
By Lemma~\ref{lemma:eta-crit}~\ref{lemma:eta-crit-a}, we detect at least one estimator within $\lceil (1 - \eta)G \rceil$ points from $\theta_j$ by having
$\max_{k: \, \vert k - \theta_j \vert < (1 - \eta)(G - \ell)} T_\ell(G, k) > \zeta_\ell(n, G)$, and none is detected outside this interval.
Then Lemma~\ref{lemma:eta-crit}~\ref{lemma:eta-crit-b} shows that there exists a unique local maximiser of $T_\ell(G, k)$ within $\lfloor \eta G \rfloor$ time points from $\theta_j$ that meets the criterion in~\eqref{eq:mosum:est}. 
Since the lemma shows $\p(\mc S_{\ell, n}) \to 1$ and $\p(\wt{\mc S}_{\ell, n}) \to 1$ (see the lemma for their definitions), the above arguments hold for all $j \in \mc I_\ell$, such that we have $\wh{q}_\ell = q_\ell$.

\paragraph{Step 3.}  For each $j \in \mc I_\ell$, let $\wh\theta_j = \argmin_{\wh\theta \in \wh{\Theta}_\ell} \vert \wh\theta - \theta_j \vert$. Then from Step~2, $\vert \wh{\theta}_{j} -\theta_{j} \vert \le G - \ell$ such that
\begin{align*}
d^{(j)}_\ell - \frac{c_0}{2} \sqrt{\frac{\log (n)}{G}} \le T_\ell (G, \theta_{j} ) \le
T_\ell (G, \wh{\theta}_{j} ) \le
\l(\frac{G - \ell - \vert\widehat{\theta}_j - \theta_j \vert}{G-\ell} \r)^2  d^{(j)}_\ell +  \frac{c_0}{2} \sqrt{\frac{\log (n)}{G}}.
\end{align*}
From this, it follows that
\begin{align*}
& d^{(j)}_\ell   \frac{\vert \wh{\theta}_{j} -\theta_{j} \vert }{G-\ell}  < d^{(j)}_\ell  \frac{\vert \wh{\theta}_{j} -\theta_{j} \vert (  2G - 2 \ell -  \vert \wh{\theta}_{j} -\theta_{j} \vert )}{(G-\ell)^2} \le 
c_0 \sqrt{\frac{\log (n)}{G}},
\\
& \text{such that \ } d^{(j)}_\ell \vert \wh{\theta}_{j} -\theta_{j} \vert < c_0 \sqrt{G \log (n)}.
\end{align*}

\subsection{\tcr{Proof of Theorem~\ref{thm:consistency-marginal}}}

Firstly, note that by the same arguments as in the proof of Theorem~~\ref{thm:consistency}, the lag-$0$ NP-MOJO procedure estimates exactly one change point inside a $G$-distance of each true change point $\theta_j$, and no change points are detected outside a $G$-distance from any $\theta_j$, with probability tending to one.

{For some $1 \le j \le q$,} let us consider all $k$ such that $| k - \theta_j | < G$. First, consider the case that $k < \theta_j$; the case where $k > \theta_j$ follows analogously. Then, since the $X_t$ are independent and the kernel bounded, we have 
\begin{align*}
\mathbb{E}\{ T_0 (G, k) \} = \frac{\{ G - (\theta_j - k) \}^2}{G^2} d^{(j)}_0 + O(G^{-1}) 
\end{align*}
so that 
\begin{align}
\mathbb{E} \left\{ T_0(G, k) - T_0(G, \theta_j ) \right\} &= - \frac{\{  \theta_j - k \} \{2G - (\theta_j - k) \}}{G^2} d^{(j)}_0 + O (G^{-1}) 
\nonumber \\
& =: f(\theta_j, G , k) d^{(j)}_0 + O (G^{-1})
\label{eq:t0:bias}
\end{align}
uniformly in $k$ satisfying $| k - \theta_j | < G$. With $\ell = 0$, and denoting $h(X_r, X_s)$ by $h_{rs}$, we have 
\begin{align}
G^{2} [ T_0(G, k) - T_0(G, \theta_j ) ] &= \sum_{r,s = k - G + 1}^{k} h_{rs} + \sum_{r,s = k + 1}^{k + G} h_{rs} - 2 \sum_{r = k - G + 1}^k \sum_{s = k+1}^{k+G} h_{rs} 
\nonumber \\
& \  - \sum_{r,s = \theta_j - G + 1}^{\theta_j} h_{rs}- \sum_{r,s = \theta_j + 1}^{\theta_j + G} h_{rs} + 2 \sum_{r = \theta_j - G + 1}^{\theta_j} \sum_{s = \theta_j + 1}^{\theta_j + G}  h_{rs} 
\nonumber \\
&= \sum_{r,s = k - G + 1}^{\theta_j - G} h_{rs} + 2 \sum_{r = \theta_j - G + 1}^{k} \sum_{s = k - G + 1}^{\theta_j - G}  h_{rs} - 2 \sum_{r = k + 1}^{\theta_j} \sum_{s = k - G + 1}^{\theta_j - G}  h_{rs} 
\nonumber \\ 
& \ - 2 \sum_{r = \theta_j + 1}^{k + G} \sum_{s = k - G + 1}^{\theta_j - G}  h_{rs}  - \sum_{r, s = k + G + 1}^{\theta_j + G} h_{rs} - 2 \sum_{r = \theta_j  + 1}^{k + G} \sum_{s = k + G + 1}^{\theta_j + G }  h_{rs} 
\nonumber \\
& \ + 2 \sum_{r = \theta_j - G + 1}^{\theta_j} \sum_{s = k + G + 1}^{\theta_j + G}  h_{rs}  + 4 \sum_{r = \theta_j + 1}^{k + G} \sum_{s = k + 1}^{\theta_j} h_{rs} - 4 \sum_{r = k+ 1}^{\theta_j} \sum_{s = \theta_j - G + 1}^{k} h_{rs} 
\nonumber \\
& =: V_1 + 2 U_1 - 2 U_2  - 2 U _3 - V_2 - 2 U_4 + 2 U_5 + 4 U_6 - 4 U_7 ,
\nonumber
\end{align}
by symmetry of the kernel $h$. Therefore, we have that 
\begin{align}
T_0(G, k) - T_0(G, \theta_j )  &= G^{-2} \left(  V_1 + 2 U_1 - 2 U_2  - 2 U _3 - V_2 - 2 U_4 + 2 U_5 + 4 U_6 - 4 U_7 \right) 
\nonumber \\
& =: \frac{(\theta_j - k)^2}{G^2} \bar{V}_1 + 2 \frac{(G -  [\theta_j - k]) (\theta_j - k)}{G^2} \bar{U}_1 - 2 \frac{(\theta_j - k)^2}{G^2} \bar{U}_2  \nonumber \\
& \ - 2 \frac{ ( G - [\theta_j - k] ) \{ \theta_j - k  \} }{G^2} \bar{U}_3   - \frac{(\theta_j - k)^2}{G^2} \bar{V}_2 
\nonumber \\
& \ - 2\frac{ \{ G - (\theta_j - k) \}  \{ \theta_j - k \}}{G^2} \bar{U}_4 
+ 2 \frac{(\theta_j - k)}{G} \bar{U}_5 
\nonumber \\
& \
+ 4 \frac{ \{ G - ( \theta_j - k )  \} \{ \theta_j - k \}}{G^2} \bar{U}_6 - 4 \frac{\{ G - ( \theta_j - k ) \}  \{\theta_j - k \}}{G^2} \bar{U}_7,    
\label{eq:t0:decomp}
\end{align}
where the bar notation represents division by the number of terms in the double summation. Terms with the notation $V$ refer to one-sample V-statistics, whilst $U$ refers to two-sample U-statistics, so that the $\bar{V}_i$ and $\bar{U}_i$ are proper $V$- and $U$-statistics.
In summary, $T_0(G, k) - T_0(G, \theta_j )$ can be represented as a sum of (scaled) one-sample V-statistics and two-sample U-statistics of independent random variables, where each sample is identically distributed. 
For each of these terms, we can apply Hoeffding's inequality for V/U-statistics \citep{hoeffding1994probability}, stated below for completeness:

For a one-sample V-statistic $\bar{V}$ of order 2 with sample size $n$, 
\begin{equation*}
\p \left(\l|  \bar{V} - \mathbb{E}(\bar{V})  \r| > z \right) \leq \exp ( - C_h n z^2 ) ,
\end{equation*}
where $C_h = (b - a)^{-2}$, and $a$ and $b$ are constants such that $a \leq h( \cdot , \cdot) \leq b$, so that for $h_1$, $C_h=1$, and for $h_2$, $C_h = \{1 + 2\exp(-2/3)\}^{-2}$.
For a two-sample U-statistic $\bar{U}$ of order 2 with sample sizes $n_1$ and $n_2$,
\begin{equation*}
\p \left(\l|  \bar{U} - \mathbb{E}(\bar{U})  \r| > z \right) \leq \exp \{ - C_h \min(n_1, n_2) z^2 \} .
\end{equation*}
The tail probabilities derived using the above are determined by the smallest sample size in the U/V-statistics which, in the case of  the $\bar{U}_i$ and $\bar{V}_i$, are greater than $|\theta_j - k|$. As an illustration, for the term $\bar{V}_1$, we have 
\begin{equation*}
\p \left( \l| \bar{V}_1 - \mathbb{E}( \bar{V}_1) \r| > z \right) \leq \exp ( - C_h | \theta_j - k| z^2 ) ,
\end{equation*}
and hence
\begin{equation*}
\p \left( \l| \frac{|\theta_j - k|^2}{G^2} \left\{ \bar{V}_1 - \mathbb{E}( \bar{V}_1) \right\} \r|  > z \right) \leq  \exp ( - C_h G^4 | \theta_j - k |^{-3} z^2 ),
\end{equation*}
and similar arguments apply to $\bar{V}_2$ and $\bar{U}_2$.
As for $\bar{U}_i$ with $i \in \{ 1, 3, 4, 6, 7 \}$, we have that
\begin{equation*}
\p \left( \l| \bar{U}_i - \mathbb{E}( \bar{U}_i) \r| > z \right) \leq \exp ( - C_h|\theta_j - k| z^2 ) ,
\end{equation*}
and hence 
\begin{align*}
\p \left( \l| \frac{( G - \vert \theta_j - k \vert ) \vert \theta_j - k \vert}{G^2} \left\{ \bar{U}_i - \mathbb{E}( \bar{U}_i) \right\} \r|  > z \right)  & \leq  \exp \{ - C_h G^4 | \theta_j - k|^{-1} {( G - | \theta_j - k | )^{-2}} z^2 \} \\ & \leq  \exp ( - C_h G^2 | \theta_j - k |^{-1} z^2 ) 
\end{align*}
Lastly, for $\bar{U}_5$, we have that
\begin{equation*}
\p \left( \l| \bar{U}_5 - \mathbb{E}( \bar{U}_5) \r| > z \right) \leq \exp ( - C_h | \theta_j - k | z^2 ) ,
\end{equation*}
and hence
\begin{equation*}
\p \left( \l| \frac{| \theta_j - k |}{G} \left( \bar{U}_5 - \mathbb{E} \{ \bar{U}_5 \} \right) \r|  > z \right) \leq  \exp ( - C_h G^2 | \theta_j - k |^{-1} z^2 ).
\end{equation*}
Combining the above with~\eqref{eq:t0:decomp}, 
we have that
\begin{align*}
& \p \l( \l| [\{T_0(G, k) - T_0(G, \theta_j ) \} -  \mathbb{E} \{ T_0(G, k) - T_0(G, \theta_j ) \} \r\vert > z \r) 
\\ 
\leq & \, {
3 \exp ( - C_h |\theta_j - k|^{-3} G^4 z^2 ) + 6 \exp ( - C_h | \theta_j - k |^{-1} G^2 z^2 )
}
\\ 
\leq & \, {
9 \exp ( - C_h | \theta_j - k |^{-1} G^2 z^2 ),
}
\end{align*}
which leads to
\begin{align*}
\p \l( \max_{1 \le j \le q} \max_{k: \, \vert k - \theta_j \vert < G} \frac{ \l\vert \{ T_0(G, k) - T_0(G, \theta_j ) \}-  \mathbb{E}\{ T_0(G, k) - T_0(G, \theta_j )\} \r \vert }{\sqrt{ \vert \theta_j - k \vert }} > z \r) 
\\
\le 9 n G \exp( - C_h G^2 z^2).
\end{align*}
This, together with~\eqref{eq:t0:bias}, implies that $\p (\mc M_{n} ) \ge 1  - n^{-1}$ where
\begin{align*}
\mc M_{n} := \l\{  \max_{1 \leq j \leq q} \max_{k: \, | k - \theta_j | < G}  \frac{ \l| \{ T_0(G, k) - T_0(G, \theta_j ) \} - f(\theta_j, G, k) d^{(j)}_0 \r\vert}{\sqrt{ {\vert \theta_j - k \vert }}}  \leq  \frac{c_0 \sqrt{\log(n)}}{G} \r\},
\end{align*}
for some large enough constant $c_0 > 0$.
Conditional on $\mc M_{n}$, we have that $T_0(G, \wh\theta) \geq T_0(G, \theta_j)$ for some $\wh\theta$ satisfying $\vert \wh\theta - \theta_j \vert < G$, implies that
\begin{align*}
0 \leq \frac{T_0(G, \wh\theta) - T_0(G, \theta_j )}{\sqrt{ \vert \wh\theta - \theta_j \vert}} & \leq \frac{f(\theta_j, G, k )}{\sqrt{ \vert \wh\theta - \theta_j \vert}} d^{(j)}_0 + \frac{c_0  \sqrt{\log(n)}}{G},
\end{align*}
so that 
\begin{align*}
\frac{\vert \wh\theta - \theta_j \vert (2G - \vert \wh\theta - \theta_j \vert )}{G^2 \sqrt{\vert \wh\theta - \theta_j \vert}} d^{(j)}_0  & \leq \frac{c_0  \sqrt{\log(n)}}{G}.
\end{align*}
Hence
\begin{align*}
\vert \wh\theta - \theta_j \vert (2G - \vert \wh\theta - \theta_j \vert ) d^{(j)}_0  & \leq   c_0 G \sqrt{ \vert \wh\theta - \theta_j \vert } \sqrt{\log(n)} \leq c_0 (2G - \vert \wh\theta - \theta_j \vert )  \sqrt{ \vert \wh\theta - \theta_j \vert } \sqrt{\log(n)},
\end{align*}
so that
\begin{align*}
(d^{(j)}_0)^2 \vert \wh\theta - \theta_j \vert  &\leq   c_0^2 \log (n),
\end{align*}
which completes the proof.

\subsection{Proof of Theorem~\ref{thm:multilag-consistent}}

Recall the definition of $\mc E_{\ell, n}$ in~\eqref{eq:Op-bound}.
In what follows, we condition our arguments on the event $\mc E_n = \cap_{\ell \in \mc L} \mc E_{\ell, n}$ which satisfies $\p(\mc E_n) \to 1$ as $n \to \infty$ for any fixed $\mc L$.
That is, in what follows, all big-O and small-o terms can uniformly be replaced by $O_P$ and $o_P$.
Throughout, we assume that there is a unique maximiser of $d^{(j)}_\ell$ with respect to $\ell \in \mc L^{(j)}$ for all $j = 1, \ldots, q$.
In the case of ties, we arbitrarily break them which does not alter the conclusion.

\begin{proof}[Proof of~\ref{thm:multilag1}]
By Step~2 in the proof of Theorem~\ref{thm:consistency}, we have for all $\ell \in \mc L$ and large enough $n$:
\begin{enumerate}[label = (\alph*)]
\item \label{cond:theta:one} For all $\wt\theta \in \wh\Theta_\ell$, there exists a unique index $j \in \mc I_\ell$ such that $\vert \wt\theta - \theta_j \vert \le \eta G$, i.e.\ $\wt\theta$ is an estimator of $\theta_j$ in view of Assumption~\ref{assum:change-point}~\ref{assum:min-space}.
\item \label{cond:theta:two} Conversely, for all $j \in \mc I_\ell$, there exists a unique element $\wt\theta \in \wh\Theta_\ell$ estimating $\theta_j$ such that $\vert \wt\theta - \theta_j \vert \le \eta G$.
\end{enumerate}
Then by Assumption~\ref{assum:lag-assumption} and~\ref{cond:theta:one}, in the first iteration of multi-lag NP-MOJO, we identify $\wt{\theta}_1$ which detects $\theta_1$ and satisfies $\vert \wt{\theta}_1 - \theta_1 \vert \le \eta G$. 
The set $\mc C_1$ contains the estimators of $\theta_1$ only.
To see this, for all $\wt\theta \in \mc C_1$ and $j > 1$,
\begin{align}
\vert \wt{\theta} - \theta_j \vert \ge \vert \theta_2 - \theta_1 \vert - \vert \wt{\theta}_1 - \theta_1 \vert  - \vert \wt{\theta} - \wt\theta_1 \vert > (2 - c - \eta) G \ge \eta G \nonumber
\end{align}
such that by~\ref{cond:theta:one}, $\wt\theta$ cannot be an estimator of $\theta_j, \, j > 1$.
Besides, any estimator of $\theta_1$ is contained in $\mc C_1$.
To see this, if $\wt\theta \notin \mc C_1$,
\begin{align}
\vert \wt{\theta} - \theta_1 \vert \ge \vert \wt{\theta} - \wt{\theta}_1 \vert - \vert \wt{\theta}_1 - \theta_1 \vert > (c - \eta) G \ge \eta G, \nonumber
\end{align}
i.e.\ such $\wt\theta$ is not an estimator of $\theta_1$ by~\ref{cond:theta:one}.
From these and~\ref{cond:theta:two}, for any $\wt\theta \in \mc C_1 \cap \wh{\Theta}_\ell$ for some lag $\ell \in \mc L^{(1)}$, we have $d^{(1)}_{\ell}  \vert \wt\theta - \theta_1 \vert \le c_0 \sqrt{G \log(n)}$ by Theorem~\ref{thm:consistency} conditional on $\mc E_n$. Then,
\begin{align}
\frac{T_{\ell}(G, \wt\theta)}{ d^{(1)}_{\ell} }  
&= \l( \frac{G - \ell - \vert \wt\theta - \theta_1 \vert }{ G - \ell } \r)^2 + O \l( \frac{\sqrt{\log (n)}}{\sqrt{G} d^{(1)}_{\ell}} \r) 
\nonumber \\
&\ge \l( 1 - \frac{c_0 \sqrt{G \log(n)}}{d^{(1)}_{\ell} (G - \ell) } \r)^2 + o(1) = 1 + o(1),
\label{eq:td}
\end{align}
where $o(1)$ terms are due to Assumption~\ref{assum:change-point}~\ref{assum:change-size}.
Therefore, for any distinct $\wt\theta, \wt\theta^\prime \in \mc C_1$ associated with lags $\ell, \ell^\prime \in \mc L^{(1)}$, respectively, we have
\begin{align*}
\frac{T_{\ell}(G, \wt\theta)}{T_{\ell^\prime}(G, \wt\theta^\prime)} = \frac{ d^{(1)}_{\ell} (1 + o(1)) } { d^{(1)}_{\ell^\prime} (1 + o(1)) },
\end{align*}
which implies that for $n$ large enough, Step~3 of multi-lag NP-MOJO identifies $\wh\theta_1 \in \mc C^{(1)} \cap \wh{\Theta}_{\ell^{(1)}}$ with $\ell^{(1)} = \argmax_{\ell \in \mc L^{(1)}} d^{(1)}_\ell$.
This, combined with Theorem~\ref{thm:consistency}, establishes that
\begin{align*}
\max_{\ell \in \mc L^{(1)}} d^{(j)}_\ell \vert \wh\theta_1 - \theta_1 \vert \le c_0 \sqrt{G \log(n)}.
\end{align*}
Step~4 of multi-lag NP-MOJO removes all estimators of $\theta_1$ from further consideration {\cblue and obtains $\wt\Theta_1$, such that $\wt\theta_2 = \min \widetilde \Theta_{2}$ is an estimator of $\theta_2$}.
Then iteratively applying the above arguments, under Assumption~\ref{assum:lag-assumption} and~\ref{cond:theta:one}--\ref{cond:theta:two}, we obtain $\wh\Theta$ satisfying the claim of the theorem.
\end{proof}

\begin{proof}[Proof of~\ref{thm:multilag2}] 
The proof proceeds analogously as the proof of~\ref{thm:multilag1}, with the following modifications of~\ref{cond:theta:one}--\ref{cond:theta:two}:
\begin{enumerate}[label = (\alph*$^\prime$)]
\item \label{cond:theta:one:n} For all $\wt\theta \in \wh\Theta_\ell, \, \ell \in \mc L$, there exists a unique index $j \in \mc I_\ell$ such that $\vert \wt\theta - \theta_j \vert < G$, i.e.\ $\wt\theta$ is an estimator of $\theta_j$ in view of Assumption~\ref{assum:change-point2}~\ref{assum:min-space2}.
\item \label{cond:theta:two:n} For all $j = 1, \ldots, q$, there exists {\cblue at least one element $\wt \theta \in \wt \Theta$ }estimating $\theta_j$ such that $\vert \wt\theta - \theta_j \vert \le \eta G$.
Among such $\wt\theta$, one is detected at lag $\ell^{(j)} = \max_{\ell \in \mc L^{(j)}} d^{(j)}_\ell$.
\end{enumerate}
Then by~\ref{cond:theta:one:n}, in the first iteration of multi-lag NP-MOJO, we identify $\wt\theta_1$ which detects $\theta_1$ and satisfies { \cblue $\vert \wt\theta_1 - \theta_1 \vert < G$. }
The set $\mc C_1$ contains the estimators of $\theta_1$ only, since for all $\wt\theta \in \mc C_1$ and $j > 1$,
\begin{align*}
\vert \wt\theta - \theta_j \vert \ge \vert \theta_2 - \theta_1 \vert - \vert \wt\theta_1 - \theta_1 \vert - \vert \wt\theta - \wt\theta_1 \vert > (3 - c)G \ge G,
\end{align*}
such that by~\ref{cond:theta:one:n}, $\wt\theta$ cannot be an estimator of $\theta_j, \, j > 1$.
Besides, any estimator of $\theta_1$ is contained in $\mc C_1$.
To see this, if $\wt\theta \notin \mc C_1$,
\begin{align*}
\vert \wt\theta - \theta_1 \vert \ge \vert \wt\theta - \wt\theta_1 \vert - \vert \wt\theta_1 - \theta_1 \vert > (c - 1) G \ge G.
\end{align*}
From~\ref{cond:theta:two:n}, there exists $\wt\theta \in \mc C_1$ detected at lag $\ell^{(1)}$ such that analogously as in~\eqref{eq:td}, we have
\begin{align*}
(d^{(1)}_{\ell^{(1)}})^{-1} T_{\ell^{(1)}}(G, \wt\theta) = 1 + o(1)
\end{align*}
conditional on $\mc E_n$. At other $\wt\theta^\prime \in \mc C_1 \setminus \{ \wt\theta \}$ detected at some $\ell \in \mc L^{(1)} \setminus \{\ell^{(1)}\}$, we have
\begin{align*}
\frac{T_{\ell}(G, \wt\theta^\prime)}{ d^{(1)}_{\ell} }  
&\le 1 + O \l( \frac{\sqrt{\log (n)}}{\sqrt{G} d^{(1)}_{\ell}} \r) = 1 + o(1), \text{ \ such that}
\\
\frac{T_{\ell}(G, \wt\theta^\prime)}{ T_{\ell^{(1)}}(G, \wt\theta) } &\le \frac{d^{(1)}_{\ell}}{d^{(1)}_{\ell^{(j)}}} (1 + o(1)) < 1
\end{align*}
for large enough $n$. 
This implies that Step~3 of multi-lag NP-MOJO identifies $\wt\theta \in \mc C^{(1)} \cap \wh{\Theta}_{\ell^{(1)}}$ as $\wh\theta_1$.
The rest of the proof is analogous to the proof of~\ref{thm:multilag1} and is omitted. 
\end{proof}

\section{Supporting lemmas}
\label{sec:aux-lemma}

\subsection{For Lemma~\ref{lemma:weight-int-identities}}

The proof of Lemma~\ref{lemma:weight-int-identities} requires the following lemmas for the weight functions $w_1$, $w_2$ and $w_3$.  Lemmas~\ref{integral-lemma} and~\ref{integral-lemma2} are stated without proof in \cite{fan2017multivariate}, whilst Lemma~\ref{integral-lemma3} is stated without proof in \cite{bakirov2006multivariate}. To the best of our knowledge, there is no proof of these results in the related literature, so we provide proofs here for completeness.
\begin{lemma}\label{integral-lemma}
For $\beta>0$ and any $x \in \mathbb{R}^p$,
\begin{align}
I_1(\beta,x) = \int_{\mathbb{R}^p} \cos (\langle t,x \rangle) \exp \l(-\frac{1}{2\beta^2} \Vert t \Vert^2 \r) dt = C_1(\beta,p) \exp \l(- \frac{\beta^2}{2} \Vert x \Vert^2 \r),
 \end{align}
where $C_1(\beta,p) =  (2 \pi)^{p/2} \beta^{p} $.
\end{lemma}

\begin{proof}
First, consider the case where $p=1$, from which the general case will follow. Recognising $I_1 (\beta ,x)$ as the (scaled) characteristic function of $\mc N(0, \beta^2)$, we have 
\begin{align*}
I_1(\beta,x) = \sqrt{2\pi} \beta \exp \l( -\frac{\beta^2 x^2}{2} \r).
\end{align*}
For the general case, note that $I(\beta,x)$ is invariant under orthogonal transformations $H$ of $x$, so that
\begin{align*}
I_1 (\beta, x)  =  \int_{\mathbb{R}^p} \cos \l( \langle Ht,Hx \rangle \r) \exp \l(-\frac{1}{2\beta^2} \Vert Ht \Vert^2 \r) dt =  \int_{\mathbb{R}^p} \cos (\langle t, Hx \rangle) \exp \l(-\frac{1}{2\beta^2} \Vert t \Vert^2 \r) dt,
\end{align*}
which follows since the inner product and Euclidean norm are invariant under orthogonal transformations, and the transformation $t \mapsto Ht$ leaves the Lebesque measure $dt$ unchanged. Therefore, to evaluate $I(\beta,x)$, we can replace $x$ with $\Vert x \Vert (1, 0, \ldots, 0)$. Letting $t = (t_1, \ldots t_p)$, we obtain
\begin{align*}
I_1(\beta,x) &= \int_{\mathbb{R}^p} \cos ( t_1 \Vert x \Vert ) \exp \l(-\frac{1}{2\beta^2} \Vert t \Vert^2 \r) dt \\
&= \int_{\mathbb{R}} \cos ( t_1 \Vert x \Vert ) \exp \l(-\frac{1}{2\beta^2} t_1^2 \r) dt_1 \int_{\mathbb{R}^{p-1}} \exp \l\{ -\frac{1}{2\beta^2} \l( t_2^2 + \cdots + t_p^2 \r) \r\} dt_2 \cdots dt_p \\
&= \sqrt{2 \pi} \beta \exp \l( - \frac{\beta^2}{2} \Vert x \Vert^2 \r) \prod_{\ell=2}^{p} \int_{\mathbb{R}} \exp \l(-\frac{ t_\ell^2}{2\beta^2}  \r) dt_\ell \\
&= (2 \pi)^{p/2} \beta^p \exp \l( - \frac{\beta^2}{2} \Vert x \Vert^2 \r).
\end{align*}
\end{proof}

\begin{lemma}\label{integral-lemma2}
For $\delta>0$ and any $x \in \mathbb{R}^p$,
\begin{align}
I_2(\delta,x) = \int_{\mathbb{R}^p} \cos (\langle t,x \rangle) \prod_{j=1}^p t_j^2 \exp(-\delta t_j^2) dt = C_2 (\delta , x) \prod_{j=1}^p \frac{2 \delta - x_j^2}{2 \delta} \exp \l( -\frac{1}{4\delta} x_j^2 \r) ,
 \end{align}
where $C_2(\delta,p) =  2^{-1} \pi^{p/2} \delta^{-3p/2}$.
\end{lemma}

\begin{proof}
 We proceed by induction on the dimension $p$. First, consider the case where $p=1$. Then, repeatedly using integration by parts and Lemma~\ref{integral-lemma}, we obtain
 \begin{align*}
 I_2 (\delta , x) &= \int_{\mathbb{R}} \cos (tx) t^2  \exp (-\delta t^2) dt   = \int_{\mathbb{R}}  \left\{ t \cos (tx) \right\} \times \left\{ t  \exp (-\delta t^2) \right\} dt   \\
 & = \frac{1}{2 \delta} \int_{\mathbb{R}} \l\{ \cos (tx) - tx \sin (tx) \r\} \exp \l( -\delta t^2 \r) dt \\ 
 &= 2^{-1} \sqrt{\pi} \delta^{-3/2} \exp \l( - \frac{1}{4 \delta} x^2  \r) -\frac{x^2}{4 \delta^2} \int_{\mathbb{R}} \cos (tx) \exp(-\delta t^2) dt \\
 & = 2^{-1} \sqrt{\pi} \delta^{-3/2} \exp \l( - \frac{1}{4 \delta} x^2  \r) -\frac{\sqrt{\pi}x^2}{4 \delta^{5/2}} \exp \l( - \frac{1}{4 \delta} x^2  \r) \\
 & = 2^{-1} \sqrt{\pi} \delta^{-3/2} \exp \l( - \frac{1}{4\delta} x^2  \r) \l( \frac{2 \delta - x^2}{2 \delta} \r) .
 \end{align*}

For general dimension $p$, assume the result is true for dimension $p-1$ and proceed via induction. Using the cosine summation formula, we have
\begin{align*}
  & \int_{\mathbb{R}^p} \cos (\langle t,x \rangle) \prod_{j=1}^p t_j^2 \exp(-\delta t_j^2) dt  \\
  =& \int_{\mathbb{R}^p} \left\{ \cos \l( \sum_{j=1}^{p-1} t_j x_j \r) \cos (t_p x_p)  -   \sin \l( \sum_{j=1}^{p-1} t_j x_j \r)   \sin (t_p x_p)  \right\} \prod_{j=1}^p t_j^2 \exp(-\delta  t_j^2) dt \\
  =& \int_{\mathbb{R}^{p-1}} \cos \l( \sum_{j=1}^{p-1} t_j x_j \r)  \prod_{j=1}^{p-1} t_j^2 \exp(-\delta  t_j^2) dt_j \int_{\mathbb{R}} \cos (t_p x_p) t_p^2  \exp(-\delta t_p^2 ) dt_{p} \\
  & \ \ - \int_{\mathbb{R}^{p-1}} \sin \l( \sum_{j=1}^{p-1} t_j x_j \r)  \prod_{j=1}^{p-1} t_j^2 \exp(-\delta  t_j^2) dt_j \int_{\mathbb{R}} \sin (t_p x_p) t_p^2  \exp(-\delta t_p^2 ) dt_{p} \\
  =& \int_{\mathbb{R}^{p-1}} \cos \l( \sum_{j=1}^{p-1} t_j x_j \r)  \prod_{j=1}^{p-1} t_j^2 \exp(-\delta  t_j^2) dt_j \int_{\mathbb{R}} \cos (t_p x_p) t_p^2  \exp(-\delta t_p^2 ) dt_{p} \\
  =& \ 2^{-p} \pi^{p/2} \delta^{-3p/2} \prod_{j=1}^p \frac{2 \delta - x_j^2}{2 \delta} \exp \l( -\frac{1}{4\delta} x_j^2 \r) ,
\end{align*}
where the third line follows since the one-dimensional integral is the integral of an odd function and integrates to $0$, and the fourth line follows from the inductive assumption and the proof in the case $p=1$. Hence the result follows by induction.
\end{proof}

\begin{lemma}\label{integral-lemma3}
For $\gamma \in (0,2)$ and any $x,y \in \mathbb{R}^{p}$,
\begin{align}\label{non-sep-weight-identity}
& I_3 (\gamma,x,y) = \int_{\mathbb{R}^{p}} \int_{\mathbb{R}^{p}} \frac{1 - \cos (\langle u,x \rangle) \cos (\langle v,y \rangle)}{(\Vert u \Vert^{2} +\Vert v \Vert^{2})^{(\gamma+2p)/2}} dudv = C_3 (\gamma , 2p) \l( \Vert x \Vert^2 + \Vert y \Vert^2  \r)^{\gamma/2},
\\
& \text{where \ } C_3(\gamma, p) = \frac{2 \pi^{p/2} \Gamma (1 - \gamma/2 ) }{\gamma 2^{\gamma} \Gamma((p+\gamma)/2 )}.    
\nonumber
\end{align}
\end{lemma}

  \begin{proof}
We first prove the case where $p=1$, from which the general case will follow. First, use the trigonometric identity $2\cos(a)\cos(b) = \cos(a+b) + \cos(a-b)$, to obtain
 \begin{align*}
 I_3 (\gamma, x,y) = \frac{1}{2} \int_{\mathbb{R}^{2}} \frac{1 - \cos (ux+vy) }{(u^{2} +v^{2})^{(\gamma+2)/2}} dudv + \frac{1}{2} \int_{\mathbb{R}^{2}} \frac{1 - \cos (ux-vy) }{(u^{2} +v^{2})^{(\gamma+2)/2}} dudv  = : A + B.
 \end{align*}
 Considering the first term, make the transformation to polar coordinates $u = r \cos \theta$, $v= r \sin \theta$, to obtain
 \begin{align*}
A = \frac{1}{2} \int_0^{2 \pi} \int_0^{\infty} \frac{1- \cos \l\{ \l(x\cos \theta + y \sin \theta \r) r  \r \} }{|r|^{\gamma+1}} dr d\theta .
 \end{align*}
 Next, note that by Identity 3.032.2 of \cite{gradshteyn2014table} and since the integrand is an even function with respect to $r$,
  \begin{align*}
A =  \frac{1}{2} \int_0^{\pi} \int_\mathbb{R} \frac{1- \cos \l( \l( \sqrt{x^2 + y^2} \cos \theta \r) r  \r )}{|r|^{\gamma+1}} dr d\theta .
 \end{align*}
 This term can be integrated with respect to $r$ using Lemma 1 from \cite{szekely2007}, to yield
  \begin{align*}
A = \frac{1}{2} C_2( \gamma,1) (x^2 + y ^2)^{\gamma/2} \int_0^{\pi} \vert \cos \theta \vert^{\gamma}  d\theta .
 \end{align*}
 Next, by elementary trigonometric integral identities, we obtain
  \begin{align*}
A = C_2(\gamma, 1) \l( x^2 + y^2 \r)^{\gamma /2} \int_0^{\pi /2} \l\vert \cos \theta \r\vert^\gamma d\theta =  C_2(\gamma, 1) \l( x^2 + y^2 \r)^{\gamma /2} \frac{\sqrt{\pi} \Gamma ( (\gamma + 1)/2 )}{2 \Gamma (\gamma/2 + 1 )}.
 \end{align*}
A similar calculation shows that
   \begin{align*}
B = C_2( \gamma,1) \l( x^2 + y^2 \r)^{\gamma /2} \frac{\sqrt{\pi} \Gamma ( (\gamma + 1)/2 )}{2 \Gamma (\gamma/2 + 1 )}.
 \end{align*}
Minor simplifications yield the required
 \begin{align*}
I_3(\gamma,x,y) = C_2 (\gamma,2) \l( x^2 + y^2 \r)^{\gamma /2} .
 \end{align*}
 For the general case, note that $I(x,\gamma)$ is invariant under orthogonal transformations $H$ of $x$ and $y$, so that
\begin{align*}
\int_{\mathbb{R}^{p}} \int_{\mathbb{R}^{p}} \frac{1 - \cos ( \langle u, x \rangle ) \cos ( \langle v, y \rangle )}{(\Vert u \Vert^{2} + \Vert v \Vert^{2})^{(\gamma+2p)/2}} dudv = \int_{\mathbb{R}^{p}} \int_{\mathbb{R}^{p}} \frac{1 - \cos ( \langle u, Hx \rangle ) \cos ( \langle v, Hy \rangle )}{(\Vert u \Vert^{2} + \Vert v \Vert^{2})^{(\gamma+2p)/2}} dudv ,
\end{align*}
which follows since the inner product and Euclidean norm are invariant under orthogonal transformations. Therefore, to evaluate $I(\gamma,x,y)$ we can replace $x$ with $\Vert x \Vert (0, 0, \ldots, 1)$, and $y$ with $|y|_2 (0, 0, \ldots, 1)$. Hence, letting $u = (u_1, \ldots, u_p)$ and $v = (v_1, \ldots , v_p)$, 
\begin{align*}
 I_3(\gamma,x,y) =  \int_{\mathbb{R}^{p}} \int_{\mathbb{R}^{p}}\frac{1 - \cos (u_p \Vert x \Vert) \cos (v_p \Vert y \Vert )}{( \Vert u \Vert^{2} + \Vert v \Vert^{2})^{(\gamma+2p)/2}} dudv .
\end{align*}
As in the $p=1$ case, we split the integral into two parts to obtain 
\begin{align*}
I_3 (\gamma, x,y) = \frac{1}{2} \int_{\mathbb{R}^{p}} \int_{\mathbb{R}^{p}} \frac{1 - \cos ( u_p \Vert x \Vert + v_p  \Vert y \Vert) }{(\Vert u \Vert^{2} + \Vert v \Vert^{2})^{(\gamma+2p)/2}} dudv 
\\
+ \frac{1}{2} \int_{\mathbb{R}^{p}} \int_{\mathbb{R}^{p}} \frac{1 - \cos (u_p \Vert x \Vert - v_p \Vert y \Vert) }{( \Vert u \Vert^{2} + \Vert v \Vert^{2})^{(\gamma+2p)/2}} dudv  = A + B.
 \end{align*}
Now, focusing on the the first term, make the transformation to $2p$-dimensional spherical coordinates, so that
\begin{align*}
u_1 &= r \cos \theta_1, v_1 = r \sin \theta_1 \cos \theta_2, \\
u_2 &= r \sin  \theta_1 \sin  \theta_2 \cos \theta_3  , v_2 = r \sin  \theta_1 \sin \theta_2 \sin \theta_3 \cos \theta_4,\\
& \vdots \\
u_p &= r \l( \prod_{j=1}^{2p-2} \sin \theta_j \r) \cos \theta_{2p-1}, v_p = r \l( \prod_{j=1}^{2p-2} \sin \theta_j \r) \sin \theta_{2p-1} ,
\end{align*}
where the $(\theta_1, \ldots , \theta_{2p-2})$ range over $D= [0, \pi]^{2p-2}$, and $\theta_{2p-1}$ ranges over $[0, 2\pi)$. The Jacobian of this transformation is given by
\begin{align*}
|J| = r^{2p-1} dr  \prod_{j=1}^{2p-2} \sin^{2p-1-j} ( \theta_{j} ) d\theta_j.
\end{align*}
Hence,
\begin{align*}
 A &=   \frac{1}{2} \int_D \int_0^{2 \pi} \int_{0}^{\infty} \frac{1 - \cos \l( r (\Vert x \Vert \cos \theta_{2p-1} + \Vert y \Vert \sin \theta_{2p-1}) \prod_{j=1}^{2p-2} \sin \theta_j  \r)  }{r^{\gamma+1}} dr d\theta_{2p-1} 
 \\
 & \qquad \qquad \qquad \qquad \times \prod_{j=1}^{2p-2} \sin^{2p-1-j} ( \theta_{j} ) d\theta_j
 \\
 &= \frac{1}{2} C_2( \gamma, 1) ( \Vert x \Vert^2 + \Vert y \Vert^2 )^{\gamma/2} \int_0^\pi \vert \cos \theta_{2p-1} \vert^\gamma d\theta_{2p-1}  \prod_{\ell=1}^{2p-2} \int_0^{\pi} \sin^{2p+\gamma-1-\ell} \mc I_\ell d \theta_{\ell} \\[1ex]
 &= \frac{1}{2} C_2(\gamma, 1) ( \Vert x \Vert^2 + \Vert y \Vert^2 )^{\gamma/2} \frac{\sqrt{\pi} \Gamma ((\gamma+1)/2)}{\Gamma (\gamma/2 +1)}  \prod_{\ell=1}^{2p-2}  \frac{\sqrt{\pi} \Gamma ((2p+\gamma - \ell)/2 )}{\Gamma ( (2p+\gamma -1 - \ell )/2 + 1 )} \\[1ex]
 &= ( \Vert x \Vert^2 + \Vert y \Vert^2 )^{\gamma/2} \frac{{\pi} \Gamma (1 - \gamma/2)}{\Gamma (1 + \gamma/2)}  \times \frac{ \pi^{p-1} \Gamma (1 + \gamma/2 )}{ \Gamma ( (2p+\gamma)/2)} = \frac{\pi^p \Gamma (1 - \gamma/2)}{\gamma 2^\gamma \Gamma ( (2p + \gamma)/2) } ( \Vert x \Vert^2 + \Vert y \Vert^2 )^{\gamma/2}.
\end{align*}
The second line follows using Lemma 1 from \cite{szekely2007}, and the third follows using the fact that $\sin \theta_i$ is nonnegative on $[0,\pi]$, and the same trigonometric identities used in the $p=1$ case. The final line follows from cancellation in the numerator and denominator of consecutive product terms. The term $B$ can be dealt with in an identical manner to yield the same expression. Therefore, we have that
\begin{align*}
I_3 (\gamma, x , y ) =     \frac{2 \pi^p \Gamma (1 - \gamma/2)}{\gamma 2^\gamma \Gamma ( (2p + \gamma)/2) } ( \Vert x \Vert^2 + \Vert y \Vert^2 )^{\gamma/2} = C_2( \gamma, 2p) ( \Vert x \Vert^2 + \Vert y \Vert^2 )^{\gamma/2} .
\end{align*}
\end{proof}

\subsection{For Theorem~\ref{thm:consistency}}

\begin{lemma}\label{uniform-expectation}
Suppose that Assumptions~\ref{assum:beta-mixing},~\ref{assum:kernel} and~\ref{assum:change-point}~\ref{assum:min-space} hold. Then, for $\mc D_\ell (G,k)$ defined in Equation~\eqref{eq:pop-test-stat}, we have that
\begin{align*}
\max_{G \leq k \leq n- G} \l\vert \mathbb{E} \{ \tstatkl \}   -   \mc D_\ell (G,k)\r\vert = O(G^{-1/2}).
\end{align*}    
\end{lemma}

\begin{proof}
Firstly, from Assumption~\ref{assum:beta-mixing}, since $\{X_t \}_{t \in \mathbb{Z}}$ is $\beta$-mixing with algebraic decaying mixing coefficients, then so too is $\{ Y_t \}_{t \in \mathbb{Z}}$ with the same decay rate, since $\ell$ is fixed. Next, we have
 \begin{align*}
\tstatkl & = \frac{1}{(G-\ell)^2} \l\{ \sum_{s,t = k - G+1}^{k - \ell} h(Y_s, Y_t) + \sum_{s,t = k+1}^{k +G - \ell} h(Y_s, Y_t) -2 \sum_{s = k - G+1}^{k - \ell} \sum_{t = k +1}^{k +G - \ell} h(Y_s, Y_t)  \r\} \\
&=: T^{(1)}_{\ell} (G,k) + T^{(1)}_{\ell}  (G,k+G) - 2 T^{(2)}_{\ell}  (G,k) .
\end{align*}
Define the following notations, which are the nonstationary analogues of the standard quantities used in a Hoeffding decomposition of a $V$-statistic of order $2$. When dealing with a vector, let $\leq$ denote the inequality which is satisfied for all coordinates of the vector. Let
\begin{align*}
h_{1,G}  (y_1) = \sum_{t = k-G+1}^{k-\ell}   \int_{\mathbb{R}^{2p}} h(y_1 , y_2) dF_{t} (y_2), 
\end{align*}
where $F_t$ denotes the cdf of $Y_{t}$. Next, let
\begin{align*}
V^{(1)}_{\ell}  (G,k) &=  \frac{1}{(G-\ell)^2} \sum_{s = k-G+1}^{k-\ell}  \int_{\mathbb{R}^{2p}} h_{1,G}  (y_1)  \times d \l( \mathbb{I}{\{  Y_s \leq y_1 \}} - F_{s} (y_{1} ) \r),
\\
V^{(2)}_{\ell}  (G,k) &=  \frac{1}{(G-\ell)^2} \sum_{s,t = k-G+1}^{k-\ell} \int_{\mathbb{R}^{2p}} 
 \int_{\mathbb{R}^{2p}} h(y_1, y_2)  \times d \l( \mathbb{I}{\{ Y_s \leq y_1 \}} - F_{s} (y_{1} ) \r) d \l( \mathbb{I}{\{ Y_t \leq y_2 \}} - F_{t} (y_{2} ) \r).
\end{align*}
Further, define
\begin{align*}
\mu^{(1)}_{\ell}  (G, k) &= \frac{1}{(G-\ell)^2} \sum_{s,t = k-G+1}^{k-\ell} \int_{\mathbb{R}^{2p}} \int_{\mathbb{R}^{2p}} h(y_1 , y_2) dF_{s} (y_1) dF_t (y_2) \\
&= \frac{1}{(G-\ell)^2} \sum_{s,t = k-G+1}^{k-\ell} \mathbb{E} ( h(Y_s, \wt Y_t) ),
\\
\mu^{(2)}_{\ell} (G, k) &= \frac{1}{(G-\ell)^2} \sum_{s=k-G+1}^{k-\ell} \sum_{t = k+1}^{k+G-\ell} \int_{\mathbb{R}^{2p}} \int_{\mathbb{R}^{2p}} h(y_1 , y_2) dF_{s} (y_1) dF_t (y_2) \\
&=  \frac{1}{(G-\ell)^2 } \sum_{s=k-G+1}^{k-\ell} \sum_{t = k+1}^{k+G-\ell}  \mathbb{E} ( h(Y_s, \wt Y_t) ),
\end{align*}
where $\wt Y_t$ denotes an independent copy of $Y_t$. Then, following e.g. \cite{harel1989limiting}, by the Hoeffding decomposition, we have that
\begin{align*}
T^{(1)}_{\ell} (G,k) = \mu^{(1)}_{\ell}  (G, k) + 2 V^{(1)}_{\ell}  (G,k) + V^{(2)}_{\ell}  (G,k) .
\end{align*}
By Lemma 2.2 of \cite{harel1989limiting}, we have that $ \mathbb{E} \{ V^{(r)}_\ell (G,k)^2\} = O(G^{-1})$ for $r = 1, 2$, from which it follows that 
\begin{align*}
\l\vert \mathbb{E} \left\{ T^{(1)}_{\ell}  (G,k) \right\} - \mu^{(1)}_{\ell}  (G, k ) \r\vert = O(G^{-1/2}). 
\end{align*}
An identical argument applies to the term $T^{(1)}_{\ell}  (G,k+G)$. For $T^{(2)}_{\ell}  (G,k)$, the Hoeffding decomposition takes the following form:
\begin{align*}
 T^{(2)}_{\ell}  (G,k) &= \mu^{(2)}_{\ell}  (G,k) + 2 U^{(1)}_{\ell} (G,k)  + U^{(2)}_{\ell}  (G,k)  , \text{ \ where}
 \\
U^{(1)}_{\ell}  (G,k) &=  \frac{1}{(G-\ell)^2} \sum_{s = k+1}^{k+G-\ell}  \int_{\mathbb{R}^{2p}} h_{1,G}  (y_1)  \times d \l( \mathbb{I}{\{ Y_s \leq y_1 \}}  - F_{s} (y_{1} ) \r), \\
  U^{(2)}_{\ell}  (G,k)  &= \frac{1}{(G-\ell)^2}\sum_{s=k-G+1}^{k-\ell} \sum_{t = k+1}^{k+G-\ell} \int_{\mathbb{R}^{2p}} 
 \int_{\mathbb{R}^{2p}} h(y_1, y_2)  \\
 & \qquad \qquad \qquad \qquad
 \times d \l( \mathbb{I}{\{ Y_s \leq y_1\}} - F_{s} (y_{1} ) \r) d \l( \mathbb{I}{\{ Y_t \leq y_2 \}} - F_{t} (y_{2} ) \r),
\end{align*}
with $\mathbb{E}\{U^{(1)}_{\ell} (G,k) \}  =0$. Analogously, we have from Equation (2.14) of \cite{harel1989limiting} that $\mathbb{E} \{ U^{(2)}_{\ell} (G,k)^2\} = O(G^{-1})$, so that 
\begin{align*}
\l\vert \mathbb{E} \left\{ T^{(2)}_{\ell}  ( G,k ) \right\} - \mu^{(2)}_{\ell}  (G, k) \r\vert = O(G^{-1/2}).     
\end{align*}
Therefore, we have that
\begin{align} \label{eq:bias}
& \l\vert \mathbb{E} \{ T_{\ell}  (G,k) \} - \mu_{\ell}  (G, k) \r\vert = O(G^{-1/2}), \text{ \ where}
\\
& \mu _{\ell} (G, k) = \mu^{(1)}_{\ell} (G, k) + \mu^{(1)}_{\ell} (G, k+G) -2 \mu^{(2)}_{\ell} (G, k) .  
\nonumber
\end{align}

Next, we work out the form of $ \mu_{\ell}  (G, k)$, in order to express it in terms of the population quantity $d^{(j)}_\ell$. Firstly, consider the case where $k \in \{ \theta_{j} +G , \ldots , \theta_{j+1} - G \}$ for some $j = 0 , \ldots , q$. Then, by Assumption~\ref{assum:change-point}~\ref{assum:min-space}, we have that
\begin{align*}
\mu^{(1)}_{\ell}  (G,k)  = \mu^{(1)}_{\ell}  (G,k+G) = \mu^{(2)}_{\ell} (G,k) = \mathbb{E} \l\{ h \l(Y^{(j)}_1 , \wt Y_1^{(j)} \r) \r\},  \end{align*}
so that $\mu_{\ell} (G,k) = 0$. Now, consider $k \in \{ \theta_j - G + \ell + 1, \ldots , \theta_j  \}.$
Denote $\wt h_{s,t} = \mathbb{E} \{ h (Y_s , \wt Y_t ) \}$, and let 
\begin{align*}
 & \mc A_\ell^{(1)} (G,k) = \{ (s,t) : k +1 \leq s,t \leq k+G-\ell, \quad s \text{ or } t \in \{\theta_j -\ell + 1 , \ldots , \theta_j \} \}, \\[1ex]
  & \mc A_\ell^{(2)} (G,k) = \{ (s,t) : k -G +1 \leq s \leq k-\ell, \quad \theta_j -\ell + 1 \leq t \leq \theta_j \} .
 \end{align*}
Then, $\mu_{\ell}  (G,k)$ is decomposed as
\begin{align*}
&  \mu_{\ell}  (G,k) = \frac{1}{(G-\ell)^2}  \l(  \sum_{s,t = k-G+1}^{k-\ell} \hst + \sum_{s,t = k+1}^{k + G -\ell} \hst  - 2 \sum_{s=k-G+1}^{k-\ell} \sum_{t = k+1}^{k + G -\ell} \hst  \r)   \\[1ex]
 & = \frac{1}{(G-\ell)^2} \l(  \sum_{s,t = k-G+1}^{k-\ell} \hst +    \sum_{s,t = k+1}^{\theta_{j} - \ell} \hst + 2 \sum_{s = \theta_{j}+1}^{k+G-\ell}  \sum_{t = k+1}^{\theta_{j} - \ell} \hst + \sum_{s,t = \thetajl +1}^{k+G-l} \hst + \sum_{s, t \in  \mc A_\ell^{(1)} (G,k)} \hst \r) \\
 & \ - \frac{2}{(G-\ell)^2} \l(  \sum_{s=k+1}^{\thetajl - \ell} \sum_{t=k-G+1}^{k-\ell} \hst + \sum_{s = k-G+1}^{k-\ell} \sum_{t=\thetajl +1}^{k+G-\ell} \hst + 2 \sum_{s,t \in  \mc A_\ell^{(2)} (G,k)} \hst \r).
\end{align*}
Let $h^{(j)} = \mathbb{E} ( h (  Y^{(j)}_1 , \wt Y^{(j)}_1 ) )$ and analogously for $h^{(j-1)}$, and let $h^{(j-1,j)} = \mathbb{E} ( h ( Y^{(j-1)}_1 , \wt Y^{(j)}_1 )  )$. Noting that $\vert\mc A^{(1)}_\ell (G,k) \vert = O (2 \ell G)$ and $\vert\mc A^{(2)}_\ell (G,k) \vert = O (\ell G)$, we have that terms involving $\mc A^{(1)}_\ell (G,k)$ and $\mc A^{(2)}_\ell (G,k)$ are of order $O (G^{-1})$, since $\wt{h}_{s,t}$ is bounded by Assumption~\ref{assum:kernel}. Then, rearranging terms and collecting remainder terms all of order $O (G^{-1})$, we have that
\begin{align}
 \mu_{\ell}  (G,k) & = h^{(j-1)} +\frac{(\thetajl - \ell - k)^2}{(G-\ell)^2}  h^{(j-1)}  + \frac{2 (G-\ell + k - \thetajl)(\thetajl - \ell - k)}{(G-\ell)^2}  h^{(j-1,j)} \nonumber \\[1ex]
 & \ \ + \frac{(G-\ell + k - \thetajl)^2}{(G-\ell)^2}  h^{(j)}  -  \frac{2 (G-\ell)(\thetajl - \ell - k)}{(G-\ell)^2}  h^{(j-1)} \nonumber \\[1ex]
 & \ \ - \frac{2 (G-\ell)(G - \ell + k - \thetajl )}{(G-\ell)^2}  h^{(j-1,j)} + O(G^{-1}) \nonumber \\[1ex]
 & =  \frac{(G-\ell)^2  +(\thetajl -k - \ell)^2 - 2 (G-\ell)(\thetajl - k - \ell)  }{(G-\ell)^2} h^{(j-1)} + \frac{(G-\ell + k - \thetajl)^2}{(G-\ell)^2}  h^{(j)}   \nonumber  \\[1ex]
 & \ \ -2 \frac{(G-\ell)(G-\ell +k - \thetajl) - (G-\ell +k -\thetajl)(\thetajl - k - \ell)}{(G-\ell)^2}  h^{(j-1,j)} + O(G^{-1}) \nonumber \\[1ex]
 & =  \frac{(G-\ell +k - \thetajl)^2 - \ell^2 + 2\ell (G + k - \thetajl)}{(G-\ell)^2} h^{(j-1)} + \frac{(G-\ell + k - \thetajl)^2}{(G-\ell)^2}  h^{(j)} \nonumber \\[1ex]
 & \ \ -2 \frac{(G-\ell +k - \thetajl)^2 -\ell (G-\ell +k - \thetajl)}{(G-\ell)^2}  h^{(j-1,j)}  + O(G^{-1}) \nonumber \\[1ex]
 & =  \l( \frac{G - \ell + k - \theta_{j}}{G-\ell} \r)^2 \left( h^{(j-1)} + h^{(j)} - 2 h^{(j-1,j)} \right) + O (G^{-1}) \nonumber \\[1ex]
 &=  \label{expectation-left-change} \l( \frac{G - \ell + k - \theta_{j}}{G-\ell} \r)^2 d^{(j)}_\ell +  O(G^{-1}).
 \end{align}
A similar argument can be used in the case where $k \in \{ \theta_j +1 , \ldots , \theta_j + G - \ell -1 \}$, to yield
\begin{align}\label{expectation-right-change}
 \mu_{\ell}  (G,k) =  \l( \frac{G - \ell - k + \theta_{j}}{G-\ell} \r)^2 d^{(j)}_\ell +  O(G^{-1}),  
\end{align}
Combining Equations~\eqref{eq:bias},~\eqref{expectation-left-change} and~\eqref{expectation-right-change}, we get that for any $ |k - \theta_{j} \vert <G - \ell$, 
\begin{align*}
 \l\vert \mathbb{E} \{ \tstatkl \}   -   \l( \frac{G - \ell -  |k - \theta_{j}\vert}{G-\ell} \r)^2  d^{(j)}_\ell \r\vert = O(G^{-1/2}).   
\end{align*}
 {\cblue Lastly, when $k \in \{  \theta_{j} - G +1, \ldots , \theta_{j} - G + \ell \}$ or $k \in \{  \theta_{j} + G -\ell, \ldots , \theta_{j} - G -1\}$, analogous calculations show that $\mu_{\ell} (G,k) = {O}(G^{-1})$}, so that $\mathbb{E}\{ \tstatkl \} = O(G^{-1/2})$ for $\vert k - \theta_{j} \vert > G - \ell$, yielding the desired result.
\end{proof}

\begin{lemma}[\cite{xu2024change}, Theorem 4]\label{lemma:xu} 
Let $\{Z_t\}_{t \in \mathbb{Z}}$ be an $\mathbb{R}$-valued, mean-zero, possibly nonstationary process which admits the form $Z_t = g_t ( \mc F_t )$, where $g_t$ is a measurable function and $\mc F_t = \sigma ( \vep_s, \, s \leq t )$ with i.i.d.\ random elements $\{\vep_s\}_{s \in \mathbb{Z}}$.  Let the cumulative functional dependence measure $\Delta_{m,2} (Z)$ be defined as in Assumption~\ref{assum:functional-dep}.  
Assume that there exist absolute constants $\gamma_1(Z), \gamma_3(Z), C_\text{F}, C_Z > 0$ such that 
\begin{align*}
    \sup_{m \geq 0}\exp(C_\text{F} m^{\gamma_1(Z)})\Delta_{m,2} (Z) \leq C_Z   
\end{align*}
and $\sup_{t \in \mathbb{Z}}\mathbb{P}(|Z_t| > x ) \leq \exp(1 - x ^{\gamma_3(Z)})$, for any $x > 0$.  
Also suppose that $\gamma(Z) = \{1/\gamma_1(Z) + 1/\gamma_3(Z)\}^{-1} < 1$. 
Then, there exist absolute constants $c_1, c_2 > 0$ such that for any $z \ge 1$ and integer $n \geq 3$,
\begin{align*}
\p\l( \frac{1}{\sqrt n} \l\vert \sum_{t = 1}^n Z_t \r\vert \geq z \r) \leq n\exp \l(- c_1 z^{\gamma(Z)} n^{\gamma(Z)/2}\r) + 2\exp \l( -c_2 z^2 \r).
\end{align*}
\end{lemma}

\begin{lemma}\label{lemma:unif-prob} 
Suppose that Assumptions~\ref{assum:functional-dep},~\ref{assum:kernel} and~\ref{assum:change-point}~\ref{assum:min-space} hold.
Then, for any $\ell \le L$ with some fixed $L < \infty$, we have as $n \to \infty$, for any $z \ge 1/\sqrt{G - \ell}$,
\begin{align*}
\p\l( \max_{G \leq k \leq n-G} \l\vert \tstatkl - \mathbb{E}\{ \tstatkl \} \r\vert \geq z  \r)
\le 6nG^2 \exp\l( - c_1^\prime z^\gamma G^\gamma \r) + 12 n G \exp\l( - c_2^\prime z^2 G \r),
\end{align*}
where $c_1^\prime$ and $c_2^\prime$ depend only on $C_F$, $C_X$, $\gamma_1$, $C_h$, $p$ and $L$. 
\end{lemma}
\begin{proof}
By symmetry and boundedness of the kernel $h$, with $|h| \leq \bar{h}$ for some $\bar{h}>0$, we can re-write the test statistic $\tstatkl$ as :
\begin{align}
&  (G-\ell)^2 \tstatkl = \sum_{s,t = k - G+1}^{k - \ell} h_{s,t} + \sum_{s,t = k+1}^{k +G - \ell} h_{s,t} -2 \sum_{s = k - G+1}^{k - \ell} \sum_{t = k +1}^{k +G - \ell} h_{s,t} 
 \nonumber \\
 & =  -\sum_{s,t=k-G+1}^{k+G-\ell}  h_{s,t}   + 2 \sum_{s,t = k - G+1}^{k-\ell} h_{s,t} + 2 \sum_{s,t = k+1}^{k+G-\ell} h_{s,t} + O \l( \ell  (2G -\ell) \bar{h} \r)  \nonumber \\
 & = - 2\sum_{s=0}^{2G-\ell -1} \sum_{t = k-G+1}^{k+G - \ell -s}  h_{t,t+s}  + 4 \sum_{s = 0}^{G - \ell -1} \sum_{t=k-G+1}^{k-\ell-s} h_{t,t+s} 
+ 4 \sum_{s = 0}^{G-\ell-1} \sum_{t=k+1}^{k+G-\ell -s} h_{t,t+s} + O (G) 
\nonumber 
\\
&{ \cblue =  - 2\sum_{s=0}^{2G-\ell -3} \sum_{t = k-G+1}^{k+G - \ell -s}  h_{t,t+s}  + 4 \sum_{s = 0}^{G - \ell -3} \sum_{t=k-G+1}^{k-\ell-s} h_{t,t+s} 
+ 4 \sum_{s = 0}^{G-\ell-3} \sum_{t=k+1}^{k+G-\ell -s} h_{t,t+s} + O (G) } \nonumber \\
& =: 2 H_{1, k} + 4 H_{2, k} + 4 H_{3, k} + O(G). \label{t-stat-rewritten} 
\end{align}
Letting $v(s)= 2G - \ell - s$, we have
\begin{align*}
& P_{1, k} = \p\l( \vert H_{1, k} - \mathbb{E}(H_{1, k}) \vert > \frac{(2G-\ell)^{2} z}{10} \r) 
\\ 
& \le \sum_{s={ 0}}^{2G - \ell - 3} \underbrace{\p \l(  \l\vert \sum_{t=k-G+1}^{k-G+v(s)} \l\{ h_{t,t+s} - \mathbb{E}(h_{t,t+s}) \r\} \r\vert > \frac{(2G-\ell) z}{10} \r)}_{P_{1, k, s}}.
\end{align*}
For the sequence $\{ h_{t,t+s} \}_{{t \in \mathbb{Z}}}$, we consider its cumulative functional dependence measure (see the definition in~\eqref{eq:func:dep}):
\begin{align*}
\delta_{u, \nu}(h_{\cdot, \cdot + s}) &= \sup_{t \in \mathbb{Z}} \l\Vert h_{t, t + s} - h_{t, t + s, \{t - u\}} \r\Vert_\nu = \sup_{t \in \mathbb{Z}} \l\Vert h (Y_t, Y_{t+s} ) - h (Y_{t, \{t  - u \}}, Y_{t+s, \{ t  - u \}} ) \r\Vert_\nu 
\\
&= \sup_{t \in \mathbb{Z}} \l\Vert h_0 (Y_t- Y_{t+s} ) - h_0 (Y_{t, \{t-u \}}- Y_{t+s, \{ t-u\}} )  \r\Vert_\nu 
\\
&\le C_h \sup_{t \in \mathbb{Z}} \l( \l\Vert \Vert Y_t- Y_{t, \{ t - u\} } \Vert \r\Vert_\nu + \l\Vert \Vert Y_{t + s}- Y_{t+s, \{ t-u\}} \Vert \r\Vert_\nu \r)
\\
&\le C_h (2pC_\nu)^{1/\nu} \l( \max_{1 \le i \le p} \delta_{u, \nu, i} + \max_{1 \le i \le p} \delta_{u + s, \nu, i}  \r)
\end{align*}
where $C_\nu = 1$ if $0 < \nu \leq 1$ and $C_\nu = 2^{\nu-1}$ if $\nu >1$. 
The second equality and the first inequality follow from Assumption~\ref{assum:kernel} and Minkowski's inequality, and the last inequality follows from Cr inequality.
Then,
\begin{align*}
\Delta_{m, \nu}(h_{\cdot, \cdot + s}) &= \sum_{u = m}^\infty \delta_{u, \nu}(h_{\cdot, \cdot + s})
\le 2 C_h (2p C_\nu )^{1/\nu} \Delta_{m, \nu},
\end{align*}
such that under Assumption~\ref{assum:functional-dep}, with $C_F$, $C_X$ and $\gamma_1$ defined therein, 
\begin{align*}
\sup_{m \ge 0} \exp(C_F m^{\gamma_1}) \Delta_{m, 2}(h_{\cdot, \cdot + s})  \le 
4 C_h \sqrt{p} C_X.
\end{align*}
This allows us to apply Theorem~4 of \cite{xu2024change} (stated as Lemma~\ref{lemma:xu} here) and obtain
\begin{align*}
P_{1, k, s} \leq v (s) \exp \l( - c_1^\prime z^\gamma G^{\gamma} \r) + 2 \exp \l(-c_2^\prime v(s)^{-1} z^2 G^2 \r)
\end{align*}
with $\gamma = 2\gamma_1/(2+\gamma_1) < 1$, from the boundedness of the kernel assumed in Assumption~\ref{assum:kernel}.
where $c_k^\prime, \, k = 1, 2$, depend only on $C_F$, $C_X$, $\gamma_1$, $C_h$, $p$ and $L$.
Then, we have 
\begin{align*}
P_{1, k}  &\leq \sum_{s = { 0}}^{2G - \ell - 3}  v (s) \exp \l( - c_1^\prime  z^\gamma G^{\gamma} \r)    +   2\sum_{s= {0}}^{2G - \ell - 3} \exp \l( - 2 c_2^\prime v(s)^{-1} z^2 G^2 \r) \\
& \leq 2 G^2 \exp \l( - c_1^\prime  z^\gamma G^{\gamma} \r)  + 4G \exp \l( - c_2^\prime  z^2 G \r) ,
\end{align*}
where the second inequality follows since $\sum_{s=0}^{2G - \ell - 3} v(s) < 2G^2$ and $v(s) = 2G - \ell - s$ is decreasing in $s$.
Then by Bonferroni correction,
\begin{align*}
\p \l( \max_{G \leq k \leq n-G} \l\vert H_{1, k} - \mathbb{E}(H_{1, k}) \r\vert > \frac{(2G-\ell)^2 z}{10} \r) 
\leq  2n G^2 \exp \l( - c_1^\prime  z^\gamma G^{\gamma} \r)  +  4nG \exp \l( - c_2^\prime  z^2 G \r),
\end{align*}
and we can similarly bound the deviations of $H_{2, k}$ and $H_{3, k}$ over $k$.
Hence, the conclusion follows.
\end{proof}

\begin{lemma}\label{lemma:eta-crit} 
Suppose that the assumptions of Theorem~\ref{thm:consistency} hold. 
\begin{enumerate}[label = (\roman*)]
\item \label{lemma:eta-crit-a}  For any $\eta \in (0,1)$, we have $\p ( \mc S_{\ell,n} (j) ) \to 1$ for any $j \in \mc I_\ell$ and $\p ( \mc S_{\ell,n} ) \to 1$, where
\begin{align*}
\mc S_{\ell,n} (j) = \left\{ T_\ell(G, \theta_{j} ) \geq \max \l(  \max_{k: \, |k - \theta_{j}\vert > (1- \eta) (G-\ell) } \tstatkl , \zeta_\ell (n,G) \r)  \right\} 
\end{align*}
and $\mc S_{\ell, n}  = \bigcap_{j \in \mc I_\ell} \mc S_{\ell, n} (j)$.

\item  \label{lemma:eta-crit-b} 
For any $\eta \in (0,1)$, we have $\p (  \wt{\mc S}_{\ell,n} (j) ) \to 1$ for any $j \in \mc I_\ell$ and $\p (  \wt{\mc S}_{\ell,n} ) \to 1$, where
\begin{align*}
\begin{split}
\wt{\mc S}_{\ell,n} (j) = & \bigcap_{0 \leq r \leq \lceil \frac{2}{\eta} \rceil -2} \left[ \left\{ T_\ell\l(G, \theta_{j} + \l\lfloor \frac{r \eta (G-\ell)}{2} \r\rfloor \r) \geq \max_{k: \, \frac{(r+1) \eta (G-\ell)}{2} \le k - \theta_{j} \le \frac{(r+2) \eta (G-\ell)}{2} } \tstatkl  \right\}   \right. \\
& \left. \bigcap \left\{ T_\ell\l(G, \theta_{j} - \l\lfloor \frac{r \eta (G-\ell)}{2} \r\rfloor \r) \geq \max_{k: \, \frac{(r+1) \eta (G-\ell)}{2} \le \theta_j - k \le \frac{(r+2) \eta (G-\ell)}{2} } \tstatkl  \right\} \right],
\end{split}
\end{align*}
and $\wt{\mc S}_{\ell, n}  = \bigcap_{j \in \mc I_\ell} \wt{\mc S}_{\ell, n} (j)$.
\end{enumerate}
 
\end{lemma}
\begin{proof}
Recall the definition of $\mc E_{\ell, n}$ in Equation~\eqref{eq:Op-bound}.
Conditional on $\mc E_{\ell, n}$, at change point $\theta_{j}, \, j \in \mc I_\ell$, we have
\begin{align*}
 T_\ell (G, \theta_{j} ) = d^{(j)}_\ell  + O \l( \sqrt{\frac{\log(n)}{G}} \r).
\end{align*}
Then by Assumption~\ref{assum:change-point}~\ref{assum:change-size}, $(\log(n))^{-1/2} \sqrt{G} T_\ell (G, \theta_{j} ) \to \infty$, so that $T_\ell(G,\theta_{j} ) > \zeta_\ell (n, G)$. 
Also for any $k$ such that  $|k - \theta_{j} \vert > (1-\eta)(G-\ell)$,
\begin{align*}
& \max_{k: \, \vert k - {\theta}_{j} \vert > (1 -\eta) (G- \ell) } T_\ell (G, k)  =  \max_{k: \, \vert k - {\theta}_{j} \vert > (1 -\eta)( G - \ell)} \frac{(G - \ell - \vert k - {\theta}_{j}  \vert )^2}{(G-\ell)^2} d^{(j)}_\ell+ O \l( \sqrt{\frac{\log(n)}{G}} \r) \\
\leq  & \, \frac{ \{ G - \ell - (1-\eta) (G-\ell) \}^2}{(G-\ell)^2} d^{(j)}_\ell+  O \l( \sqrt{\frac{\log(n)}{G}} \r) = \eta^2 d^{(j)}_\ell+  O \l( \sqrt{\frac{\log(n)}{G}} \r) 
\end{align*}
conditional on $\mc E_{\ell, n}$, so that the assertion for $ \mc S_{\ell,n} (j)$ follows. 
Since the above arguments hold for any $j \in \mc I_\ell$ conditional on $\mc E_{\ell, n}$, the assertion for $\mc S_{\ell, n}$ also holds.
Analogous arguments apply to the proof of~\ref{lemma:eta-crit-b} and are omitted.
\end{proof}

\end{document}